\documentclass[a4paper,onecolumn,10pt,accepted=2022-08-22]{quantumarticle}
\pdfoutput=1
\usepackage{authblk}
\usepackage{amsmath}
\usepackage{amsthm}
\usepackage{amsfonts,dsfont}
\usepackage{mathtools}
\usepackage{subcaption}
\usepackage{soul} % for strikethorugh
\usepackage{float}
\usepackage{setspace}
\usepackage{cite}

\usepackage{qcircuit}
\usepackage{hyperref}
\usepackage{algorithm}
\usepackage{algpseudocode}
\algrenewtext{For}[3]{\algorithmicfor\ $#1 \gets #2$ \algorithmicto\ $#3$ \algorithmicdo}
\DeclareCaptionFormat{myformat}{#3}
\usepackage[left=1in,right=1in,top=1in,bottom=1in]{geometry}

\usepackage{xcolor}

\newtheorem{theorem}{Theorem}[section]
\newtheorem{lemma}[theorem]{Lemma}

\newtheorem{definition}[theorem]{Definition}

\newtheorem{corollary}[theorem]{Corollary}

\DeclarePairedDelimiterX\braket[2]{\langle}{\rangle}{#1 \delimsize\vert #2}

\makeatletter
\DeclareFontFamily{OMX}{MnSymbolE}{}
\DeclareSymbolFont{MnLargeSymbols}{OMX}{MnSymbolE}{m}{n}
\SetSymbolFont{MnLargeSymbols}{bold}{OMX}{MnSymbolE}{b}{n}
\DeclareFontShape{OMX}{MnSymbolE}{m}{n}{
    <-6>  MnSymbolE5
   <6-7>  MnSymbolE6
   <7-8>  MnSymbolE7
   <8-9>  MnSymbolE8
   <9-10> MnSymbolE9
  <10-12> MnSymbolE10
  <12->   MnSymbolE12
}{}
\DeclareFontShape{OMX}{MnSymbolE}{b}{n}{
    <-6>  MnSymbolE-Bold5
   <6-7>  MnSymbolE-Bold6
   <7-8>  MnSymbolE-Bold7
   <8-9>  MnSymbolE-Bold8
   <9-10> MnSymbolE-Bold9
  <10-12> MnSymbolE-Bold10
  <12->   MnSymbolE-Bold12
}{}

\let\llangle\@undefined
\let\rrangle\@undefined

\DeclareMathDelimiter{\llangle}{\mathopen}%
                     {MnLargeSymbols}{'164}{MnLargeSymbols}{'164}
\DeclareMathDelimiter{\rrangle}{\mathclose}%
                     {MnLargeSymbols}{'171}{MnLargeSymbols}{'171}
                     
\makeatother

\def\cH{\mathcal{H}}

\def\cO{\mathcal{O}}

\def\cR{\mathcal{R}}
\def\cS{\mathcal{S}}
\def\cT{\mathcal{T}}
\def\cU{\mathcal{U}}
\def\cV{\mathcal{V}}
\def\cW{\mathcal{W}}

\def\cZ{\mathcal{Z}}

\def\sy{{\rm s}}

\def\Re{{\rm Re}}
\def\Im{{\rm Im}}

\def\tr{{\rm tr}}

\def\one{{\mathchoice {\rm 1\mskip-4mu l} {\rm 1\mskip-4mu l} {\rm 1\mskip-4.5mu l} {\rm 1\mskip-5mu l}}}

\newcommand{\ketbra}[1]{\, | #1 \rangle \hspace{-0.05em} \langle #1 | \,} 
 
\newcommand{\ket}[1]{\,| #1 \rangle \,}  
\newcommand{\bra}[1]{\,\langle #1|\,}

\begin{document}

\title{Quantum algorithms from fluctuation theorems: Thermal-state preparation}

\author[1]{Zoe Holmes}
%\orcid{0000-0001-6841-4507}
%\email{zholmes@lanl.gov}

\author[2]{Gopikrishnan Muraleedharan}
%\orcid{0000-0002-1454-0353}
%\email{gmuraleedharan@lanl.gov}

\author[2]{Rolando D. Somma}
%\orcid{0000-0003-4335-2607}
%\email{somma@lanl.gov}

\author[1]{Yi\u{g}it Suba\c{s}\i}
%\orcid{0000-0003-1167-6527}
%\email{ysubasi@lanl.gov}

\author[2]{Burak \c{S}ahino\u{g}lu}
%\orcid{0000-0003-4488-2300}
%\email{buraksahinoglu@gmail.com}

\affil[1]{\normalsize {Computer, Computational, and Statistical Sciences Division, Los Alamos National Laboratory, Los Alamos, NM 87545, USA}}
\affil[2]{ \normalsize{Theoretical Division, Los Alamos National Laboratory, Los Alamos, NM 87545, USA}}

%\date{}

\maketitle

%%%%%%%%%%%%%%%%%%%%%%%%%%%%%%%%%%%%%%%%%%%%%%

\begin{abstract}
 {\normalsize{Fluctuation theorems provide a correspondence between properties of quantum systems in thermal equilibrium and a work distribution arising in a non-equilibrium process that connects two quantum systems with Hamiltonians $H_0$ and $H_1=H_0+V$.
 Building upon these theorems, 
 we present a quantum algorithm to prepare a purification of the thermal state of $H_1$ at inverse temperature $\beta \ge 0$ starting from a purification of the thermal state of $H_0$ at the same temperature. The complexity 
 of the quantum algorithm, given by the number of uses of certain unitaries, is $\tilde \cO(e^{\beta  (\Delta \! A- w_l)/2})$, where $\Delta \! A$ is the free-energy difference between the two quantum systems
 and $w_l$
 is a work cutoff that
 depends on the properties of the work distribution and the approximation error $\epsilon>0$.
 If the non-equilibrium process is trivial, this complexity is exponential in $\beta \|V\|$,
 where $\|V\|$ is the spectral norm of $V$.
 This represents a significant improvement over  prior
 quantum algorithms that have complexity exponential in $\beta \|H_1\|$ in the regime where  $\|V\|\ll \|H_1\|$.
 The quantum algorithm is then expected to be advantageous in a setting where an efficient quantum circuit is available for preparing the purification of the thermal state of $H_0$ but not for preparing the thermal state of $H_1$. This can occur, for example, when $H_0$ is an integrable quantum system and $V$ introduces interactions
 such that $H_1$ is non-integrable.
 The dependence of the complexity in $\epsilon$, when all other parameters are fixed, varies according to the structure of the quantum systems. 
 It can be exponential in $1/\epsilon$ in general, but we show it to be sublinear in $1/\epsilon$ if $H_0$ and $H_1$ commute, or polynomial in $1/\epsilon$
 if $H_0$ and $H_1$ are local spin systems. 
 In addition, the possibility of 
 applying a unitary that drives the system out of equilibrium allows one to increase the value of $w_l$ and improve the complexity even further. 
To this end, we analyze the complexity for preparing the thermal state of the transverse field Ising model using different non-equilibrium unitary processes and see significant complexity improvements.
}}
\end{abstract}

\newpage

\tableofcontents

%%%%%%%%%%%%%%%%%%%%%%%%%%%%%%%%%%%%%%%%%%%%%%
%%%%%%%%%%%%%%%%%%%%%%%%%%%%%%%%%%%%%%%%%%%%%%
\section{Introduction}

%%%%%%%%%%%%%%%%%%%%%%%%%%%%%%%%%%%%%%%%%%%%%%
\subsection{Background and outline}

The simulation of physical systems at  thermal equilibrium 
is an important problem in many fields of research, including
physics, quantum information, chemistry, biology, and discrete optimization (cf.~\cite{MRR+53,LL51,Suz87,LB97,TDV00,SBB08,TOVPV11,CP07,SKT+05,KC14,KGV83,Lov95}).
One way to approach this problem is via Gibbs sampling, that is, by using a computer to sample states or configurations according to the Gibbs or thermal distribution. Indeed, two of the most important algorithms in statistical physics,  namely Markov chain Monte Carlo (MCMC)~\cite{MRR+53,NB98} and quantum Monte Carlo (QMC)~\cite{Suz87,Nu99},  produce such configurations.
These methods are powerful and used widely, but can also suffer from severe limitations, which include the production of correlated outcomes and the well-known sign problem that can lead
to prohibitive runtimes~\cite{NB98,LGS+90,TW05}. Developing novel, possibly more efficient methods for studying physical systems in thermal equilibrium is an ongoing quest.

The simulation of quantum systems is expected to be one of the main applications of quantum computing.
Accordingly, in recent years we have seen 
a handful of quantum algorithms for studying thermal equilibrium. These algorithms are described as sequences
of unitary operations or quantum gates, with a guarantee that the final (reduced) state of the quantum computer encodes, or is a good approximation of, the thermal state of a quantum system of interest~\cite{LB97,SBB08,PW09,CW10,BB10,KB16,CS16,KKB20}. In some cases, the complexity of the 
quantum algorithm is provably better than that of its corresponding classical analogue, 
thereby providing a type of quantum speedup~\cite{Sze04,SBB08,BS17,AGG17}. 
However, in contrast to 
 known quantum algorithms for simulating the {\em dynamics} of quantum systems~\cite{Llo96,SOGKL02,SOKG03,BAC07,WBH+10,CW12,BCC+15,LC17},
existing quantum algorithms for thermal states are not very sophisticated and still present various limitations. For example, the quantum algorithm in Ref.~\cite{TOVPV11}
can lead to very large relaxation times if the choice of updates is poor; a problem related
to critical slowing down in Monte Carlo simulations~\cite{Wol90}. 
The quantum algorithms
in Refs.~\cite{PW09,CS16} transform the maximally entangled state into a coherent version (purification)
of the thermal state. This is related to a thermalization process that brings an infinite-temperature state
to one at  finite temperature and, again, can be highly inefficient. 
% In addition, even if a (polynomial) quantum speedup can be demonstrated for worst-case instances~\cite{PW09,CS16}, the empirical performance of these quantum algorithms can be substantially worse~\cite{SBC+20}, raising concerns about their practicality~\cite{BMN+21}.
Moreover, except for a few known cases~\cite{BB10,KB16,KKB20}, the runtimes of these quantum algorithms are often exponential in the strength (spectral norm) of the Hamiltonian that models the quantum system. While this is reasonable for worst-case instances from complexity-theory arguments~\cite{KSV02},
improved runtimes should be expected under further considerations, such as when the initial state is not ``far'' from the target thermal state, or when physical properties of the systems can be exploited.  The algorithms mentioned above are not capable of taking advantage of such considerations, which is a main motivation for developing the algorithm presented here.

In this paper, we propose a novel quantum algorithm for preparing thermal states that avoids some limitations of prior algorithms, as we discuss below. Our quantum algorithm
is inspired by fluctuation theorems, which play an important role in statistical mechanics~\cite{Jarzynski1997Equilibrium,Jarzynski1997Nonequilibrium,Jarzynski2011,Crooks1999Entropy,Crooks2000Path}. These theorems provide useful computational tools for studying physical systems in thermal equilibrium from non-equilibrium properties. 
They introduce the notion of ``work'' $w \in \mathbb R$, which is a random quantity that can be measured  in a non-equilibrium process 
that drives an initial thermal equilibrium state out of equilibrium by parametrically changing the system Hamiltonian. The distribution $P(w)$ contains information that relates to thermal equilibrium properties of $H_0$ and $H_1$.
For instance, by simulating a non-equilibrium process and sampling work values one can determine the free-energy difference between two thermal equilibrium states~\cite{Roncaglia2014,bassman2021computing}.

Our quantum algorithm does not directly simulate the non-equilibrium process to sample from $P(w)$, but applies a sequence of operations which, in effect, reproduces it. 
The produced quantum state can be shown to be a good approximation to the thermal state of $H_1$ at given inverse temperature $\beta \ge 0$. 
Our algorithm assumes access to a unitary that prepares the purification of the thermal state of $H_0$, which might be implemented efficiently for integrable systems.
The runtime of our quantum algorithm depends on several factors, including the magnitude of the perturbation $V$ such that $H_1=H_0+V$, the approximation error, and the characteristics of the non-equilibrium process. In particular, we show
that the runtime is bounded by an exponential in $\beta \|V\|$, where $\|V\|$ is the spectral norm of $V$, when the non-equilibrium process is trivial. This is a significant improvement
over prior quantum algorithms for thermal states~\cite{PW09,CW10,CS16}, especially when $\|V\| \ll \|H_1\|$; see Sec.~\ref{sec:relatedwork} for a detailed comparison. The runtime of our quantum algorithm can be further improved with a choice for the non-equilibrium process that  depends
on the physical properties of the quantum systems. 
Our results are then another example where a
powerful technique from statistical physics
can be adapted to provide sophisticated quantum simulation algorithms with various improvements.

\vspace{0.2cm}

The outline of our paper is as follows. In Sec.~\ref{sec:TSPP} we describe the thermal-state preparation problem, and in Sec.~\ref{sec:MainResults}, we present the main complexity results of our quantum algorithm.
In particular, we discuss how this complexity depends on  properties of the non-equilibirum process and the Hamiltonians, such as when $H_0$ and $H_1$ commute, or when $H_0$ and $H_1$ are local spin Hamiltonians. In Sec.~\ref{sec:relatedwork} we compare  our results with those of the related works. 

In Sec.~\ref{sec:fluctuationtheorems}, we set the stage
for our quantum algorithm, where we first discuss the {\em two-time measurement scheme}  that provides a simple connection between non-equilibrium and thermal-equilibrium properties in quantum systems. We discuss how this scheme could be adapted to construct a quantum algorithm for  preparing thermal states, but this approach presents other shortcomings, as we discuss in Sec.~\ref{sec:shortcomings}. Then, in Sec.~\ref{sec:two-copy}, we provide  another realization of the scheme, which we name the two-copy measurement scheme, that avoids those shortcomings and forms the basis of our quantum algorithm
described in Sec.~\ref{sec:qalg}. 

Our quantum algorithm requires implementing a good approximation of an exponential operator, which is discussed in Sec.~\ref{sec:exponentialoperator}. The approximation is based on a Fourier series (Sec.~\ref{sec:Fourier}). The steps of the algorithm are given in Sec.~\ref{sec:algorithm} and 
the proof of our main result on complexity is in Sec.~\ref{sec:algorithmcomplexity}. The complexity is exponential in $\beta (\Delta \! A-w_l)/2$, where $\Delta \! A$ is the free-energy difference 
between the quantum systems $H_1$ and $H_0$ at inverse temperature $\beta$, and $w_l \in \mathbb R$
is what we call a work cutoff. The work cutoff is then analyzed in Sec.~\ref{sec:applications}, where we provide several bounds for $w_l$ for general Hamiltonians (Sec.~\ref{sec:generalHamiltonians}), commuting Hamiltonians (Sec.~\ref{sec:commutingHamiltonians}), and local spin Hamiltonians (Sec.~\ref{sec:localHamiltonians}). These bounds also determine the complexity of the quantum algorithm in terms of its approximation error. The work cutoff, and hence the complexity of our algorithm, can be further improved by first applying a unitary that drives the system out of equilibrium. We explore this numerically for the case of preparing the thermal state of the transverse-field Ising model and see significant complexity improvements even for small system sizes ($\mathfrak n \le 11$).
Furthermore, in Appendix~\ref{app:optimal_unitary} we derive the form of an optimal unitary for which the work cutoff is maximized.

In Sec.~\ref{sec:conclusions}, we conclude with a summary of the results and also present some open problems.
These regard further improvements of our algorithm by using other approximations to the exponential (e.g., a Chebyshev approximation), by considering other non-equilibrium processes, or by other uses of fluctuation theorems to compute thermal properties.  To ease readability, the technical proofs of some complexity statements are given in respective appendices.

 %%%%%%%%%%%%%%%%%%%%%%%%%%%%%%%%%%%%%%%%%%%%%%%%
\subsection{Problem statement}
\label{sec:TSPP}

We describe the problem for $\mathfrak{n}$-qubit systems for simplicity but our results apply to all finite-dimensional quantum systems, including systems of qudits and fermionic systems, given access to the relevant unitaries. We write $\cH_{\sy}=\mathbb C^N$ for the relevant Hilbert space of dimension $N=2^{\mathfrak{n}}$. The thermal-state preparation problem (TSPP) is defined as follows:

\begin{definition}[TSPP]\label{def:TSPP}
Given a Hamiltonian $H_1$ acting on $\cH_{\sy}$, an inverse temperature $\beta \ge 0$, and an approximation error $\epsilon \ge 0$, the goal is to prepare a mixed quantum state $\tau_1$ that satisfies
\begin{align}
\label{eq:TSPP}
  \frac{1}{2}   \| \tau_1 - \rho_1  \|_{1} \le \epsilon \;,
\end{align}
where
\begin{align}
    \rho_1 : = \frac{e^{-\beta H_1}}{\cZ_1}
\end{align}
is the thermal or Gibbs state and $\cZ_1:=\tr(e^{-\beta H_1})$ is
the partition function.
\end{definition}

For any linear operator $A$ acting on $\cH_{\sy}$, $\|A\|_1:= \tr(\sqrt{A^\dagger A})$ is the trace norm, $\tr(.)$ is the trace, and the left hand side of Eq.~\eqref{eq:TSPP} is the trace distance between $\tau_1$ and $\rho_1$.
A state $\tau_1$ that satisfies Eq.~\eqref{eq:TSPP} cannot be distinguished from $\rho_1$ with probability greater than $\epsilon$ in a single-shot experiment~\cite{Bar09,NC01}. Moreover, expectations computed in   $\tau_1$ differ from those in  $\rho_1$ by an amount that is, at most, proportional to $\epsilon$.

%%%%%%%%%%%%%%%%%%%%%%%%%%%%%%%%%%%%%%%%%%%%%%
%%%%%%%%%%%%%%%%%%%%%%%%%%%%%%%%%%%%%%%%%%%%%%
\subsection{Main results}
\label{sec:MainResults}
In Sec.~\ref{sec:fluctuationtheorems} we describe how fluctuation theorems determine
thermal properties of a quantum system through
properties of a work distribution $P(w) \ge 0$. Here, $w \in \mathbb R$ is a random variable that represents the work done on a system which is initially in thermal equilibrium with respect to $H_0$ but then is driven out of equilibrium by a unitary $\cU$.
For closed systems, $w$ is simply equal to the change of the system's energy, which can be determined by performing projective measurements of the Hamiltonian of the system at the start ($H_0$) and end ($H_1$) of the process. The probability of observing a change in energy corresponding to $w$ determines $P(w)$, which we formally define in Eq.~\eqref{eq:TTdist}. We assume $w \in [w_{\min},w_{\max}]$, with $-\infty < w_{\min}\le w_{\max}<\infty$. The lower tail of $P(w)$ gives rise to the notion of a work cutoff, $w_l$, 
where those samples satisfying $w<w_l$ do not contribute significantly to the thermal properties of the system and can be discarded.

Building upon these theorems, our main result is a quantum algorithm for preparing the thermal state $\rho_1$ starting from the thermal state $\rho_0= e^{-\beta H_0}/\cZ_0$, given in a purified form $\ket{\Psi_0}$, where $\cZ_0=\tr(e^{-\beta H_0})$ is the partition function of $H_0$. The purification is such that $\ket{\Psi_0} \in \cH_{\sy} \otimes \cH_{\sy'}$,
where $\cH_{\sy} \equiv \cH_{\sy'}$, 
encodes $\rho_0$ in $\cH_{\sy}$, which
is obtained by discarding or tracing out $\cH_{\sy'}$.
The main figure of interest is the complexity of the quantum algorithm and its asymptotic scaling in terms of the inverse temperature $\beta$ and the  error $\epsilon$.
Certain results of this section show these dependencies explicitly for clarity,
where  complexities are stated using the big-${\cal O}$ notation.
The following theorem is our main result, and is proven in Sec.~\ref{sec:algorithmcomplexity}.

%%%%%%%%%
\begin{theorem}[Quantum complexity for the thermal-state preparation problem]
\label{thm:main}
Let $H_0$ and $H_1$ be two Hamiltonians acting on $\cH_{\sy}$, $\beta \ge 0$ be the inverse temperature, and $\epsilon \ge 0$ be the approximation error.
Let $\ket{\Psi_0} \in \cH_{\sy} \otimes \cH_{\sy'}$ denote a particular purification of the thermal state $\rho_0$ 
where $\cH_{\sy} \equiv \cH_{\sy'}$,
$U_0$ be the unitary that performs the map $\ket 0 \rightarrow \ket{\Psi_0}=U_0 \ket 0$, and $\cU$ be a given unitary acting on $\cH_{\sy}$.
Then, there exists a quantum algorithm that solves the TSPP and prepares a quantum state $\tau_1$ that satisfies Eq.~\eqref{eq:TSPP}, that is,
\begin{align}
\frac{1}{2} \| \tau_1 - \rho_1 \|_1 \le \epsilon \;.
\end{align}
The quantum algorithm makes $Q= \cO (e^{\sqrt{\ln (1/\epsilon)}}e^{\frac{\beta}{2} (\Delta \!A(\beta)- w_l(\beta,\epsilon))})$ amplitude amplification rounds, on average, where $\Delta \! A(\beta):=-\frac{1}{\beta} \ln(\cZ_1(\beta)/\cZ_0(\beta))$ is the free-energy difference and $w_l(\beta,\epsilon)\ge w_{\min}$ is a work cutoff that satisfies
\begin{align}
\label{eq:Wlconditionmain}
    \sum_{w<w_l(\beta,\epsilon)} P(w) e^{-\beta \, (w-\Delta \! A(\beta))} 
      \le \left(\frac{\epsilon}{6} \right)^2 \;.
\end{align}
Each amplitude amplification round uses the unitaries $\cU$, $\cU^\dagger$, $U_0$, and $U_0^\dagger$ once.
It also makes $\cO((\ln(1/\epsilon)+\beta (w_{\max}-w_l(\beta,\epsilon)))^{3/2})$ uses of a controlled-$U$ unitary and its inverse, acting on $\cH_{\sy}\otimes \cH_{\sy'}$, where $U:=e^{i \delta \beta W/2}$, $W:=   H_1 \otimes \one - \one \otimes   H_0^*$,
and $\delta \le \pi/16$.
The Hamiltonian $H_0^*$ is the complex conjugate of $H_0$ in the computational basis and $\one$ is the identity.
Additionally, the number of two-qubit gates per amplitude amplification round is
$\cO((\ln(1/\epsilon)+\beta (w_{\max}-w_l(\beta,\epsilon)))^{3/2})$.
\end{theorem}

In Sec.~\ref{sec:fluctuationtheorems} we will see that $\sum_w P(w)e^{-\beta w}=e^{-\beta \, \Delta \! A(\beta)}$ which is known as the Jarzynski equality~\cite{Jarzynski1997Nonequilibrium}. Then,
the results in Thm.~\ref{thm:main} could be  alternatively stated in terms of properties of the work distribution alone. We stress that this work distribution depends not only on the initial and final Hamiltonians $H_0$ and $H_1$, and the inverse temperature $\beta$, but also on the unitary $\cU$. Hence the overall performance of the quantum algorithm depends on the choice of $\cU$. 

Our quantum algorithm prepares a pure state $\ket{\Psi_1'} \in \cH_{\sy} \otimes\cH_{\sy'}$, a purification of $\tau_1$, and $\tau_1$ is obtained by discarding or tracing out $\cH_{\sy'}$.
Working with purifications is advantageous for several reasons, including the efficient computation of thermal properties, where quantum-metrology techniques can be applied~\cite{KOS07}.
The quantum algorithm implements certain unitaries $\cU$, $U_0$, $U$, their inverses, and the controlled versions thereof. We work in a model where the implementation details of these unitaries are not essential, so the main complexity results presented in Thm.~\ref{thm:main} can be regarded as the {\em quantum query complexity} of the quantum algorithm.
One could obtain the full gate complexity
from explicit constructions for those unitaries, but a detailed analysis of these constructions, which also depend on the presentation of the Hamiltonians, is outside the scope of this paper. Nevertheless, for completeness,  in Appendix~\ref{app:Ucomplexity} we give one possible construction for $U$ using quantum signal processing and qubitization, when the Hamiltonians are presented as linear combination of unitaries~\cite{LC17,LC19}. Additionally, finding $U_0$ is
a difficult problem in general.
Our quantum algorithm is then expected to be mostly useful
in instances where $U_0$ can be accessed or efficiently
constructed, and we provide an example of this in Sec.~\ref{sec:localHamiltoniansUne1}, where $H_0$ is a non-interacting spin system.  

The quantum algorithm in Thm.~\ref{thm:main} is a unitary operation that acts on a Hilbert space ${\cal H}=\cH_{\rm anc}\otimes \cH_{\sy} \otimes \cH_{\sy'}$, where $\cH_{\rm anc}$  denotes an ancillary space. This ancillary space consists of 
 $\mathfrak m$ ancilla qubits that are needed to implement a particular linear combination of $U^j$'s, for $j \in \mathbb{Z}$, and the $\mathfrak m'$ ancilla qubits that are needed to implement each $U^j$, which corresponds to time-evolution with $W$. If the linear combination is implemented using quantum signal processing~\cite{LC17},  then $\mathfrak m  \le 2$. If the Hamiltonians $H_0$ and $H_1$ are presented as linear combinations of $L$ unitaries and also quantum signal processing is used
 to simulate each $U^j$, then $\mathfrak m'=\cO(\log(L))$, as 
 explained in Appendix~\ref{app:Ucomplexity}. 
 Additional details on the ancilla requirements are provided in Secs.\ref{sec:LCU} and~\ref{sec:QSP}, and Appendix~\ref{app:Ucomplexity}.

\vspace{0.1cm}

It is well-known that very small probability events, so-called rare events, are crucial for calculations based on fluctuation theorems~\cite{Jarzynski2006}. The same holds for thermal state preparation as discussed in this paper. 
The work cutoff $w_l(\beta,\epsilon)$, formally defined implicitly by Eq.~\eqref{eq:Wlconditionmain}, is such that those work values below $w_l(\beta,\epsilon)$ can be ignored within the error tolerance even accounting for the subtleties associated with rare events.
In each amplitude amplification round,
our quantum algorithm implements
an operation that is a function of a ``work operator'' $W$ given in Thm.~\ref{thm:main}. 
This operation modifies the amplitudes corresponding to the eigenstates
of $W$,
but those amplitudes associated with eigenvalues of $W$ below $w_l(\beta,\epsilon)$ are irrelevant: these can be modified significantly incorrectly or left unchanged, 
while still producing a sufficiently good approximation of $\rho_1$. 

This work cutoff also depends on $\cU$, as different $\cU$'s give rise to different work distributions. Since the dominant term of the complexity is expected to be $e^{\frac{\beta}{2} (\Delta \! A(\beta)- w_l(\beta,\epsilon))}$ in many interesting instances, to reduce this complexity, Thm.~\ref{thm:main} suggests a trade off between implementing a $\cU$ that can increase $w_l(\beta,\epsilon)$ and the complexity of implementing such a $\cU$.
(Note that, in general, $e^{\frac{\beta}{2} (\Delta \! A(\beta)- w_l(\beta,\epsilon))} > .98$\footnote{Jarzynski equality and Eq.~\eqref{eq:Wlconditionmain} imply $e^{-\beta \Delta \! A(\beta)}=\sum_{w<w_l(\beta,\epsilon)} P(w) e^{-\beta w}+\sum_{w\ge w_l(\beta,\epsilon)} P(w) e^{-\beta w} \le (\epsilon/6)^2 e^{-\beta \Delta  \! A(\beta)} +e^{-\beta w_l(\beta,\epsilon)}$. For $\epsilon \le 1$, we obtain $e^{\frac{\beta}{2} (\Delta \! A(\beta)- w_l(\beta,\epsilon))} \ge \sqrt{35/36}\approx .986$.}.)
In particular, small values of $w_l(\beta,\epsilon)$ can occur for a trivial $\cU$, i.e., $\cal{U}=\one$, if the thermal state $\rho_0$ has sufficiently high support on the subspaces of large and small eigenvalues of $H_0$ and $H_1$, respectively; a simple example being $H_1=-H_0$. This suggests that
a good choice for $\cU$ to increase work values and the work cutoff will be a transformation that maps eigenstates of $H_0$ to those of $H_1$
in some way. Indeed, in Appendix~\ref{app:optimal_unitary}, Thm.~\ref{thm:Uopt_finite_eps}, we show that for any $\epsilon \ge 0$, there exists a non-equilibrium unitary $\cU_{opt}^\epsilon$
that maps eigenstates of $H_0$ to those of $H_1$ such that $w_l$ is maximized in Eq.~\eqref{eq:Wlconditionmain}. When $\epsilon=0$, Cor.~\ref{cor:Uopt_eps0}
states that $\cU_{opt}^0$ also preserves the (ascending) order of the eigenvalues.

Finding or implementing $\cU_\text{opt}^\epsilon$
might be intractable in general. The same holds for finding the largest cutoff $w_l$ for a given unitary $\cU$. 
However,  $\cU_\text{opt}^\epsilon$ can be efficiently constructed in some instances, e.g., when there is an efficient quantum circuit that transforms eigenstates of $H_0$ to those of $H_1$ or when $H_0$ and $H_1$ are both integrable.
In addition, Thm.~\ref{thm:Uopt_finite_eps} suggests
the existence of other $\cU$'s that could be simpler to implement and help improve the complexity of our algorithm by increasing $w_l(\beta,\epsilon)$. In Sec.~\ref{sec:localHamiltoniansUne1}, we numerically investigate the complexity improvements of various choices in  $\cU$ for the case of preparing the thermal state of the transverse-field Ising model.

\vspace{0.2cm}

Theorem~\ref{thm:main} is useful when $H_1=H_0+V$, where $V$ is a perturbation that satisfies $\|V\| \ll \|H_1\|$.
This case is important in physics, as we could use our quantum algorithm to prepare the thermal state
of a complex quantum system $H_1$ starting from (a purification of) that of a simpler quantum system $H_0$.
To this end, we analyze the work cutoff for many classes of Hamiltonians and find corresponding lower bounds that, in conjunction with Thm.~\ref{thm:main} and Thm.~\ref{thm:ratioPF} below, allow us to determine the complexity
of the  quantum algorithm for these examples. While the cutoff can depend on $\beta$ and $\epsilon$, our 
results provide suitable choices for $w_l$ that do not depend on $\beta$.
The following result applies to general Hamiltonians and is proven in Sec.~\ref{sec:generalHamiltonians}.

\begin{theorem}[Work cutoff for general Hamiltonians]
\label{thm:w_lgeneral}
Let $H_0$ and $H_1=H_0+V$ be two Hamiltonians acting on $\cH_{\sy}$, $\epsilon > 0$ be the approximation error, and $\cU$ be a given  unitary acting on $\cH_{\sy}$. Then, 
for all $w_l$ satisfying
\begin{align}
\label{eq:w_lgeneral}
    w_l \le -\frac {6 \|V_\cU\|}{\epsilon} \; ,
%w_l =-\|V\| \cO( 1/\epsilon)
\end{align}
where $V_\cU:=H_1 - \cU H_0 \cU^\dagger$,
we obtain
\begin{align}
    \sum_{w<w_l} P(w) e^{-\beta \, (w-\Delta \! A(\beta))} & \le  \left( \frac {\epsilon} {6} \right)^2  \;,
\end{align}
and Eq.~\eqref{eq:Wlconditionmain} is satisfied.  
\end{theorem}

Note that $V_\cU=V$ for $\cU=\one$, where the non-equilibrium process is trivial. When the Hamiltonians commute, i.e., $[H_0, H_1]=0$, we obtain 
a simpler bound for the work cutoff that is independent of $\epsilon$. The following result is proven in 
 Sec.~\ref{sec:commutingHamiltonians}.

\begin{theorem}[Work cutoff for commuting Hamiltonians]
\label{thm:w_lcommuting}
Let $H_0$ and $H_1=H_0+V$ be two Hamiltonians acting on $\cH_{\sy}$ that satisfy $[H_0,H_1]=0$, $\epsilon \ge 0$ be the approximation error, and $\cU=\one$. 
Then, for all $w_l$ satisfying
\begin{align}
    w_l \le -\|V\| \;,
\end{align}
we obtain
\begin{align}
\label{eq:Wlconditioncommuting}
    \sum_{w<w_l} P(w) e^{-\beta \, ( w-\Delta \! A(\beta))} & =0 \;,
\end{align}
and Eq.~\eqref{eq:Wlconditionmain} is satisfied.
\end{theorem}

The previous examples could be interpreted as limiting cases,
for which a suitable choice of $w_l(\beta,\epsilon)$ in Thm.~\ref{thm:main} can scale with $1/\epsilon$ or does not depend on $\epsilon$ at all. Nevertheless, many other cases 
are expected to lie in between, with a milder dependence of $w_l(\beta,\epsilon)$ on $\epsilon$ than $1/\epsilon$. To this end,
we analyze the case where $H_0$, $H_1$, and $V$ are local spin Hamiltonians, and
in Sec.~\ref{sec:localHamiltonians} we show:

\begin{theorem}[Work cutoff for local spin Hamiltonians]
\label{thm:w_llocal}
Let $H_1=H_0+V$, where $H_0=\sum_{X \in \Lambda} h_{0,X}$ and $V=\sum_{X \in \Lambda} v_X$ are $k$-local Hamiltonians of degree $g$ in a lattice $\Lambda$, as in Def.~\ref{def:LocalHamiltonian}, $\epsilon > 0$ be the approximation error, and $\cU=\one$. The number of local terms in $V$ is at most $M$, and these local terms satisfy $\|h_{0,X}\| \le h$ and $\|v_X\| \le v$ for all $X \in \Lambda$. Then, 
for all $w_l$ satisfying
\begin{align}
\label{eq:w_llocal}
    w_l \leq -2Mv - 2 h g k \ln (6/\epsilon) \;,
\end{align}
we obtain
\begin{align}
    \sum_{w<w_l} P(w) e^{-\beta \, (w-\Delta \! A(\beta))} \le \left( \frac{ \epsilon} {6} \right)^2   \;,
\end{align}
and Eq.~\eqref{eq:Wlconditionmain} is satisfied.
\end{theorem}

Often, $Mv$ can be bounded using $\|V\|$, the spectral norm of $V$, as $Mv \le \xi \|V\|$ for some constant $\xi \ge 0$ that is independent of $M$.
In this case, the bound given in Eq.~\eqref{eq:w_llocal} is linear in $\|V\|$. 
The result in Thm.~\ref{thm:w_llocal} can also be generalized to the case $\cU=\cU_1^\dagger \cU_0 \ne \one$ if $V_\cU=\cU_1 H_1 \cU_1^\dagger- \cU_0 H_0 \cU_0^\dagger$ is also a $k$-local Hamiltonian of degree $g$. 
Then, $V_\cU=\sum_{X \in \Lambda} v'_X$ is a sum of $M'$ local terms satisfying $\|v'_X\| \le v'$ for all $X \in \Lambda$, and one needs to replace $M \rightarrow M'$ and $v \rightarrow v'$
in Eq.~\eqref{eq:w_llocal}, as explained in Sec.~\ref{sec:localHamiltoniansUne1}.

Equation~\eqref{eq:w_llocal} is an exponential improvement in terms
of $1/\epsilon$ over the general case given in Eq.~\eqref{eq:w_lgeneral}.
This is due to the locality of the Hamiltonians, which implies a work distribution $P(w)$ that decays exponentially outside a region of work values that depends on $Mv$.
To show this exponential decay, in Lemma~\ref{lem:EigenspaceOverlap} of Sec.~\ref{sec:localHamiltoniansU=1} we give a result for local Hamiltonians that might be of independent interest. Roughly, this lemma states that 
the support of an eigenspace of $H_0$ of sufficiently low energies (eigenvalues) on an eigenspace of $H_1$ of sufficiently large energies decays exponentially 
with the energy difference minus $2Mv$.
The proof of Lemma~\ref{lem:EigenspaceOverlap} is in Appendix~\ref{app:EigenspaceOverlap}.

\vspace{0.2cm}

In addition to $e^{-\beta w_l(\beta,\epsilon)/2}$,  another relevant factor in the complexity of our  algorithm is $e^{\beta \, \Delta \! A(\beta)/2}$. 
In Appendix~\ref{app:ratioPF} we show that the absolute value of the free-energy difference can be bounded in terms of $\|V\|$ in general:
\begin{theorem}[Free-energy difference]
\label{thm:ratioPF}
    Let $H_0$ and $H_1=H_0+V$ be two Hamiltonians, and $\Delta \! A(\beta) = -\frac{1}{\beta} \ln(\cZ_1(\beta)/\cZ_0(\beta))$ be the free-energy difference at inverse temperature $\beta \ge 0$.
    Then,
    \begin{align}
      \label{eq:ratioPF}
        |\Delta \! A(\beta)| \le \|V\| \; .
    \end{align}
\end{theorem}

\vspace{0.3cm}

Theorem~\ref{thm:main}, together with Thms.~\ref{thm:w_lgeneral},~\ref{thm:w_lcommuting},~\ref{thm:w_llocal}, and~\ref{thm:ratioPF}, allow us to determine the complexity of our quantum  algorithm for the previous three examples. In particular, if we consider a trivial non-equilibrium process and assuming $Mv \le \xi \|V\|$
for local Hamiltonians, these theorems
 imply:

\begin{corollary}[Number of amplitude amplification rounds for $\cU=\one$]
\label{cor:main}
Let $\cU=\one$. Then, the average number of amplitude amplification rounds of our quantum  algorithm is
$Q=\tilde \cO(e^{(\frac{3}{\epsilon} + \frac 1 2 )\beta \|V\|})$
for general Hamiltonians, $Q=\tilde \cO(e^{\beta \|V\|})$
for commuting Hamiltonians, and $Q=\tilde \cO ((6/\epsilon)^{\beta h g k} \; e^{(\xi+\frac 1 2)\beta \|V\|})$
for local Hamiltonians. 
The $\tilde \cO$ notation hides the term $e^{\sqrt{\ln(1/\epsilon)}}$ that
is subpolynomial in $1/\epsilon$.
\end{corollary}

Corollary~\ref{cor:main} becomes particularly useful in the  regime $\beta \|V\| =\cO( 1)$, where the complexity of our quantum algorithm can be very mild. Additionally, the dependence of $Q$ on $1/\epsilon$, when all other parameters are fixed, varies strongly depending on the case: It is exponential in $1/\epsilon$ in general, subpolynomial in $1/\epsilon$ for commuting Hamiltonians, and polynomial in $1/\epsilon$ for local Hamiltonians. 

\vspace{0.1cm}

Last, we note that if a unitary $\cU^*$ is implemented at the end of the quantum algorithm, where $\cU^*$ is the complex conjugate of $\cU$ in the computational basis, 
the prepared pure state is taken to a form that is suitable
for running the quantum algorithm again. This could allow one to prepare a sequence of purifications of thermal states corresponding to Hamiltonians $H_0,H_1,H_2,\ldots$, by repeated uses of our quantum algorithm; see Sec.~\ref{sec:two-copy} for details.

%%%%%%%%%%%%%%%%%%%%%%%%%%%%%%%%%%%%%%%%%%%%%%%%%%%%
%%%%%%%%%%%%%%%%%%%%%%%%%%%%%%%%%%%%%%%%%%%%%%%%%%%%
\subsection{Related work}
\label{sec:relatedwork}

There are several results on quantum algorithms for simulating physical systems in thermal equilibrium. We summarize some and detail their similarities and differences with our contributions. 

References~\cite{PW09,CS16,tong2021fast} provide quantum algorithms for preparing thermal states of quantum systems by  applying the exponential operator $e^{-\beta H_1/2}$ to the maximally-entangled initial state, which is $\ket{\Psi_0}$ for $H_0=0$ in our case. This exponential operator can be implemented by means of  quantum phase estimation~\cite{Kit95,CEMM98}, as in Ref.~\cite{PW09}, or by using more recent techniques such as the linear combination of unitaries, as in Ref.~\cite{CS16}, or other operators (e.g., $1/(z_j-H_1)$ for $z_j \in \mathbb C$), as in Ref.~\cite{tong2021fast}. The dominant term in the complexity of those algorithms scales as $\sqrt{N/\cZ_1(\beta)}$ if $H_1 \ge 0$, and this is $\cO(e^{\beta \|H_1\|/2})$. The problem we analyze here is similar in that we also aim at implementing an exponential operator, but the complexity of our algorithm greatly improves that of Refs.~\cite{PW09,CS16,tong2021fast} when $\|V\| \ll \|H_1\|$
due to our assumptions. One assumption is that the initial state $\ket{\Psi_0}$ is the purification of the thermal state of some quantum system $H_0 \ne 0$,
and this helps when $H_1$ is a perturbation of $H_0$. For example,
$H_0$ could correspond to a non-interacting, or integrable quantum system and $V$
is a perturbation that introduces interactions, or makes the system non-integrable.
Another assumption is the ability to implement a non-equilibirum unitary $\cU$ that can improve the complexity further. In fact, the problem 
analyzed in Refs.~\cite{PW09,CS16,tong2021fast} can be thought of an instance
of the more general problem analyzed in this work, one where $H_0=0$, $H_1 \ge 0$, and $\cU=\one$. For this instance, the complexity of our quantum algorithm approximately matches that of Refs.~\cite{PW09,CS16}.

Reference~\cite{Roncaglia2014} provides a quantum algorithm to sample from the work distribution $P(w)$ that, in combination with fluctuation theorems, can be used to estimate the free-energy difference $\Delta \! A$ or the ratio of partition functions $\cZ_1(\beta)/\cZ_0(\beta)$. To this end, it uses a generalized measurement implemented via a version of the quantum phase estimation algorithm~\cite{Kit95}, where an ancillary register encodes 
the work value resulting from the difference of energies of $H_1$
and $H_0$. The analysis in Ref.~\cite{Roncaglia2014} does not consider the problem of preparing thermal states; however, it could be easily adapted to the latter by measuring the ancilla and accepting the outcome with probability proportional to $e^{-\beta w}$ (see Sec.~\ref{sec:two-time}).
 While this approach works, it has significantly worse convergence than our quantum algorithm, as we discuss in Sec.~\ref{sec:shortcomings}. 
 One reason is that quantum phase estimation performs poorly in terms of the target approximation error $\epsilon'$, having complexity polynomial in $1/\epsilon'$. 
 Another reason is that this approach prepares a mixed state, and techniques such as amplitude amplification do not apply~\cite{BHMT02} . Our quantum algorithm does not suffer from these shortcomings, because it prepares a purification of the thermal state and uses techniques such as the linear-combination-of-unitaries subroutine or quantum signal processing to avoid quantum phase estimation.

Reference~\cite{ozols2013quantum} studies a quantum version
of the classical rejection sampling algorithm. That algorithm prepares a superposition state that encodes a target probability distribution in its amplitudes. It uses an ancilla qubit that functions as a ``quantum coin'', determining whether a configuration is accepted or rejected. Using amplitude amplification, that algorithm can prepare a target distribution from an initial one, and provide a polynomial quantum speedup under some assumptions. If the probability distributions correspond to thermal states of classical systems, then this algorithm would serve as a method for thermal-state preparation. However, it does not prepare thermal states of quantum systems and thus our results apply more generally.
Indeed, our algorithm could also be interpreted as a quantum version of rejection sampling, where we transform initial quantum states into quantum states with certain target amplitudes through the implementation of an exponential operator. The ancillary system required to implement this exponential operator can be thought of as a ``quantum coin''. If we replace the exponential by some other operator, our quantum algorithm would also allow for the preparation of other target quantum states that are not necessarily thermal states of physical Hamiltonians.  

In the limit $\beta \gg 1$, our quantum algorithm can be used to prepare
the ground state of $H_1$, within some approximation error, but the complexity is exponential in $\beta \|V\|$. Other known quantum algorithms for preparing ground states, including those in Refs.~\cite{Poulin2009,BKS10,Ge2019,Lin2022}, have complexities bounded by factors that depend on the inverse of the spectral gap $\Delta_1>0$ of $H_1$ and the inverse of the overlap $\gamma>0$ between the ground states of $H_0$ and $H_1$. These results are not directly comparable; both situations
$e^{\beta \|V\|} \gg \|H_1\|/(\Delta_1 \gamma)$ and $e^{\beta \|V\|} \ll \|H_1\|/(\Delta_1 \gamma)$
could arise depending on the case. Additionally, if the perturbation is such that $\beta \|V\| \gg 1$, then $\gamma$ is expected to decrease exponentially fast in $\|V\|$ or the size of the system, in which case the complexity bounds for all these methods are exponential.

Reference~\cite{TOVPV11} provides a quantum algorithm, known as quantum Metropolis sampling (QMS), to prepare thermal states of quantum systems via a generalization of the classical Metropolis sampling algorithm, as  used in Markov chain Monte Carlo simulations. QMS is based on a random walk that is applied to the state of the system; that is, it constructs a superoperator that has the thermal state as the unique fixed point. The random walk applies a random unitary, and the move is accepted or rejected based on the Metropolis rule. 
 More recently, in Ref.~\cite{chen2021fast}, QMS was 
shown 
to have fast convergence to the thermal state
if the so-called eigenstate thermalization hypothesis (ETH) holds. In this case, the runtime is polynomial in $\beta \|H_1\|$, $\mathfrak n$, and other parameters given by the properties of the system  (i.e., $H_1$). Other quantum algorithms for thermal state preparation are also discussed in Ref.~\cite{shtanko2021algorithms} where, under the ETH, the runtimes can be similar to that of QMS. (Note that the ETH does not always hold, as is the case for some integrable quantum systems~\cite{rigol2008thermalization}.)
Reference~\cite{motta2020determining} provides another quantum algorithm for preparing thermal states based on quantum imaginary time evolution, 
which  requires applying the exponential operator $e^{-\beta H_1/2}$ to product states. The performance  of that algorithm then depends on the structure of $H_1$.
In particular, the exponential operator is approximated by exploiting the exponential decay of spatial correlations in the quantum states generated during the execution. 
Like in QMS, that algorithm also follows a Markov chain that has the thermal state as a fixed point, under some assumptions. 
% % A similar fast convergence result has been shown in Ref.~\cite{shtanko2021algorithms} using a ``digital bath" assuming ETH.}
% with physical, yet hard to verify or prove, assumptions such as eigenstate thermalization hypothesis (ETH). 
Nevertheless, those algorithms are not comparable to ours since, rather than applying a random walk or introducing a digital bath
%with random interactions 
%as in Refs.~\cite{TOVPV11,motta2020determining, chen2021fast,shtanko2021algorithms} 
that might need to be refreshed,
%and wait for convergence,
we prepare the thermal state through the implementation of an exponential operator acting on a pure state and other unitaries.
However, as in our case, those algorithms may also benefit from starting from an initial state that is easy to prepare and ``closer'' to the target thermal state $\rho_1$. 
This remains as an open problem.

Very recently a number of variational approaches have been proposed for preparing thermal states on near-term quantum computing hardware; see Refs.~\cite{ Sagastizabal2021Variational, Martyn2019Product, Verdon2019Quantum,chowdhury2020variational,Wang2021Variational, Foldager2022noise}. Similar to our approach, Ref.~\cite{Sagastizabal2021Variational} prepares a purification of the thermal state, i.e., a thermofield double state of a target Hamiltonian, via some easier to prepare thermofield double state of another Hamiltonian. However, the transformation to prepare the target   state in this algorithm, as well as those in Refs.~\cite{Martyn2019Product, Verdon2019Quantum, Foldager2022noise, chowdhury2020variational,Wang2021Variational}, is found by utilizing a variational representation of a `guess' thermal state and a cost function formulated in terms of free-energy. Given the variational nature of these algorithms, they lack rigorous complexity guarantees and, indeed, they may experience issues due to barren plateaus~\cite{mcclean2018barren,cerezo2020cost, holmes2021barren,holmes2021connecting,marrero2020entanglement} and local minima~\cite{Bittel2021Training}.

%%%%%%%%%%%%%%%%%%%%%%%%%%%%%%%%%%%%%%%%%%%%%%
%%%%%%%%%%%%%%%%%%%%%%%%%%%%%%%%%%%%%%%%%%%%%%
%%%%%%%%%%%%%%%%%%%%%%%%%%%%%%%%%%%%%%%%%%%%%%
%%%%%%%%%%%%%%%%%%%%%%%%%%%%%%%%%%%%%%%%%%%%%%

\section{Preliminaries}
\label{sec:fluctuationtheorems}

%%%%%%%%%%%%%%%%%%%%%%%%%%%%%%%%%%%%%%%%%%%%%%%%%%%%%

\subsection{Thermal-state preparation via the two-time measurement scheme}
\label{sec:two-time}

The goal is to prepare the thermal state of a target Hamiltonian $H_1$ given some form of access to the thermal state of an initial Hamiltonian $H_0$. In this section we provide an approach based on non-equilibrium fluctuation theorems. These theorems are manifestations of the fact that information about thermal equilibrium can be extracted from non-equilibrium processes. 
For example, Jarzysnki equality~\cite{Jarzynski1997Nonequilibrium} and Crooks's fluctuation theorem~\cite{Crooks1999Entropy} relate fluctuations of work performed during a non-equilibrium process to free-energy differences between thermal states. While fluctuation theorems have been derived for a variety of dynamics~\cite{Jarzynski2011,Campisi2011}, here we focus on unitarily evolving quantum systems; the unitary $\cU$ appearing in Thm.~\ref{thm:main}. In this setting, the unitary $\cU$ is viewed as implementing a non-equilibrium driving process because it drives a system that is initially thermal with respect to $H_0$ out of equilibrium.

Our starting point is the so-called \textit{two-time measurement} scheme~\cite{Tasaki2000Jarzynski, Kurchan2000Quantum, Talkner2007TasakiCrooks}, which is the most common setting in which fluctuation theorems hold. 
This scheme is based on the observation that the work done on a closed system is equal to the change of its energy, which can be determined by performing projective measurements of the Hamiltonians at the start and end of the non-equilibrium process. 
In more detail, we let $\ket{\phi_{0,m}}$, $m=0,1,\ldots$, be the eigenstates of $H_0$  with corresponding eigenvalues $\varepsilon_{0,m}$,
and $\ket{\phi_{1,n}}$, $n=0,1,\ldots$, the eigenstates of $H_1$ with  corresponding eigenvalues 
 $\varepsilon_{1,n}$.
For simplicity, we assume that $H_0$ and $H_1$ are non-degenerate, having distinct eigenvalues $\varepsilon_{0,0}<\varepsilon_{0,1}<\ldots$ and $\varepsilon_{1,0}<\varepsilon_{1,1}<\ldots$. 
 The initial state of 
 the system is the thermal state $\rho_0=e^{-\beta H_0}/\cZ_0$.  The two-time measurement scheme 
consists of three steps: 
i) a projective measurement of $H_0$  on $\rho_0$, which outputs an eigenvalue $\varepsilon_{0,m}$ with (Boltzmann) probability $P_0(\varepsilon_{0,m}) = e^{-\beta \varepsilon_{0,m}}/\cZ_0$ and leaves the system in the eigenstate 
$\ket{\phi_{0,m}}$, ii) an implementation of a unitary $\cU$, which corresponds to the non-equilibrium process, and iii) a projective measurement of $H_1$, 
which outputs an eigenvalue $\varepsilon_{1,n}$ with probability $P(\varepsilon_{1,n}|\varepsilon_{0,m})= | \!\bra{\phi_{1,n}}\cU \ket{\phi_{0,m}}\! |^2$
if the outcome of the first measurement was $\varepsilon_{0,m}$, and leaves the system in the eigenstate $\ket{\phi_{1,n}}$.

The work 
performed during this process is defined as $w:=\varepsilon_{1,n}-\varepsilon_{0,m}$. This is a random variable sampled according to 
a distribution
\begin{align}
\label{eq:TTdist}
    P(w) & = \sum_{m,n: \varepsilon_{1,n}- \varepsilon_{0,m}=w}  P(\varepsilon_{1,n}|\varepsilon_{0,m}) P_0(\varepsilon_{0,m})\; .
\end{align}
 Fluctuation theorems establish a connection between $P(w)$ and thermal quantities~\cite{Jarzynski1997Nonequilibrium,Crooks1999Entropy}. For example, Jarzynski equality
states 
\begin{align}
    \sum_w P(w) e^{-\beta w}
    &= \sum_{m,n} P(\varepsilon_{1,n}|\varepsilon_{0,m}) P_0(\varepsilon_{0,m}) e^{-\beta (\varepsilon_{1,n} -\varepsilon_{0,m})}  \\
    &=  \frac 1 {\cZ_0}\sum_{m,n} P(\varepsilon_{1,n}|\varepsilon_{0,m})  e^{-\beta \varepsilon_{1,n}}\\
    & = \frac{\cZ_1}{\cZ_0} \\
    & = e^{-\beta  \, \Delta \! A}\;,
\end{align}
where we used the property $\sum_{m}P(\varepsilon_{1,n}|\varepsilon_{0,m})=1$ for all $n$. Then, an estimate of $\cZ_1/\cZ_0$ or the free-energy difference $\Delta \! A$ can be obtained from the two-time measurement scheme by sampling $w$ and estimating the expectation of $e^{-\beta w}$.
This idea is investigated in Ref.~\cite{Roncaglia2014,bassman2021computing}.

The two-time measurement scheme
can be slightly modified to prepare the thermal state $\rho_1$ from $\rho_0$ as follows. 
After the projective measurement of $H_1$, the outcome
is accepted with probability  $q(w) \propto e^{-\beta w}$ or rejected with probability $1-q(w)$. Post-selecting on accepted outcomes only, the resulting state $\rho_f$ can be shown to be $\rho_1$, since:
\begin{align}
    \label{eq:reweighting}
    \rho_f &\propto \sum_{m,n} e^{-\beta (\varepsilon_{1,n}-\varepsilon_{0,m})} 
    P(\varepsilon_{1,n}|\varepsilon_{0,m}) P_0(\varepsilon_{0,m}) \ketbra{\phi_{1,n}}\\
     &= \frac{1}{\cZ_0} \sum_{m,n} e^{-\beta \varepsilon_{1,n}} P(\varepsilon_{1,n}|\varepsilon_{0,m}) \ketbra{\phi_{1,n}} \\
    &= \frac{1}{\cZ_0} \sum_{n} e^{-\beta \varepsilon_{1,n}} \ketbra{\phi_{1,n}} \\
    & \propto \rho_1 \;.
\end{align}
 
This re-weighting approach to thermal-state preparation
can be interpreted as a generalization of {\em rejection sampling}, a standard method used to generate a target probability distribution from a given initial one using an accept/reject routine. However, as an approach for solving the TSPP, it presents several shortcomings 
that we discuss below.

When the eigenvalues of $H_0$ (or $H_1$) are not distinct, a projective measurement of the Hamiltonian projects the state into a combination of states with definite eigenvalue. This degenerate case can be interpreted as a limiting instance of the non-degenerate one, and our prior statements on fluctuation theorems and thermal-state preparation using the two-time measurement scheme hold in general.

%%%%%%%%%%%%%%%%%%%%%%%%%%%%%%%%%%%%%%%%%%%%%%%%%%%%%%%%%%%%%
\subsubsection{Shortcomings of the two-time measurement scheme}
\label{sec:shortcomings}

The probability of accepting a given work value $w \in [w_{\min},w_{\max}]$ could be exponentially small, i.e., scaling as $e^{\beta w_{\min}}$ to satisfy $q(w) \le 1$. Note that $w_{\min}$ could be as small as $-\|H_0\|-\|H_1\|$, which occurs when, for example, the first measurement outputs the largest eigenvalue of $H_0$ and the second measurement outputs the lowest eigenvalue of $H_1$. This already gives a complexity that can scale as $e^{\beta (\|H_0\|+\|H_1\|)}$, even if $H_1$ is a small perturbation of $H_0$. Additionally, one may try to (quadratically) improve this complexity by means of amplitude amplification~\cite{BHMT02}. However, the two-time measurement scheme assumes access to the thermal state $\rho_0$, which is a mixed state, and amplitude amplification does not apply to this setting.

The two-time measurement scheme also requires perfect projective measurements of $H_0$ and $H_1$. On a quantum computer, these measurements are commonly implemented using quantum phase estimation~\cite{Kit95,CEMM98}, which is the  approach discussed in Ref.~\cite{Roncaglia2014}. However, quantum phase estimation can only implement approximate measurements, and its complexity depends polynomially on the inverse of the error. As explained, the acceptance probability following the two-time measurement scheme can be exponentially small, and the additional error made in each measurement has to be exponential in $ w_{\min}$ for guaranteed convergence. The resulting complexity of each use of quantum phase estimation is then exponentially large, being exponential in $-\beta w_{\min}$ or exponential in $\beta (\|H_0\|+\|H_1\|)$, in the worst case.

In order to avoid using quantum phase estimation, we could alternatively
use a different scheme, and attempt to prepare $\rho_1$ by acting directly on $\rho_0$ with two exponentials, since 
\begin{align}
\label{eq:expproduct}
  \rho_1 &\propto  e^{-\beta H_1} e^{\beta H_0} \rho_0 \\
  & = e^{-\beta H_1/2}   e^{\beta H_0/2} \rho_0 e^{\beta H_0/2} 
  e^{-\beta H_1/2} \;. 
\end{align}
Quantum algorithms to implement exponentials with polylogarithmic complexity on the inverse of the error, that do not use quantum phase estimation, exist~\cite{CS16,CSS18}. However, this idea suffers from an additional problem because the probability of success in implementing each exponential can be small, i.e., scaling as $e^{-\beta \|H_0\|}$ and $e^{-\beta \|H_1\|}$, respectively. 
That is, each exponential is implemented through a unitary operation, and thus the best we can hope for is to implement the normalized exponential operators $e^{\beta H_0/2}/e^{\beta \|H_0\|/2}$ and $e^{-\beta H_1/2}/e^{\beta \|H_1\|/2}$, to satisfy the norm condition. 
Unfortunately, this exponential scaling appears even if $w_{\min} \gg -\|H_0\| -\|H_1\|$ and $\beta \|V\| \ll 1$, which  can be prohibitively small. This is related to the fact that Eq.~\eqref{eq:expproduct} can be interpreted as a sequence of two steps, one that takes $\rho_0$ to $\one$, associated with the infinite-temperature state, and another that takes $\one$ back to $\rho_1$. This approach also results in unnecessary complexity overheads.

These issues show that the complexity of the two-time measurement scheme, or simple modifications thereof, for solving the TSPP can be  prohibitive.  Improved schemes are needed.

%%%%%%%%%%%%%%%%%%%%%%%%%%%%%%%%%%%%%%%%%%%%%
\subsection{Thermal-state preparation via the two-copy measurement scheme}
\label{sec:two-copy}

The prior shortcomings can be avoided as follows. First, in order to use amplitude amplification 
and (quadratically) improve complexity due to small acceptance probabilities, we consider a scheme
that assumes access to a pure state $\ket{\Psi_0} \in \cH_{\sy} \otimes \cH_{\sy'}$, which is a purification of $\rho_0$, and prepares the pure state 
$\ket{\Psi_1} \in  \cH_{\sy} \otimes \cH_{\sy'}$, which is a purification of $\rho_1$. The Hilbert spaces are $\cH_{\sy} \equiv \cH_{\sy'}\equiv \mathbb C^N$. 
The thermal state $\rho_1$ is then obtained by discarding or tracing out one of the subsystems, which will be $\cH_{\sy'}$ in this case.
Second, to fix the unnecessary overhead due to the sequential action of $e^{-\beta H_1/2}e^{\beta H_0/2}$, we use proper purifications for $\rho_0$ and $\rho_1$, where the product of exponentials in Eq.~\eqref{eq:expproduct}  can be effectively implemented by acting concurrently --rather than sequentially-- with two exponentials on the different subsystems, corresponding to $H_0$ and $H_1$. In Sec.~\ref{sec:qalg} we show that this approach works and avoids the prior complexity issues.

Our  \textit{two-copy measurement} scheme then works as follows. Rather than assuming access to $\rho_0$, we assume access to  the unitary $U_0$ that prepares the following state:
\begin{align}
\label{eq:Psi0}
\ket{\Psi_0} =\frac{1}{\sqrt{\cZ_0}}\sum_{m} e^{-\beta \varepsilon_{0,m}/2} \ket{\phi_{0,m}}\ket{\phi^*_{0,m}} \; ,
\end{align}
which satisfies $\rho_0=\tr_{\cH_{\sy'}}(\ketbra{\Psi_0})$. 
For a linear operator $X$ acting on $\cH_{\sy}\otimes \cH_{\sy'}$,
$\tr_{\cH_{\sy'}}(X)$ is the partial trace over $\cH_{\sy'}$.
The asterisk stands for complex conjugation:
if $\ket{\phi}=\sum_\sigma c_{\sigma}\ket \sigma$ in the computational basis $\{\ket \sigma\}_\sigma$, then $\ket{\phi^*}=\sum_\sigma c^*_{\sigma}\ket \sigma$.
We also let $H^*$  be a Hamiltonian that is the complex-conjugate of a Hamiltonian $H$ in the computational basis. These definitions imply, for example, $H^*_0 \ket{\phi^*_{0,m}} =\varepsilon_{0,m}\ket{\phi^*_{0,m}}$ and $H^*_1 \ket{\phi^*_{1,n}} =\varepsilon_{1,n}\ket{\phi^*_{1,n}}$.
The scheme consists of three steps:
i) a projective measurement of $H_0^*$ on the second copy of $\ket{\Psi_0}$, which outputs an eigenvalue $\varepsilon_{0,m}$ with (Boltzmann) probability $P_0(\varepsilon_{0,m})=e^{-\beta \varepsilon_{0,m}}/\cZ_0$ and leaves the systems in
$\ket{\phi_{0,m}}\ket{\phi^*_{0,m}}$, ii) an implementation of a unitary $\cU$ on the first copy, which corresponds to the non-equilibrium process, and iii) a projective measurement of $H_1$   on the first copy, which outputs an eigenvalue $\varepsilon_{1,n}$ with probability $P(\varepsilon_{1,n}|\varepsilon_{0,m})=|\!\bra{\phi_{1,n}} \cU \ket{\phi_{0,m}}\!|^2$ if the outcome of the first measurement was $\varepsilon_{0,m}$, and leaves the systems in $\ket{\phi_{1,n}}\ket{\phi^*_{0,m}}$.
If $w=\varepsilon_{1,n}-\varepsilon_{0,m}$,
then the two-copy measurement scheme outputs $w$
according to the same distribution $P(w)$ as Eq.~\eqref{eq:TTdist}.

This two-copy measurement scheme can be slightly modified to 
prepare the thermal state $\rho_1$ from $\ket{\Psi_0}$ as follows.
If we disregard the measurement in step i) and consider step ii), the initial state is transformed to
\begin{align}
\label{eq:two-copysecondstep}
     ( \cU \otimes \one ) \ket{\Psi_0} &=\frac{1}{\sqrt{\cZ_0}}\sum_{m,n}  e^{-\beta \varepsilon_{0,m}/2} \bra{\phi_{1,n}} \cU \ket{\phi_{0,m}}  \ket{\phi_{1,n}}\ket{\phi_{0,m}^*} \;.
\end{align}
Note that the state of $\cH_{\sy}$ obtained by tracing out $\cH_{\sy'}$ is the same as that prepared if the measurement in step i) was done but the outcome was not recorded. 
Let
\begin{equation}\label{eq:WorkOp}
    W :=  H_1 \otimes \one - \one \otimes H^*_0  \; 
\end{equation}
be a {\em work operator}
that has eigenstates $\ket{\psi_{m,n}}:=\ket{\phi_{1,n}}\ket{\phi_{0,m}^*}$
and eigenvalues $w_{m,n}:=\varepsilon_{1,n}-\varepsilon_{0,m}$; that is, $W  \ket{\psi_{m,n}}= w_{m,n} \ket{\psi_{m,n}}$.
Then, Eq.~\eqref{eq:two-copysecondstep} can be written as
\begin{align}
\label{eq:stateworkbasis}
   ( \cU \otimes \one ) \ket{\Psi_0} &=  \sum_{m,n}  \sqrt{p_{m,n}} e^{i\varphi_{m,n}} \ket{\psi_{m,n}} \; ,
\end{align}
where $p_{m,n}:=(e^{-\beta \varepsilon_{0,m}}/\cZ_0) |\bra{\phi_{1,n}} \cU \ket{\phi_{0,m}}|^2$. In particular,  $p_{m,n}=P(\varepsilon_{1,n}|\varepsilon_{0,m}) P_0(\varepsilon_{0,m})$ is the probability of measuring first $\varepsilon_{0,m}$ and then $\varepsilon_{1,n}$ in the two-time measurement scheme and, due to normalization, $\sum_{m,n} p_{m,n}=1$. The phases $\varphi_{m,n}$ are irrelevant to the argument that follows. Equation~\eqref{eq:TTdist} implies
\begin{align}
\label{eq:TCdist}
    P(w)=\sum_{m,n: \varepsilon_{1,n}-\varepsilon_{0,m}=w} p_{m,n} \;.
\end{align}
The action of $e^{-\beta W/2}$ on 
Eq.~\eqref{eq:stateworkbasis} transforms the state as
\begin{align}
    e^{-\beta W/2} ( \cU \otimes \one ) \ket{\Psi_0} & = \sum_{m,n} e^{-\beta w_{m,n}/2} \sqrt{p_{m,n}}e^{i\varphi_{m,n}} \ket{\psi_{m,n}}
    \\
    \label{eq:two-copythirdstep}
    &=
    \frac{1}{\sqrt{\cZ_0}}\sum_{m,n} e^{-\beta \varepsilon_{1,n}/2} \bra{\phi_{1,n}} \cU \ket{\phi_{0,m}}  \ket{\phi_{1,n}}\ket{\phi_{0,m}^*} \;.
\end{align}
This exponential operator has a similar effect as that of the re-weighting approach discussed in Sec.~\ref{sec:two-time}. 
In fact, the state of $\cH_{\sy}$ in Eq.~\eqref{eq:two-copythirdstep} is proportional to $\rho_1$:
\begin{align}
\nonumber
    \tr_{\cH_{\sy'}} ( e^{-\beta W/2}& ( \cU \otimes \one )  \ketbra{\Psi_0}( \cU^\dagger \otimes \one ) e^{-\beta W/2}) \\
    &\propto  \sum_{m,n,n'}e^{-\beta (\varepsilon_{1,n}+\varepsilon_{1,n'})/2}
    \bra{\phi_{1,n}} \cU \ket{\phi_{0,m}} \bra{\phi_{0,m}} \cU^\dagger \ket{\phi_{1,n'}}  \ket{\phi_{1,n}} \! \bra{\phi_{1,n'}} \\
    \label{eq:proptorho1}
    & \propto \rho_1 \;.
\end{align}
The exponential operator also has two properties that are useful. One is that $e^{-\beta W/2}=e^{-\beta H_1/2} \otimes e^{-\beta H^*_0/2}$,
and this can be implemented by acting concurrently  on $\cH_{\sy}$ and $\cH_{\sy'}$ with independent operators. Another is that
we can use a high-precision quantum algorithm, rather than quantum phase estimation, to implement  $e^{-\beta W/2}$;
see Sec.~\ref{sec:Fourier}.

A quantum algorithm to prepare $\rho_1$ from $\ket{\Psi_0}$
could follow the previous two steps. However, we find it clearer --and potentially advantageous--  to add a third step in which we implement $\cU^*$
on $\cH_{\sy'}$. This unitary, which is the complex conjugate of $\cU$, does not modify the state of $\cH_{\sy}$, which is still proportional to $\rho_1$. However, its action on Eq.~\eqref{eq:two-copythirdstep} implies
\begin{align}
\nonumber
( \one \otimes \cU^* ) & e^{-\beta W/2} ( \cU \otimes \one ) \ket{\Psi_0}  \\
& =\frac{1}{\sqrt{\cZ_0}} \sum_{m,n,n'} e^{-\beta \varepsilon_{1,n}/2} \bra{\phi_{1,n}} \cU \ket{\phi_{0,m}}\! \bra{\phi^*_{1,n'}} \cU^* \ket{\phi^*_{0,m}}  \ket{\phi_{1,n}} \ket{\phi_{1,n'}^*} \\
& = \frac{1}{\sqrt{\cZ_0}} \sum_{m,n,n'} e^{-\beta \varepsilon_{1,n}/2} \bra{\phi_{1,n}} \cU \ket{\phi_{0,m}}\! \bra{\phi_{0,m}} \cU^\dagger \ket{\phi_{1,n'}}  \ket{\phi_{1,n}} \ket{\phi_{1,n'}^*} \\
& = \frac{1}{\sqrt{\cZ_0}} \sum_{n} e^{-\beta \varepsilon_{1,n}/2}  \ket{\phi_{1,n}} \ket{\phi_{1,n}^*} \\
& = \sqrt{\frac{\cZ_1}{\cZ_0} }\ket{\Psi_1} \\
\label{eq:psi0psi1mapB}
& = e^{-\beta \, \Delta \! A/2} \ket{\Psi_1}\;,
\end{align}
where 
\begin{align}
\label{eq:Psi1}
    \ket{\Psi_1}:=\frac 1 {\sqrt{\cZ_1}}\sum_n e^{-\beta \varepsilon_{1,n}/2} \ket{\phi_{1,n}} \ket{\phi^*_{1,n}}
\end{align}
is the purification of $\rho_1$;  that is, $\rho_1=\tr_{\cH_{\sy'}}(\ketbra{\Psi_1})$.
The similitude between Eqs.~\eqref{eq:Psi0} and~\eqref{eq:Psi1} 
is clear. Preparing $\ket{\Psi_1}$ rather than Eq.~\eqref{eq:two-copythirdstep} can be useful for applying our thermal-state preparation algorithm repeatedly many times. For example, $\ket{\Psi_1}$ could serve as the initial state of the quantum algorithm to prepare the thermal state of another Hamiltonian $H_2$, and so on.     

We note that a two-copy measurement scheme (with some differences) has been previously used in \cite{solfanelli2021experimental} to sample $(\varepsilon_{0,m},\varepsilon_{1,n})$   with joint probabilities $p_{m,n}$. These samples are then used to compute quantities like $\Delta \! A$ to verify fluctuation relations such as Jarzynski equality. In contrast, here the two-copy measurement scheme is used for thermal state preparation.

%%%%%%%%%%%%%%%%%%%%%%%%%%%%%%%%%%%%%%%%%%%%%%%%%%%%%%%%%%%%%%%%%%%%%%%
%%%%%%%%%%%%%%%%%%%%%%%%%%%%%%%%%%%%%%%%%%%%%%%%%%%%%%%%%%%%%%%%%%%%%%%
\section{Quantum algorithm: Methods and complexity}
\label{sec:qalg}

%%%%%%%%%%%%%%%%%%%%%%%%%%%%%%%%%%%%%%%%%%%%%%%%%%%%%%%%%%%%%%%%%%%%%%%
\subsection{High-level description} 
\label{sec:highlevel}

We discuss the key aspects of our quantum algorithm,
which essentially follows the 
 steps discussed in Sec.~\ref{sec:two-copy}, namely the implementation of the non-equilibrium unitary $(\cU \otimes \one)$ on $\ket{\Psi_0}$, 
the action of the exponential $e^{-\beta W/2}$, and the implementation of $(\one \otimes \cU^*)$, in combination with amplitude amplification~\cite{BHMT02,KLM07}. The state $\ket{\Psi_0}$ is defined in Eq.~\eqref{eq:Psi0} and the operator $W$ is defined in Eq.~\eqref{eq:WorkOp}.
Ideally, we would like to implement the exponential operator through the action of a unitary that is of the form
\begin{align}
\label{eq:exactalgorithmunitary}
    \begin{pmatrix}  \frac 1  \alpha e^{-\beta W/2}  & . \cr . & . \end{pmatrix} \;,
\end{align}
where $\alpha >0$ is needed for normalization. 
This unitary, which is a block-encoding of $e^{-\beta W/2}$, in combination with the other operations implies [Eq.~\eqref{eq:psi0psi1mapB}]
\begin{align}
 \label{eq:exactalgorithmmap}
    \ket 0  \ket{\Psi_0} &\mapsto  \ket 0 \left(\frac 1 \alpha  (\one \otimes \cU^*) e^{-\beta W/2} (\cU \otimes \one)\ket{\Psi_0} \right)+  \ket{\chi^\perp} \\
    \label{eq:exactalgorithmmap2}
    & =  \ket 0 \left( \frac 1 \alpha e^{-\beta \, \Delta \! A/2} \ket{\Psi_1} \right)+  \ket{\chi^\perp} \;,
\end{align}
where $\ket{\chi^\perp}$ is a (subnormalized) quantum state
orthogonal to $\ket 0$ of the ancilla. The state $\ket{\Psi_1}$ is defined in Eq.~\eqref{eq:Psi1}.
Post-selecting on the state $\ket 0$ of the ancilla, this approach would produce $\ket{\Psi_1}$; that is, it would prepare $\rho_1$ after tracing out $\cH_{\sy'}$. 
Rather than measuring the ancilla, we can use amplitude amplification to produce $\ket{\Psi_1}$ with high probability and less complexity~\cite{BHMT02,KLM07}. 
The number of necessary amplitude amplification rounds is given by $\cO( (\alpha  e^{\beta \, \Delta \! A/2}))$, which is the inverse of the (absolute) amplitude of the desired state, i.e., the part of the state in Eq.~\eqref{eq:exactalgorithmmap2} where the ancilla is in $\ket 0$. 
Amplitude amplification also requires a similar number of uses of  $U_0$, the unitary needed to prepare $\ket{\Psi_0}=U_0 \ket 0$, and its inverse.

In practice, our quantum algorithm does not implement the unitary
of Eq.~\eqref{eq:exactalgorithmunitary}. It does implement, however, a different unitary that performs the map in Eq.~\eqref{eq:exactalgorithmmap} approximately when acting on $\ket 0 \ket{\Psi_0}$, but with a smaller value for the normalization constant $\alpha$.
This is to avoid undesired complexity overheads, implying
fewer rounds of amplitude amplification than if we were to act with the operation of Eq.~\eqref{eq:exactalgorithmunitary}, while still preparing a satisfactory state that approximates $\ket{\Psi_1}$ and satisfies Eq.~\eqref{eq:TSPP} after tracing out $\cH_{\sy'}$. That is, since Eq.~\eqref{eq:exactalgorithmunitary} is unitary and $\| e^{-\beta W/2}/\alpha\|\le 1$, $\alpha$ is at least $e^{-\beta w_{\min}/2}$ and $w_{\min}=-\|H_0\|-\|H_1\|$
in the worst case.
Then, $\alpha$ can be as large as $e^{\beta (\|H_0\|+\|H_1\|)/2}$, even if $H_1$ is a small perturbation of $H_0$.

To avoid this problem, we will show that it suffices to implement an approximation of $e^{-\beta W/2}$ 
on a particular subspace only, where the eigenvalues of $W$ satisfy
$w_{m,n} \ge w_l$, and $w_l \ge w_{\min}$ is the cutoff appearing in Thm.~\ref{thm:main}. 
This will imply that the normalization factor 
appearing in Eq.~\eqref{eq:exactalgorithmmap}, when using this approximation, is  $\alpha =\tilde  \cO(e^{-\beta w_l/2})$ rather than $e^{-\beta w_{\min}/2}$ (see Sec.~\ref{sec:approximations}).
Combining this result with those in Sec.~\ref{sec:applications}  that  show $w_l \gg w_{\min}$ for many instances, we obtain a significantly improved algorithm.

%%%%%%%%%%%%%%%%%%%%%%%%%%%%%%%%%%%%%%%%%%%%%%%%%%%%%%%%%%%%%%
\subsection{Access model}
\label{sec:accessmodels}

The query complexity of the main quantum algorithm, given in Thm.~\ref{thm:main}, is best formulated using a model
that considers access to unitaries $\cU$, $\cU^*$, $U_0$, and $U_0^\dagger$.
The unitary $\cU$ describes the non-equilibrium process 
appearing in the fluctuation theorems; see Sec.~\ref{sec:fluctuationtheorems}. The unitary $U_0$ is the state
preparation operation that performs the map $\ket 0 \rightarrow U_0\ket 0 =\ket{\Psi_0}$.
 Our algorithm also uses the controlled-$U$ and controlled-$U^\dagger$ unitaries, where $U$
 is a short-time evolution under $W$ defined in Eq.~\eqref{eq:WorkOp}, and implements the map
\begin{align}
\label{eq:Uaction}
    U \ket 0 \ket \phi = \ket 0 e^{ i \lambda W} \ket \phi \;,
\end{align}
for all $\ket \phi \in \cH_{\sy} \otimes \cH_{\sy'}$. The parameter $\lambda = \delta \beta/2$ is discussed in  Thm.~\ref{thm:main} and Appendix~\ref{app:Ucomplexity}.

For simplicity, we do not consider the inner workings of these unitaries in the assumed access model, thereby avoiding implementation details 
that depend on the particular method used. The gate complexities of these unitaries follow from explicit constructions, and these also depend on the presentation of the Hamiltonians. 
For example, in Appendix~\ref{app:Ucomplexity}, we analyze the 
gate complexity of $U$ using quantum signal processing and qubitization~\cite{LC17,LC19}, for the special case where $H_0$ and $H_1$ (and $W$) are presented as linear combinations of unitaries.
Quantum signal processing and the related Taylor series method also apply to the $d$-sparse model, 
where the Hamiltonians are represented by $d$-sparse matrices~\cite{BCC+14,BCC+15}. The gate complexity of these methods is polylogarithmic in the inverse of the approximation error. This is a desired feature,
as the number of amplitude amplification rounds can be exponential in $1/\epsilon$, implying that each round
needs to be simulated with  precision that is exponentially small in $1/\epsilon$, in the worst case.

%%%%%%%%%%%%%%%%%%%%%%%%%%%%%%%%%%%%%%%%%%%%%%%%%%%%%%%%%%%%%%%%%%%%%%%
\subsection{The exponential operator}
\label{sec:exponentialoperator}

We analyze approximations of the exponential operator $e^{-\beta W/2}$ given that it is acting on $(\cU \otimes \one)\ket{\Psi_0}$.
First, we show that an $(\epsilon/2)$-relative approximation of this operator suffices to prepare $\tau_1$ from $\ket{\Psi_0}$ and satisfy Eq.~\eqref{eq:TSPP}. 
We introduce the work cutoff $w_l$ that, roughly, has the property that $(\cU \otimes \one) \ket{\Psi_0}$ has negligible support on the subspace spanned by eigenstates of $W$ of eigenvalues below $w_l$.
In other words, using either measurement scheme in Sec.~\ref{sec:fluctuationtheorems} to sample from $P(w)$, those work values below $w_l$ are not the rare events that need to be captured.
It suffices then to approximate the action of 
$e^{-\beta W/2}$ on a subspace where $w \ge w_l$ only.
We construct one such approximation using a Fourier series, with the desired relative error. This is a linear combination of unitaries, where each unitary is an integer power of $U$ in Eq.~\eqref{eq:Uaction}.

%%%%%%%%%%%%%%%%%%%%%%%%%%%%%%%%%%%%%%%%%%%%%%%%%%%%%%%%%%%%%%%%%%%%%%%
\subsubsection{Approximation errors and work cutoff}
\label{sec:approximations}

A main part of our quantum algorithm  consists of the implementation of an operator proportional to $e^{-\beta W/2}$ on $(\cU \otimes \one)\ket{\Psi_0}$.
This transformation is implemented approximately by an operator $X$, and we set sufficient conditions for the approximation errors below.

\begin{lemma}[Relative approximation]
\label{lem:relativeerror}
Let $\epsilon\ge 0$ be the approximation error and $X$ be a linear operator acting on $\cH_{\sy}\otimes \cH_{\sy'}$ that approximates $e^{-\beta W/2}$ on $(\cU \otimes \one )\ket{\Psi_0}$ and satisfies
\begin{align}
\|  (   e^{-\beta W/2} - X)(\cU \otimes \one ) \ket{\Psi_0} \| &\le  \frac{\epsilon} 2 \|e^{-\beta W/2} (\cU \otimes \one ) \ket{\Psi_0} \| \label{eq:expapprox}\;.
\end{align}
Then, the (mixed) quantum state 
\begin{align}
\label{eq:tau1}
    \tau_1:= \frac {\tr_{\cH_{\sy'}} ( X (\cU \otimes \one )\ketbra{\Psi_0} (\cU^\dagger \otimes \one )X^\dagger) }{\|X (\cU \otimes \one )\ket{\Psi_0}\|^2}
\end{align}
satisfies Eq.~\eqref{eq:TSPP}, that is,
\begin{align}
    \frac 1 2 \|\tau_1 - \rho_1\|_1 \le \epsilon \;.
\end{align}
\end{lemma}

\begin{proof}
Let
\begin{align}
   \ket{ \Phi_1} &:=\frac{e^{-\beta W/2} (\cU \otimes \one )\ket{\Psi_0}}{\| e^{-\beta W/2}(\cU \otimes \one ) \ket{\Psi_0}\|}  \;   ,\\
   \label{eq:Phi1'}
   \ket{ \Phi_1'} &:=\frac{X (\cU \otimes \one )\ket{\Psi_0}}{\| X(\cU \otimes \one ) \ket{\Psi_0}\|}  \; ,
\end{align}
be normalized states
so that, from Eqs.~\eqref{eq:proptorho1} and~\eqref{eq:tau1},
\begin{align}
\rho_1&=\tr_{\cH_{\sy'}} (\ketbra{\Phi_1}) \; ,\\
    \tau_1 &= \tr_{\cH_{\sy'}} (\ketbra{\Phi_1'})\;.
\end{align}
Let $\ket{\Gamma}:=\ket{\Phi_1}-\ket{\Phi_1'}$. The triangle inequality implies
\begin{align}
\| \ket{\Gamma}\| &=
 \left \| \frac{(e^{-\beta W/2}- X) (\cU \otimes \one )\ket{\Psi_0}}{\| e^{-\beta W/2} (\cU \otimes \one )\ket{\Psi_0}\|} + \left(\frac{\| X(\cU\otimes\one)\ket{\Psi_0}\|} {\|e^{-\beta W/2}(\cU\otimes\one)\ket{\Psi_0}\|}-1\right) \ket{\Phi_1'}\right \|\\
    &\le \frac{\epsilon}{2} + \left| \frac{\| X(\cU \otimes \one ) \ket{\Psi_0}\|}{\|e^{-\beta W/2} (\cU \otimes \one )\ket{\Psi_0}\|}- 1\right| \\
    & \le \epsilon \;.
\end{align}
We use the non-increasing property of the trace distance to obtain
\begin{align}
    \frac 1 2  \| \tau_1 - \rho_1\|_1 &=  \frac 1 2 \|\tr_{\cH_{\sy'}} \left(\ketbra{ \Phi'_1} - \ketbra{\Phi_1}\right)\|_1 \\
    & \le  \frac 1 2 \|  \ketbra{\Phi'_1} - \ketbra{\Phi_1} \|_1 \\
    & = \frac 1 2 \| - \ket \Gamma \bra{ \Phi'_1} -\ket{\Phi_1}\bra{\Gamma}\|_1\\
    & \le \| \ket \Gamma\| \\
    & \le \epsilon  \;,
\end{align}
and Eq.~\eqref{eq:TSPP} is satisfied.
\end{proof}

Note that Eq.~\eqref{eq:psi0psi1mapB} implies
\begin{align}
 \|e^{-\beta W/2} (\cU \otimes \one ) \ket{\Psi_0} \|&= \|(\one \otimes \cU^*)e^{-\beta W/2} (\cU \otimes \one ) \ket{\Psi_0} \|
 \\
 & = e^{-\beta \, \Delta \! A/2} \|\ket{\Psi_1}\| \\
 \label{eq:norm-PFratio}
 & =e^{-\beta \, \Delta \! A/2} \;.
\end{align}
This is nothing but Jarzynski equality in disguise. Then, the condition in Eq.~\eqref{eq:expapprox} can be alternatively written as
\begin{align}
\|  (   e^{-\beta W/2} - X)(\cU \otimes \one ) \ket{\Psi_0} \| &\le  \frac{\epsilon} 2 e^{-\beta \, \Delta \! A/2}   \label{eq:expapproxPF}\;.
\end{align}

For $w_l \ge w_{\min}$, we define
\begin{align}
    \Pi_{< w_l} &:= \sum_{m,n \colon w_{m,n} < w_l}
    \ketbra{\psi_{m,n}} \;,
\end{align}
which acts on $\cH_{\sy} \otimes \cH_{\sy'}$.
This is the orthogonal projector into the subspace spanned by eigenstates $\ket{\psi_{m,n}}$ of $W$ of eigenvalue $w_{m,n}$ less than $w_l$. We obtain:
%%%%%%%%
\begin{lemma}[Work cutoff]
\label{lem:cutoffs}
Let $\epsilon>0$ be the approximation error and $w_l\in\mathbb R$ be a work cutoff that satisfies
\begin{align}
   \label{eq:Wlcondition}
    \| e^{-\beta W/2} \Pi_{< w_l}(\cU \otimes \one ) \ket{\Psi_0}\| &\le \frac \epsilon {6} \| e^{-\beta W/2} (\cU \otimes \one )\ket{\Psi_0}\| \; .
\end{align}
Let $X$ be a normal operator that satisfies $[X,W]=0,$ and approximates $e^{-\beta W/2}$ as
\begin{align}
\label{eq:lconstraint}
    \|(e^{-\beta W/2}-X) \ket{\psi_{m,n}}\|& \le 2 e^{-\beta w_{m,n}/2} \; , \; \forall m,n \colon w_{m,n}< w_l \; , \\
    \label{eq:hconstraint}
     \|(e^{-\beta W/2}-X)\ket{\psi_{m,n}}\|& \le \frac \epsilon {3} e^{-\beta w_{m,n}/2}\; , \; \forall m,n \colon w_{m,n} \ge w_l \; .
\end{align}
Then, Eq.~\eqref{eq:expapprox} is satisfied; that is,
\begin{align}
\|  (   e^{-\beta W/2} - X) (\cU \otimes \one)\ket{\Psi_0} \| &\le  \frac{\epsilon} 2 \|e^{-\beta W/2}  (\cU \otimes \one)\ket{\Psi_0} \| \;.
\end{align}
\end{lemma}

\begin{proof}
Following Sec.~\ref{sec:two-copy}, we write $(\cU \otimes \one)\ket{\Psi_0} = \sum_{m,n} \sqrt{p_{m,n}}e^{i\varphi_{m,n}} \ket{\psi_{m,n}}$. 
By assumption, the eigenstates of $X$ are also
$\ket{\psi_{m,n}}$.
Then, 
Eq.~\eqref{eq:lconstraint} implies 
\begin{align}
     \|(e^{-\beta W/2}-X) \Pi_{<w_l} (\cU \otimes \one)\ket{\Psi_0}\|^2 & = 
     \sum_{m,n \colon w_{m,n}<w_l} p_{m,n} \; \|(e^{-\beta w_{m,n}/2} - X)\ket {\psi_{m,n}}\|^2   \\
     & \le 4 \sum_{m,n \colon w_{m,n}<w_l} p_{m,n} e^{-\beta w_{m,n}}  \\
     & =4 \|e^{-\beta W/2} \Pi_{<w_l} (\cU \otimes \one) \ket{\Psi_0}\|^2 \\
     & \le  \frac {\epsilon^2} {9} \|e^{-\beta W/2} (\cU \otimes \one)\ket{\Psi_0} \|^2 \;,
\end{align}
where the last inequality follows from the assumption in Eq.~\eqref{eq:Wlcondition}.
Let $\Pi_{\ge w_l}:=\one -\Pi_{<w_l}$
be the orthogonal projector into the subspace spanned by eigenstates of $W$ of eigenvalue at least $w_l$. Then, 
 Eq.~\eqref{eq:hconstraint} implies
\begin{align}
    \|(e^{-\beta W/2}-X) \Pi_{\ge w_{l}} (\cU \otimes \one) \ket{\Psi_0}\|^2 & =
    \sum_{m,n \colon w_{m,n} \ge w_l} p_{m,n} \; \|(e^{-\beta W/2}-X) \ket{\psi_{m,n}}\|^2 
    \\& \le \frac {\epsilon^2} {9} \sum_{m,n \colon w_{m,n} \ge w_l} p_{m,n} e^{-\beta w_{m,n}}  \\
    % &\le \frac \epsilon {12} \left(\sum_{m,n} p(W_{m,n}) e^{-\beta W_{m,n}}  \right)^{1/2}  \\
    & = \frac {\epsilon^2} {9} \|e^{-\beta W/2} \Pi_{\ge w_l} (\cU \otimes \one) \ket{\Psi_0}\|^2 \\
   & \le \frac {\epsilon^2} {9} \|e^{-\beta W/2} (\cU \otimes \one) \ket{\Psi_0}\|^2 \;.
\end{align}
It follows that
\begin{align}
\nonumber
\| &(e^{-\beta W/2}- X) (\cU \otimes \one)\ket{\Psi_0}\| = \\
& = \| (e^{-\beta W/2}- X) (\Pi_{<w_l} +\Pi_{\ge w_{l}} )(\cU \otimes \one)\ket{\Psi_0}\| \\
& = \left( \| (e^{-\beta W/2}-X) \Pi_{<w_l} (\cU \otimes \one) \ket{\Psi_0}\|^2 + \| (e^{-\beta W/2}- X) \Pi_{\ge w_l} (\cU \otimes \one)\ket{\Psi_0}\|^2 \right)^{1/2}\\
& \le \left(\frac 2 9 \epsilon^2 \| e^{-\beta W/2} (\cU \otimes \one)\ket{\Psi_0}\|^2 \right)^{1/2} \\
& \le \frac \epsilon 2 \| e^{-\beta W/2}(\cU \otimes \one) \ket{\Psi_0}\|\;.
\end{align}

\end{proof}

Lemma~\ref{lem:cutoffs} basically states
that the action of $X$ needs to be an $\cO(\epsilon)$-relative approximation of $e^{-\beta W/2}$ only in the subspace where the eigenvalues of $W$ are, at least, $w_l$. For the subspace where the eigenvalues are less than $w_l$, an $\cO(1)$-relative approximation suffices.

The assumption in Eq.~\eqref{eq:Wlcondition}, stated in that way for clarity, is equivalent to Eq.~\eqref{eq:Wlconditionmain}:
\begin{lemma}
Equation~\eqref{eq:Wlcondition} is equivalent to Eq.~\eqref{eq:Wlconditionmain}, that is,
\begin{align}\label{eq:WlconditionRewrite}
    \sum_{w<w_l} P(w) e^{-\beta \, (w-\Delta \! A)} \le \left( \frac {\epsilon} {6} \right)^2  \;.
\end{align}
\end{lemma}
\begin{proof}
According to Eq.~\eqref{eq:TCdist}, the left hand side of Eq.~\eqref{eq:Wlcondition}
satisfies 
\begin{align}
 \|e^{-\beta W/2}\Pi_{<w_l}(\cU \otimes \one)\ket{\Psi_0}\|^2 &=
    \sum_{m,n: w_{m,n}<w_l} p_{m,n} e^{-\beta w_{m,n}}  \\
    \label{eq:supportandnorm}
    & = \sum_{w<w_l} P(w) e^{-\beta w} \;.
\end{align}
According to Eq.~\eqref{eq:norm-PFratio},
the right hand side of Eq.~\eqref{eq:Wlcondition}
satisfies 
\begin{align}
  \left(  \frac{\epsilon} {6} \| e^{-\beta W/2} (\cU \otimes \one )\ket{\Psi_0}\| \right)^2 = \left( \frac {\epsilon} {6} \right)^2 e^{-\beta \Delta \! A} \; .
\end{align}
\end{proof}

Equation~\eqref{eq:WlconditionRewrite} can alternatively be stated as
\begin{align}
    \label{eq:WlconditionPrev}
    \sum_{w>-w_l} P^{\rm rev}(w) \le \left(\frac{\epsilon}{6}\right)^2 \; ,
\end{align}
where $P^{\rm rev}(w)$ is the work distribution of the reverse process further discussed in Sec.~\ref{sec:generalHamiltonians}. In the reverse process, the roles of $H_0$ and $H_1$ are interchanged, and the non-equilibrium unitary is $\cU^\dagger$. For concreteness, in the two-time measurement scheme we start with the thermal state $\rho_1$, perform a projective measurement of $H_1$ to obtain the energy (eigenvalue)  $\varepsilon_{1,n}$, evolve with $\cU^\dagger$, and finally perform a projective measurement of $H_0$ to obtain $\varepsilon_{0,m}$. Work for this realization is $\varepsilon_{0,m}-\varepsilon_{1,n}$. At first, it might come as a surprise that the complexity of our quantum algorithm can also be determined from properties of this reverse process. However, Crooks fluctuation theorem~\cite{Crooks1999Entropy} relates the work distributions of the two processes, $P(w)$ and $P^{\rm rev}(w)$, allowing us to express the condition in Eq.~\eqref{eq:Wlconditionmain} in terms of either process.

We note that Ref~.\cite{lu2001accuracy} proposed a ‘‘neglected-tail’’ bias model that estimates the bias of free-energy calculations by considering explicitly the effect of poor sampling of the tails of the work distribution. This model was later refined in Ref.~\cite{halpern2016number}.
In both works a threshold value for work is introduced for analyzing the error of the algorithm which itself does not depend on the threshold. In contrast, in our case the work cutoff plays an important role in the quantum algorithm and its complexity.

%%%%%%%%%%%%%%%%%%%%%%%%%%%%%%%%%%%%%%%%%%%
%%%%%%%%%%%%%%%%%%%%%%%%%%%%%%%%%%%%%%%%%%%
\subsubsection{Fourier series approximation}
\label{sec:Fourier}

We approximate the action of the exponential operator $e^{-\beta W /2}$
using an approach based on Fourier series. This approach differs from that of Ref.~\cite{CSS21} for computing partition functions as we are addressing a more general case where the Hamiltonians are not necessarily positive semidefinite. Another difference is that here 
we  use the cutoff $w_l$ for better convergence.

Our starting point is the identity ($x \in \mathbb R)$
\begin{align}
\label{eq:expconvolution}
    e^{- x}&= (f \star g) (x) \\
    & = \int dy \; f(x-y) g(y) \;,
\end{align}
where $f(x)$ is the normal distribution
\begin{align}
\label{eq:fFourier}
    f(x):= \frac{1}{\sqrt \pi}e^{- x^2}
\end{align} 
and $g(x):= e^{-1/4} e^{- x}$.
As we are only interested in the case where $x$ is bounded, we modify $g(x)$ such that 
\begin{align}
\label{eq:gFourier}
    g(x):= \left \{ \begin{matrix} e^{-1/4} e^{- x} & x \ge -\Delta -1/2 \\ 0 & x < -\Delta -1/2  \end{matrix} \right . \;,
\end{align}
where $\Delta >0$ depends on the approximation error and is chosen below. That is, we approximate $e^{-x}$ by $h(x):=(f \star g)(x)$, which will be a good approximation in the region of interest ($x \ge 0$). A simple integration gives
\begin{align}
\label{eq:gdef}
    h(x)& = \frac{e^{-1/4}}{\sqrt \pi} \int_{-\Delta-1/2}^\infty dy \; e^{-(x-y)^2 -y} \\
    & = e^{-x}\frac{1}{\sqrt \pi} \int_{-\Delta-1/2}^\infty dy \; e^{-(y-(x-1/2))^2} \\
    & = e^{-x}\frac{1}{\sqrt \pi} \int_{-\Delta-x}^\infty dy \; e^{-y^2}  \\
    \label{eq:herf}
    & = e^{-x} (1+{\rm Erf}(\Delta+x))/2 \;,
\end{align}
where ${\rm{Erf}}(z)= \frac{2}{\sqrt{\pi}} \int^z_{0} dy \; e^{-y^2}$ is the error function.
The convolution theorem implies
\begin{align}
\label{eq:H(omega)}
    H(\omega)=\sqrt{2\pi} F(\omega)G(\omega) \;,
\end{align}
where $F(\omega)$, $G(\omega)$, and $H(\omega)$ are the (unitary) Fourier transforms of $f(x)$, $g(x)$, and $h(x)$, respectively.
 In particular,
\begin{align}
\label{eq:FFouriertransform}
  F(\omega) &= \frac {1} {\sqrt{2\pi}} e^{-\frac{\omega^2}{4 }} \;, \\
  G(\omega) & =  \frac {e^{-1/4}} {\sqrt{2\pi}}  \int_{-\Delta -1/2}^\infty  dx \; e^{-  x} e^{-i \omega x }\\
  \label{eq:GFouriertransform}  & = \;  \frac {e^{-1/4}} {\sqrt{2\pi}} \frac{e^{(1+ i \omega)(\Delta + 1/2) }}{1+ i \omega},
\end{align}
and thus $|H(\omega)|$ decays rapidly with $|\omega|$, a property that is useful for our final approximation to work. Note that $h(x)=\frac 1 {\sqrt{2\pi}} \int d \omega \; H(\omega) e^{i \omega x}$ and we can start from this expression for the final approximation; see below. In Fig.~\ref{fig:Fourierapproximation} we plot $h(x)$ and $|H(\omega)|$.

\begin{figure}\centering
\includegraphics[scale=.4]{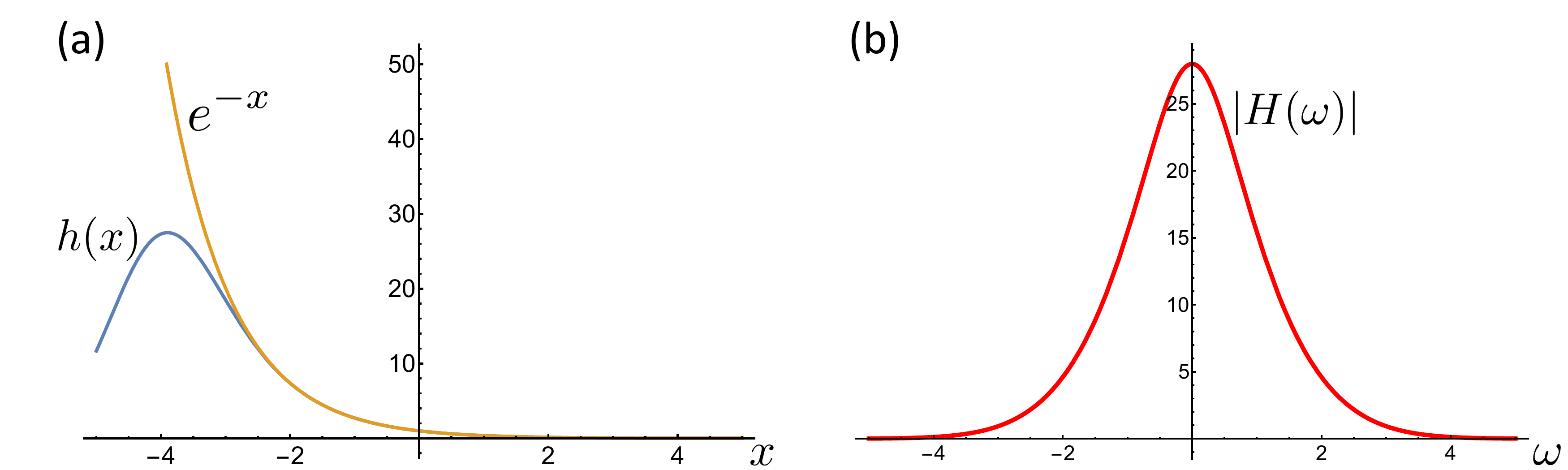}
\caption{(a) The approximation $h(x)$ of $e^{-x}$ for $\Delta=4$. (b) The absolute value of the Fourier transform of $h(x)$, $|H(\omega)|$, which is observed to decay rapidly with $|\omega|$.}
\label{fig:Fourierapproximation}
\end{figure}

As discussed in Sec.~\ref{sec:approximations}, we are interested in an approximation that uses the cutoff $w_l$. Using the previous identity and replacing $x \rightarrow \beta (W-w_l)/2$ , we propose 
\begin{align}
\label{eq:Fourierapproximationoperator}
 X(\beta,\epsilon)&:=
    \sum_{j=-J}^J \alpha_j U^j, \\
    \label{eq:Fouriercoefficients}
    \alpha_j &:= e^{-\beta w_l/2} \frac{\delta} {\sqrt{2\pi}}H(\omega_j)e^{-i \omega_j \beta w_l/2}  \;,
\end{align}   
where $\omega_j=j \delta$, and the unitary $U$ is that of Eq.~\eqref{eq:Uaction}, that is,
\begin{align}
\label{eq:Fourierunitary}
U:=   e^{i \delta \beta W /2} \;.
\end{align}
The integer $J>0$ and the values of $\delta>0$ and $\Delta>0$
determine the accuracy of the approximation. These
are chosen in the following lemma so that the error bounds in Eqs.~\eqref{eq:lconstraint} and~\eqref{eq:hconstraint} of Lemma~\ref{lem:cutoffs} are satisfied.

%%%%%%%%%%%%%%%%
%%%%%%%%%%%%%%%%%
\begin{lemma}[Fourier series approximation]
\label{lem:lemmafourier}
Let $\epsilon>0$ and $w_{\max} \ge w_l \ge w_{\min}$.
Then, if
\begin{align}
\Delta &= \max\{4,\sqrt{\ln(6/\epsilon)}\} \;,\\
z & =   \beta(w_{\max}-w_l)+  2\Delta^2 \; ,\\
 \; 
\delta& = 2\pi/z \;,
\end{align}
and for all integer $J \ge \lceil \frac 1 3 z^{3/2} \rceil-1$, the operator $X(\beta,\epsilon)$ in Eq.~\eqref{eq:Fourierapproximationoperator} satisfies Eqs.~\eqref{eq:lconstraint} and~\eqref{eq:hconstraint}, that is,
\begin{align}
%\label{eq:lconstraint}
    \|(e^{-\beta W/2}- X(\beta,\epsilon)) \ket{\psi_{m,n}}\|& \le 2 e^{-\beta w_{m,n}/2} \; , \; \forall m,n \colon w_{m,n}< w_l \; , \\
%    \label{eq:hconstraint}
     \|(e^{-\beta W/2}- X(\beta,\epsilon))\ket{\psi_{m,n}}\|& \le \frac \epsilon {3} e^{-\beta w_{m,n}/2}\; , \; \forall m,n \colon w_{m,n} \ge w_l \; .
\end{align}
In addition, the coefficients appearing in $X(\beta,\epsilon)$ satisfy
\begin{align}
\label{eq:L1fourier}
\alpha:= \sum_{j=-J}^J |\alpha_j| \le 2 e^{\Delta} e^{-\beta w_l/2} \;.
\end{align}

\end{lemma}
%%%%%%%%%%%%%%%%%%%

The proof of this result is given in Appendix~\ref{app:lemmafourier}.
It consists of three approximation steps: if $\hat x:=\beta(W-w_l)/2$, and starting from $e^{-\beta W/2}=e^{-\beta w_l/2}e^{-\hat x}$, we approximate this operator by $e^{-\beta w_l/2}h(\hat x)=e^{-\beta w_l/2}\frac 1 {\sqrt{2\pi}} \int d\omega H(\omega) e^{i  \omega \hat x}$, where $H(\omega)$ is given in Eq.~\eqref{eq:H(omega)}. Then we approximate this integral by an infinite (Riemann) sum and, last, we truncate this sum to obtain $X(\beta,\epsilon)$. The three approximation steps set sufficient conditions in $\Delta$, $\delta$, and $\omega_{J+1}$ to satisfy Eqs.~\eqref{eq:lconstraint} and~\eqref{eq:hconstraint}. Lemma~\ref{lem:lemmafourier} presents some choices for these parameters but further optimizations might be possible; see Appendix~\ref{app:lemmafourier} for details.
Also, in Appendix~\ref{app:HSapproximation}, we provide
an improved result for the special case where $W\ge 0$~\cite{CS16,CSS18}.

\vspace{0.2cm}
It is important to remark that the dependence of $\Delta$ and $\delta$
in $1/\epsilon$ is only (sub)logarithmic. Also, 
if we choose $J+1=\lceil \frac 1 3 z^{3/2}\rceil$,
$\omega_{J+1}= (J+1)\delta$ is polylogarithmic in $1/\epsilon$. 
 Moreover,  the term $e^\Delta$ is subpolynomial in $1/\epsilon$ and implies $e^\Delta =\cO(1/\epsilon^{o(1)})$ in Eq.~\eqref{eq:L1fourier}.
In practice, we expect that $\Delta=4$
according to our choices for the parameters, since $\sqrt{\ln(6/\epsilon)} > 4$ implies $\epsilon < 6.75 \times 10^{-7}$.

Our quantum algorithm  implements the operator $X(\beta,\epsilon)$.
%and this is almost proportional to $e^{-\beta W/2}$. 
This is a linear combination of unitaries $U^j$ that
can be realized as the evolution under $W$ for times $\omega_j \beta /2$. 
To implement the linear combination we discuss two approaches, one based on linear-combination-of-unitaries (LCU) in Sec.~\ref{sec:LCU},
and the other based on quantum signal processing (QSP) in Sec.~\ref{sec:QSP}. In the rest of the paper we drop the explicit dependence of $X(\beta,\epsilon)$ on $\beta$ and $\epsilon$ for simplicity.

%%%%%%%%%%%%%%%%%%%%%%%%%%%%%%%%%%%%%%%

\subsubsection{Implementation of $X$ using linear combination of unitaries}
\label{sec:LCU}

Our quantum algorithm requires  the implementation of the operator $X=\sum_{j=-J}^J \alpha_j U^j$ in Eq.~\eqref{eq:Fourierapproximationoperator}, which acts on $\cH_{\sy}\otimes \cH_{\sy'}$. This is a linear combination of unitaries (LCU), which
can be realized using ancilla qubits as follows.
Without loss of generality, we assume $2J+2=2^{\mathfrak {m}}$, where $\mathfrak m$
is the number of ancilla qubits and $\cH_{\rm anc} \equiv \mathbb C^{2^{\mathfrak m}}$
is the ancilla space.
We let $B$ be a unitary acting on $\cH_{\rm anc}$ that
maps the $\ket 0$ state of the ancilla to $\frac 1 {\sqrt{\alpha}} \sum_{j=0}^{2J} (\alpha_{-J+j})^{1/2} \ket j$, where the $\ket j$
are some basis states for the ancilla and $\alpha=\sum_{j=-J}^J |\alpha_j|$ is given in Eq.~\eqref{eq:L1fourier}. We also let $R:=\sum_{j=0}^{2J+1} \ketbra j \otimes  U^j$ be the unitary that implements $U^j$ conditional on $\ket j$. (This is the select unitary operation used in prior work on Hamiltonian simulation~\cite{BCC+15}.) Using $B$, $R$, and $B^T$, which is the transpose of $B$ in the computational basis,
we can realize the LCU.
The following result is a generalization of Lemma~6 in Ref.~\cite{CKS17} for the case $\alpha_j\in \mathbb C$:
\begin{lemma}[LCU]
\label{lem:lemmaLCU}
Let $X$ be a degree-$J$ Laurent polynomial of $U$ given as in Eq.~\eqref{eq:Fourierapproximationoperator}. Then, there exists a quantum circuit $S_{\rm{LCU}}$ that uses $\mathfrak m=\log_2(2J+2)$ ancilla qubits, $\mathcal{O}(J)$ controlled-$U$ and $U^\dagger$ operations, and $\mathcal{O}(J)$ two-qubit gates, such that
\begin{align}
\label{eq:lemmaLCU}
     S_{{\rm LCU}} \ket 0  \! \ket {\Psi}& =\ket 0  \left( \frac X \alpha   \ket \Psi \right) + \ket{\chi^\perp}  \;,
\end{align}
for all input states $\ket \Psi \in \cH_{\sy}\otimes \cH_{\sy'}$.
The state 
$\ket{\chi^\perp}$ is a (subnormalized) state supported on the  subspace orthogonal to $\ket 0$ of the ancilla, i.e., $(\ketbra 0 \otimes \mathds{1} \otimes \mathds{1}) \ket{\chi^\perp}=0$.
\end{lemma}
\begin{proof}
The quantum circuit is $S_{{\rm LCU}}:= (\one \otimes U^{-J})(B^{\rm T} \otimes \one \otimes \one) R (B\otimes \one \otimes \one)$ and
 is given in Fig.~\ref{fig:LCU}.
 First, we note
\begin{align}
\label{eq:SBaction}
   R (B\otimes \one \otimes \one)\ket 0 \! \ket{\Psi} = \frac 1 {\sqrt \alpha} \sum_{j=0}^{2J}(\alpha_{-J+j})^{1/2} \ket j   U^j \ket \Psi \;.
\end{align}
Then, since $\bra 0 B^{\rm T} \ket j =(\alpha_{-J+j}/\alpha)^{1/2} $, we obtain
\begin{align}
\label{eq:lemmaLCU2}
  (B^{\rm T}\otimes \one \otimes \one)   R (B\otimes \one \otimes \one)\ket 0 \! \ket{\Psi} =\ket 0 \left(\frac 1 { \alpha}
  \sum_{j=0}^{2J}(\alpha_{-J+j})   U^j \ket \Psi \right) + \ket{\chi'^\perp} \;,
\end{align}
where $\ket{\chi'^\perp}$ is a (subnormalized)
state supported on the subspace orthogonal to $\ket 0$ of the ancilla. Acting with $U^{-J}=U^{-(2^{\mathfrak m -1}-1)}$
in Eq.~\eqref{eq:lemmaLCU2} implies Eq.~\eqref{eq:lemmaLCU}. Note that $\alpha_j \in \mathbb C$ and
we have some flexibility in the choice of $(\alpha_{-J+j})^{1/2}$, but any such choice works.

The unitary $S_{{\rm LCU}}$ acts on
$\cH_{\rm anc} \otimes \cH_{\sy} \otimes \cH_{\sy'}$ and implements $X/\alpha$ on any state $\ket{\Psi}$ of $\cH_{\sy} \otimes \cH_{\sy'}$. That is, $S_{{\rm LCU}}$ is a block-encoding of $X/\alpha$. Assuming $2J+2=2^{\mathfrak m}$, the quantum circuit uses $\mathfrak m$ ancilla qubits. This simplifies the implementation of $R$, which can be done by controlling $U$ on individual ancillas, as shown in Fig.~\ref{fig:LCU}. Implementing $S_{{\rm LCU}}$ then requires $\cO(J)$ uses of controlled-$U$ and   $U^\dagger$, and also requires $\cO(J)$ two-qubit gates for $B$, $B^{\rm T}$, and other controlled gates.
\end{proof}

\begin{figure}
    \centering
    \mbox{
    \Qcircuit @C=1.5em @R=0.7em {
        &&&&&&& R &&&&& \\
        \hspace{5em}&\lstick{\ket{0}} & \qw   & \multigate{3}{B} & \qw          & \ctrl{4} & \qw & \qw &&& \qw & \multigate{3}{B^{\rm{T}}} & \qw\\
        \mathfrak{m} \hspace{5em}&\lstick{\ket{0}} & \qw   & \ghost{B}        & \qw  & \qw& \ctrl{3}  & \qw   &&& \qw & \ghost{B^{\rm{T}}} & \qw\\
        \hspace{5em}&& \vdots &&&&&& \cdots &&&&& \\
        \hspace{5em}&\lstick{\ket{0}} & \qw   & \ghost{B}        & \qw & \qw & \qw & \qw  &&& \ctrl{1} & \ghost{B^{\rm{T}}} & \qw \gategroup{2}{1}{5}{1}{2em}{\{} \\
        %\vspace{20em} \\
        &                   & \qw   & \qw & \qw & \multigate{3}{U} & \multigate{3}{U^2} & \qw &&& \multigate{3}{U^{2^{\mathfrak{m}-1}}}& \multigate{3}{U^{-2^{\mathfrak{m}-1}+1}} & \qw\\
        &                   & \qw   & \qw & \qw & \ghost{U} & \ghost{U^2} & \qw &&& \ghost{U^{2^{\mathfrak{m}-1}}}& \ghost{U^{-2^{\mathfrak{m}-1}+1}} & \qw\\
        &\lstick{\ket{\Psi}}& \vdots &    &     &           &           && \cdots & &&&&\\
        &                   & \qw   & \qw &\qw & \ghost{U} & \ghost{U^2} & \qw   &&& \ghost{U^{2^{\mathfrak{m}-1}}} & \ghost{U^{-2^{\mathfrak{m}-1}+1}} & \qw \gategroup{2}{6}{9}{11}{1em}{--}
    }}
    \caption{
        The quantum circuit $S_{\rm LCU}$, which is a block-encoding of $X/\alpha$, and realizes the LCU. The filled circles denote operations controlled on the state $\ket 1$ of the corresponding ancilla qubit. The number of ancilla qubits is $\mathfrak{m}= \log_2(2J+2)$. }
    \label{fig:LCU}
\end{figure}
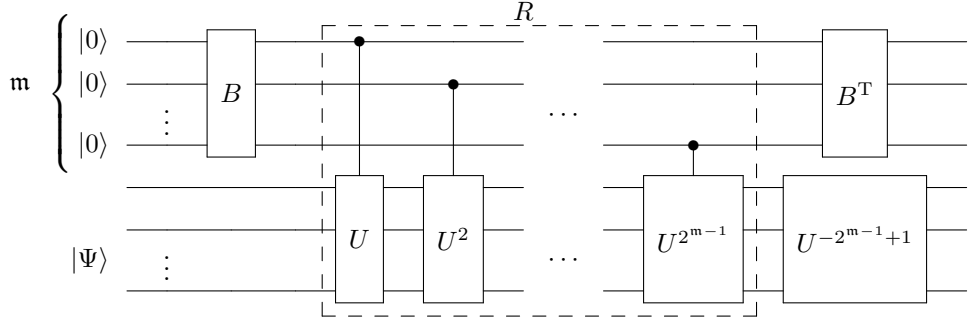

%%%%%%%%%%%%%%%%%%%%%%%%%%%%%%%%%%%%%%%
%%%%%%%%%%%%%%%%%%%%%%%%%%%%%%%%%%%%%%%

\subsubsection{Implementation of $X$ using quantum signal processing}
\label{sec:QSP}
Quantum signal processing (QSP) also allows us to implement $X= \sum^J_{j= -J} \alpha_j U^j$~\cite{LC17,LYC16}.
QSP is a sequence of rotations on an ancilla qubit interleaved
with controlled-$U$ and controlled-$U^\dagger$ operations. It allows one
to reduce the ancilla overhead to a constant, which contrasts the LCU approach of Sec.~\ref{sec:LCU}, at the expense of classically
computing the phases for the ancilla rotations.

Let $\Phi= \{\phi_0, \ldots, \phi_d\}$ be a set
of $d+1$ real phases and 
\begin{align}
\label{eq:QSP_encoding}
V_{\Phi}:= e^{i\phi_0 ( {\rm Z} \otimes \one)} \cW e^{i\phi_1 ({\rm Z} \otimes \one)} \cdots \cW e^{i\phi_d( {\rm Z} \otimes \one)}
\end{align}
be a quantum circuit shown in Fig~\ref{fig:QSPPACKcircuit}, where 
\begin{align}
 %   Z&=
    {\rm Z} \otimes \one 
    & =  \begin{pmatrix} \one & 0 \cr 0& -\one  \end{pmatrix}
\end{align}
is the unitary operator that applies the diagonal Pauli ${\rm Z}=\begin{pmatrix} 1& 0 \cr 0 & -1 \end{pmatrix}$ on the ancilla qubit,
\begin{align}
\label{eq:cWoperator}
    \cW &= \ketbra + \otimes \cV + \ketbra - \otimes \cV^\dagger \\
    & = \frac 1 2 \begin{pmatrix} \cV + \cV^\dagger & \cV - \cV^\dagger \cr \cV- \cV^\dagger & \cV + \cV^\dagger\end{pmatrix}\;,
\end{align}
and $\cV$ is any unitary. That is, the unitary $\cW$ implements $\cV$ or $\cV^\dagger$ conditioned on the state of the ancilla being $\ket +=\frac 1 {\sqrt 2}(\ket 0 + \ket 1)$ or $\ket -=\frac 1 {\sqrt 2}(\ket 0 - \ket 1)$, respectively.
The quantum circuit $V_{\Phi}$ is given in Fig~\ref{fig:QSPPACKcircuit}.
It can be shown to satisfy~\cite{GSYW18}
\begin{align}
\label{eq:VPhi}
   V_{\Phi}= \begin{pmatrix} P(\frac 1 2 (\cV+\cV^\dagger)) & \frac 1 2 (\cV-\cV^\dagger) Q(\frac 1 2 (\cV+\cV^\dagger)) \cr \frac 1 2 (\cV-\cV^\dagger) Q^*(\frac 1 2 (\cV+\cV^\dagger)) & P^*(\frac 1 2 (\cV+\cV^\dagger))\end{pmatrix}  \;,
\end{align}
where $P(y)$ and $Q(y)$ are complex polynomials of degree $\le d$ and $\le d-1$, respectively,  $P(y)=P(-y)$ is an even function and $Q(y)=-Q(-y)$ is an odd function, if $d$ is even. Our objective is to find two sets of phases, $\Phi_1$ and $\Phi_2$,
such that $V_{\Phi_1}$ and $V_{\Phi_2}$
encode parts of the operator $X$ in their blocks when we replace $\cV \rightarrow U^{\frac 1 2}$.

To show that $\Phi_1$ and $\Phi_2$
exist, we apply Thm. 5 of Ref.~\cite{GSYW18} 
in the following result.
(See also Ref.~\cite{haah2019product} and Ref.~\cite{dong2021} for an overview of QSP and how the phases can be determined).

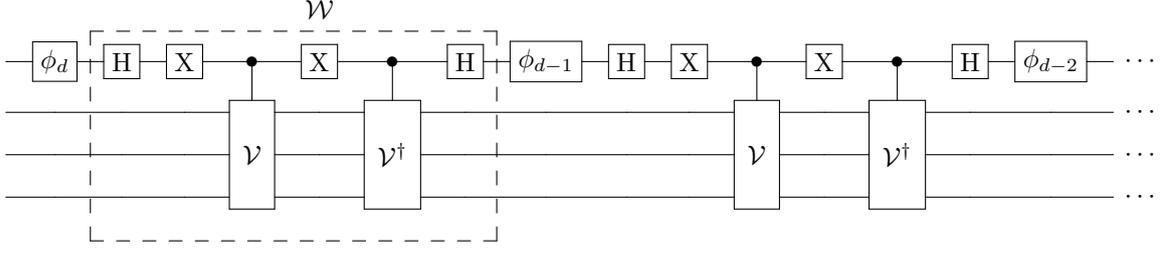
\begin{figure}[htb]
    \centering
    \mbox{
   \Qcircuit @C=1em @R=0.7em {
 && &&\push{\rule{0em}{1em}}&\mbox{$\cW$}&&&&&&&&&&&& \\
& \gate{\phi_d} & \gate{{\rm H}} & \gate{{\rm X}} & \ctrl{1}         &\gate{{\rm X}} & \ctrl{1}                   & \gate{{\rm H}} &  \gate{\phi_{d-1}} & \gate{{\rm H}} & \gate{{\rm X}} & \ctrl{1}         &\gate{{\rm X}} & \ctrl{1}                   & \gate{{\rm H}} & \gate{\phi_{d-2}} & \qw   & \cdots \\
& \qw           & \qw      & \qw      & \multigate{2}{\cV} & \qw     & \multigate{2}{\cV^{\dagger}}  & \qw      & \qw           & \qw      & \qw      & \multigate{2}{\cV} & \qw     & \multigate{2}{\cV^{\dagger}}  & \qw    &\qw            & \qw   & \cdots  \\
& \qw           & \qw      & \qw      & \ghost{\cV}       & \qw     & \ghost{\cV^{\dagger}}         & \qw      & \qw            & \qw      & \qw      & \ghost{\cV}       & \qw     & \ghost{\cV^{\dagger}}         & \qw   &\qw            & \qw   & \cdots  \\
& \qw           & \qw      & \qw      &  \ghost{\cV}       & \qw     & \ghost{\cV^{\dagger}}         & \qw      & \qw            & \qw      & \qw      & \ghost{\cV}       & \qw     & \ghost{\cV^{\dagger}}         & \qw    &\qw            & \qw   & \cdots   \gategroup{2}{3}{6}{8}{1em}{--} \\
 &&&&&&&&&&&&&&&&&
}}
    \caption{
        The quantum circuit to implement the unitary $V_\Phi=e^{i \phi_0 ( \rm{Z} \otimes \one)}\cW e^{i \phi_1 ( \rm{Z} \otimes \one)}\cdots \cW e^{i \phi_d ( \rm{Z} \otimes \one)}$. The one-qubit gates H and X are the Hadamard and Pauli X gates, respectively.}
    \label{fig:QSPPACKcircuit}
\end{figure}

\begin{lemma}[QSP]
\label{lem:lemmaQSP}
Let $X$ be a degree-$J$ Laurent polynomial of $U$ given as in Eq.~\eqref{eq:Fourierapproximationoperator}. 
Then there exists a quantum circuit $S_{\rm QSP}$ that uses $\mathfrak m=3$ ancilla qubits, $\mathcal{O}(J)$ controlled-$U$ and controlled-$U^\dagger$ operations, and $\mathcal{O}(J)$ two-qubit gates,  
such that 
\begin{align}\label{eq:lemmaQSP}
S_{\rm QSP} \ket{0}  \ket{\Psi}=\ket 0 \left( \frac{X}{2 \alpha}   \ket{\Psi} \right) + \ket{\chi^\perp}\;,
\end{align}
for all input states $\ket \Psi \in \cH_{\sy}\otimes \cH_{\sy'}$.
The state $\ket{\chi^\perp}$ is a (subnormalized) quantum state supported on the  subspace orthogonal to $\ket 0$ of the ancilla, i.e., $(\ketbra{0} \otimes \one \otimes \one )\ket{\chi^\perp}=0$.
\end{lemma}

\begin{proof}
We begin our proof by decomposing the polynomial $X$ of Eq.~\eqref{eq:Fourierapproximationoperator} as a sum of other polynomials, such that Thm. 5 of Ref.~\cite{GSYW18} can be used.
Each polynomial will then be realized by unitaries $V_{\Phi_1}$ and $V_{\Phi_2}$, respectively, and $X/(2\alpha)$
can be implemented as a LCU involving these unitaries. The method of Sec.~\ref{sec:LCU}
can be used for the last step.

As a function of $U$, the operator $X$
does not have definite parity. To address this issue, we define the unitary $\cV:=U^{\frac 1 2}$, so that
\begin{align}
    X= \sum_{j=-J}^J \alpha_j \cV^{2j}\;.
\end{align}
The transformation $\cV \rightarrow -\cV$
leaves $X$ invariant, which now is an even  function of $\cV$. Since $\alpha_j=(\alpha_{-j})^*$ for all $j$, we can write
$X=X_1 + X_2$, where
\begin{align}
\label{eq:X1def}
    X_1&:= \alpha_0 \one+ 2\sum_{j=1}^J \Re(\alpha_j)\frac{\cV^{2j}+\cV^{-2j}} 2 \\
    \label{eq:X2def}
    X_2&:=- 2\sum_{j=1}^J \Im(\alpha_j)\frac{\cV^{2j}-\cV^{-2j}}{2i} \; .
\end{align}
These operators are also even functions of $\cV$.
We define the polynomials
\begin{align}
\label{eq:pQSP-main}
    p_1(y)&:=\frac 1 \alpha \left(\alpha_0 +2 \sum_{j=1}^J \Re(\alpha_j) \cT_{2j}(y) \right)\;, \\
\label{eq:qQSP-main}
    q_2(y)&:=-\frac 2 \alpha \sum_{j=1}^J \Im(\alpha_j) \cR_{2j-1}(y)\;,
\end{align}
where $\mathcal{T}_{j}(y)$ and $\mathcal{R}_{j}(y)$ are the $j$-th Chebyshev polynomials of the first and second kind, respectively, and $y \in [-1,1]$. In particular, if $y=\cos(\theta/2)$, where $\theta$ is an eigenphase of $U$, we obtain
\begin{align}
    p_1(\cos(\theta/2))&= \frac 1 \alpha \left( \alpha_0 + 2 \sum_{j=1}^J \Re(\alpha_j) \cos(j \theta)\right)\\
    \label{eq:qQSP-maincondition}
   i\sin(\theta/2) q_2(\cos(\theta/2))&=- i\frac  2 \alpha \sum_{j=1}^J \Im(\alpha_j)\sin(j \theta) \;,
\end{align}
where we used the properties $\cT_k(\cos \gamma)=\cos(k\gamma)$ and $\sin(\gamma)\cR_k(\cos(\gamma))=\sin((k+1)\gamma)$. Replacing $\cV \rightarrow e^{i\theta/2}$
and $\cV^\dagger \rightarrow e^{-i\theta/2}$
in Eqs.~\eqref{eq:X1def} and~\eqref{eq:X2def}
produces $\alpha p_1(\cos(\theta/2))$ and
$\alpha \sin(\theta/2) q_2(\cos(\theta/2))$, respectively. This implies
\begin{align}
  p_1\left(\frac 1 2 (\cV+\cV^\dagger)\right) & =\frac  {X_1} \alpha  \; , \\
 \frac 1 2 (\cV-\cV^\dagger)  q_2\left(\frac 1 2 (\cV+\cV^\dagger)\right) & =i\frac{  X_2 } \alpha \;.
\end{align}
Following Eq.~\eqref{eq:VPhi}, the
operators $X_1/\alpha$ and $i X_2/\alpha$
can appear in the first or the second block of unitaries $V_{\Phi_1}$ and $V_{\Phi_2}$, respectively.

We now show that the polynomials $p_1(y)$ and $q_2(y)$ satisfy the conditions in Thm. 5 of Ref.~\cite{GSYW18}.
First, we note that $p_1(y)$ and $q_2(y)$ have real coefficients, i.e., $p_1, q_2 \in \mathbb{R}[y]$. Second, $p_1(y)$
is an even function of degree $\le 2J$, and $q_2(y)$ is an odd function of degree $\le 2J-1$. Third,  $|p_1(y)|^2 < 1$, and $(1-y^2)|q_2(y)|^2 < 1$ for $y \in[-1,1]$, which result from Eq.~\eqref{eq:pQSP-main} and Eq.~\eqref{eq:qQSP-maincondition}, respectively.  
It follows that there exists an odd polynomial $q_1 \in \mathbb{R}[y]$ of degree $\le 2J-1$, and an even polynomial $p_2 \in \mathbb{R}[y]$ of  degree $\le 2J$, such that $(p_1(y))^2 + (1-y^2)(q_1(y))^2 \le 1$ and $(p_2(y))^2 + (1-y^2)(q_2(y))^2 \le 1$, for all $y \in [-1,1]$. 
There are many ways of choosing $q_1(y)$ and $p_2(y)$ (e.g., $q_1(y) =0$ and $p_2(y)=0$), 
however any such choice suffices.  Then, Thm. 5 of Ref.~\cite{GSYW18} implies the existence of two sets $\Phi_1$ and $\Phi_2$, of $2J+1$ phases each, such that
\begin{align}
    \Re(V_{\Phi_1})& = \begin{pmatrix} \frac {X_1}\alpha & . \cr . & . \end{pmatrix} \;, \\
    \Im(V_{\Phi_2}) &= \begin{pmatrix} . & \frac {X_2}\alpha \cr . & . \end{pmatrix} \;.
\end{align}

We can combine the unitaries so that
\begin{align}
\label{eq:LCUQSP-main}
  \frac 1 4 (  V_{\Phi_1} + V^\dagger_{\Phi_1}+  V_{\Phi_2} e^{-i \frac \pi 2 X} + e^{i \frac \pi 2 X}V^\dagger_{\Phi_2})=
  \begin{pmatrix} \frac X {2\alpha} &. \cr . & . \end{pmatrix} \;,
\end{align}
and the operator $X/(2\alpha)$ appears in the first block
of a linear combination of four unitaries.
The weight of the linear combination is 1 and the LCU can 
be implemented following the method in Sec.~\ref{sec:LCU}, where $B$ is now composed of two Hadamard gates. This requires two additional ancilla qubits and implements a unitary $S_{\rm QSP}$, shown in Fig.~\ref{fig:QSP-LCU},
that acts on $\cH_{\rm anc}\otimes \cH_{\sy}\otimes \cH_{\sy'}$ and is a block-encoding of $X/(2\alpha)$; that is
\begin{align}
   S_{\rm QSP} \ket{0}  \ket{\Psi}= \ket 0 \left( \frac{X}{2 \alpha}   \ket{\Psi} \right) + \ket{\chi^\perp} \;,
\end{align}
where $\ket{\chi^\perp}$ is supported on the subspace orthogonal to $\ket 0$ of the ancilla. The total number of ancilla qubits is $1+2=3$, $\cH_{\rm anc}=\mathbb C^{2^3}$, and the number of rotations on the ancilla qubit for $V_{\Phi_1}$, $V^\dagger_{\Phi_1}$, $V_{\Phi_2}$, or $V^\dagger_{\Phi_2}$, is $2J+1$, bringing the gate complexity to $\cO(J)$.

\begin{figure}
    \centering
    \mbox{
    \Qcircuit @C=1.5em @R=0.7em {
        &\lstick{\ket{0}}   & \gate{\rm H}  & \ctrl{1} & \ctrlo{1} & \ctrl{1} & \ctrl{1} & \ctrlo{1} & \ctrlo{1} & \gate{\rm H} & \qw \\
       &\lstick{\ket{0}}   & \gate{\rm H} & \ctrl{2} & \ctrl{2}  & \ctrlo{2}   & \ctrlo{2}   & \ctrlo{2} & \ctrlo{2} & \gate{\rm H} & \qw \\
      && &&&&&& && \\
        & \lstick{\ket{0}} & \qw & \multigate{3}{V_{\Phi_1}} & \multigate{3}{V^{\dagger}_{\Phi_1}} & \multigate{3}{e^{-i\frac{\pi}{2} X}} & \multigate{3}{V_{\Phi_2}} & \multigate{3}{V^{\dagger}_{\Phi_2}} & \multigate{3}{e^{i\frac{\pi}{2} X}} & \qw & \qw\\
        & & \qw & \ghost{V_{\Phi_1}} & \ghost{V^{\dagger}_{\Phi_1}} & \ghost{e^{-i\frac{\pi}{2} X}} & \ghost{V_{\Phi_2}} & \ghost{V^{\dagger}_{\Phi_2}} & \ghost{e^{i\frac{\pi}{2} X}} & \qw & \qw\\
        & \lstick{\ket{\Psi}} & \qw & \ghost{V_{\Phi_1}} & \ghost{V^{\dagger}_{\Phi_1}} & \ghost{e^{-i\frac{\pi}{2} X}} & \ghost{V_{\Phi_2}} & \ghost{V^{\dagger}_{\Phi_2}} & \ghost{e^{i\frac{\pi}{2} X}} & \qw & \qw \\ 
        & & \qw & \ghost{V_{\Phi_1}} & \ghost{V^{\dagger}_{\Phi_1}} & \ghost{e^{-i\frac{\pi}{2} X}} & \ghost{V_{\Phi_2}} & \ghost{V^{\dagger}_{\Phi_2}} & \ghost{e^{i\frac{\pi}{2} X}} & \qw & \qw \gategroup{5}{2}{7}{2}{0.7em}{\{}
    }}
    \caption{
        The quantum circuit $S_{\rm QSP}$ that is a block-encoding of $X/(2 \alpha)$ and implements the linear combination in Eq.~\eqref{eq:LCUQSP-main}. The filled and open circles denote operations controlled on the states $\ket 1$ and $\ket 0$ of a corresponding ancilla qubit, respectively.
        Each unitary $V_{\Phi_1}$, $V^\dagger_{\Phi_1}$, $V_{\Phi_2}$, and $V^\dagger_{\Phi_2}$, uses one ancilla qubit from QSP. The total number of ancilla qubits is $\mathfrak m=3$.}
    \label{fig:QSP-LCU}
\end{figure}
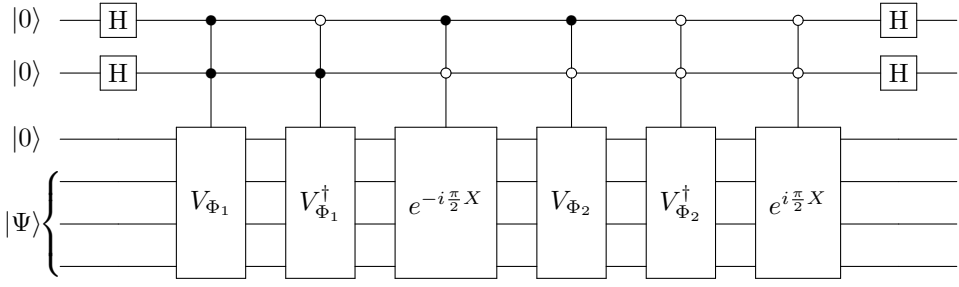

Last, we note that  $V_{\Phi_1}$, $V_{\Phi_2}$, or their inverses, make $2J$ uses of $\cW$
given in Eq.~\eqref{eq:cWoperator}. This requires $2J$ uses of $\cV$ and $2J$ uses of $\cV^\dagger$. However, a simple compilation of the circuit, discussed in Appendix~\ref{app:QSPImplementation}, allows us to implement $V_{\Phi_1}$, $V_{\Phi_2}$, or their inverses, using $J$ controlled-$U$ and $J$ controlled-$U^\dagger$ operations; see Fig.~\ref{fig:VPhi}.

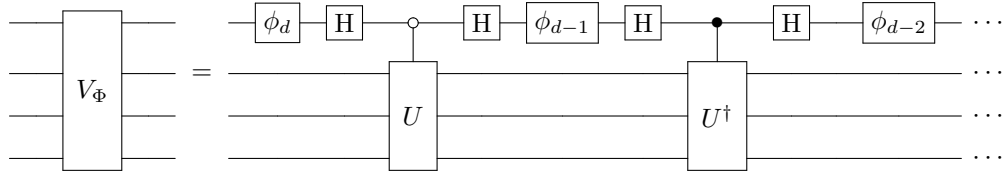
\begin{figure}[htb]
    \centering
    \mbox{
 \Qcircuit @C=1em @R=0.7em {
&&\qw & \multigate{3}{V_{\Phi}}&\qw& \qw &&& \gate{\phi_d} & \gate{{\rm H}}  & \ctrlo{1}                                             & \gate{{\rm H}} &  \gate{\phi_{d-1}} & \gate{{\rm H}}                      & \ctrl{1}                         & \gate{{\rm H}} & \qw   &\gate{\phi_{d-2}} & \qw   & \cdots  \\
 &&\qw & \ghost{V_{\Phi}}&\qw& \qw &\mbox{=} && \qw           & \qw       & \multigate{2}{U}   & \qw & \qw            & \qw       & \multigate{2}{U^{\dagger}}  & \qw      & \qw  &\qw            & \qw   & \cdots  \\
&&\qw & \ghost{V_{\Phi}}&\qw& \qw &  && \qw           & \qw       & \ghost{U}                 & \qw       & \qw           & \qw            & \ghost{U^{\dagger}}          &  \qw     & \qw  &\qw            & \qw   & \cdots  \\
&&\qw & \ghost{V_{\Phi}}&\qw& \qw &&& \qw           & \qw       & \ghost{U}                & \qw       &\qw            & \qw            & \ghost{U^{\dagger}}          & \qw     & \qw  &\qw            & \qw   & \cdots 
}
}
 \caption{
        Quantum circuit to implement the unitary $V_{\Phi}$ of Fig.~\ref{fig:QSPPACKcircuit} using controlled-$U$ and controlled-$U^\dagger$ operations, for $\cV=U^{\frac 1 2}$. The filled and open circles denote operations controlled on the states $\ket 1$ and $\ket 0$ of the ancilla qubit, respectively. The number of single qubit rotations is $d=2J+1$. The one-qubit gate H is the Hadamard gate.}
    \label{fig:VPhi}
\end{figure}

\end{proof}

In Appendix~\ref{app:QSPImplementation}
we obtain the sets of phases $\Phi_1$ and $\Phi_2$
using the MATLAB package QSPPACK (available at https://github.com/qsppack/QSPPACK).

%%%%%%%%%%%%%%%%%%%%%%%%%%%%%%%%%%%%%%%%%%%%%
\subsection{Main quantum algorithm}
\label{sec:algorithm}

To solve the TSPP, our quantum algorithm prepares an approximation of $\ket{\Psi_1}$, the purification of $\rho_1$ given in Eq.~\eqref{eq:Psi1}, from $\ket{\Psi_0}$,
the purification of $\rho_0$ given in Eq.~\eqref{eq:Psi0}.
Let $X$ be the operator of Sec.~\ref{sec:Fourier} that approximates $e^{-\beta W/2}$ and $\cU$ any unitary that acts on $\cH_{\sy}\equiv \cH_{\sy'}$.
\begin{definition}[Quantum state prepared by the quantum algorithm]
\label{def:preparedstate}
The (normalized) quantum state that is prepared by the quantum algorithm with high probability is
\begin{align}
    \ket{\Psi_1'}:= \frac{(\one \otimes \cU^*)X(\cU \otimes \one) \ket{\Psi_0}} {\| (\one \otimes \cU^*)X(\cU \otimes \one) \ket{\Psi_0}\|} \;.
\end{align}
\end{definition}

Let $S$ be a unitary that is a block encoding of 
$X/\alpha$, which can be implemented using the LCU approach of Sec.~\ref{sec:LCU} ($S_{\rm LCU}$) or the QSP approach of Sec.~\ref{sec:QSP} ($S_{\rm QSP}$)\footnote{The unitary $S_{\rm QSP}$ of Sec.~\ref{sec:QSP} is a block-encoding of $X/(2\alpha)$. Here we disregard the factor $1/2$ for simplicity. Nevertheless, this factor needs to be considered in the complexity if we implement $X$ using QSP.}. This unitary acts on the  space $\cH =\cH_{{\text anc}} \otimes \cH_{\sy} \otimes \cH_{\sy'}$. We also let $S':= S (\one \otimes \cU \otimes \one)$ be another unitary that acts on $\cH$, and 
define the projectors $P':= S' (\ketbra{0} \otimes \ketbra{\Psi_{0}})S'^\dagger$ and $P:= \ketbra{0} \otimes \one \otimes \one$.

\begin{algorithm}[H]
\caption{Quantum algorithm for the TSPP}
\label{algo:main} 
\vspace{1mm}
{\bf{Input:}} Given are an inverse temperature $\beta \ge 0$, an approximation error $\epsilon>0$, an
arbitrary unitary $\cU$ that acts on $\cH_{\sy}\equiv \cH_{\sy'}$. \\
\\
{\bf{Steps:}}\\
1. Obtain a work cutoff $w_l$ that satisfies Eq.~\eqref{eq:Wlconditionmain} and compute $X$ according to Eq.~\eqref{eq:Fourierapproximationoperator}  and Lemma~\ref{lem:lemmafourier}. Construct the unitaries $S'$ and $S'^\dagger$.\\
2. Prepare the initial state $\ket 0 \ket{\Psi_0}$ and perform amplitude amplification with reflections $R':= \mathds{1} - 2P'$ and $R:= \mathds{1} - 2P$. \\
3. Apply the unitary $\one \otimes \one \otimes \cU^*$. \\
\\
{\bf{Output:}} A pure state $\ket{\Psi_1'}$ such that $\tau_1=\tr_{\cH_{\sy'}}(\ketbra{\Psi_1'})$ satisfies $\frac 1 2 \| \tau_1- \rho_1 \|_1 \le \epsilon$.
\vspace{1mm}
\end{algorithm}

To show that this quantum algorithm prepares $\ket{\Psi_1'}$ of Def.~\ref{def:preparedstate}, we first note
 \begin{align}
 \label{eq:S'action}
     S' \ket 0 \ket{\Psi_0} &=\ket 0  \frac 1 \alpha  X (\cU \otimes \one) \ket{\Psi_0}+ \ket{\chi^\perp} \;,
 \end{align}
where $\ket{\chi^\perp}$ is a subnormalized state supported in the subspace orthogonal to $\ket 0$ of the ancilla. Amplitude amplification using the reflections $R$ and $R'$ prepares
\begin{align}
    \frac{X (\cU \otimes \one) \ket{\Psi_0}} {\|X (\cU \otimes \one) \ket{\Psi_0}\|}  
\end{align}
with high probability.
This coincides with $\ket{\Phi_1'}$ in Eq.~\eqref{eq:Phi1'}.  The action of  $( \one \otimes \cU^*)$ on $\ket{\Phi_1'}$ transforms it to $\ket{\Psi_1'}$. The error bound follows from Lemma~\ref{lem:relativeerror}, where the hypothesis
is satisfied due to Lemma~\ref{lem:lemmafourier} and Lemma~\ref{lem:cutoffs}.

 %%%%%%%%%%%%%%%%%%%%%%%%%%%%%%%%%%%%%%%%%%%%%
\subsection{Complexity: Proof of Thm.~\ref{thm:main}}
\label{sec:algorithmcomplexity}

The query complexity of the quantum algorithm is determined by the number of amplitude amplification rounds and the number of uses of the unitaries in each round.
From Eq.~\eqref{eq:S'action}, the average number of amplitude amplification rounds is~\cite{BHMT02,KLM07}
\begin{align}
\label{eq:aarounds}
    Q= \cO \left( \frac \alpha {\|X(\cU \otimes \one) \ket{\Psi_0}\|}\right) \;.
\end{align}
Replacing for $\Delta$ in Eq.~\eqref{eq:L1fourier}
gives $\alpha =\cO(e^{\sqrt{\ln(1/\epsilon)}} e^{-\beta w_l/2})$. The triangle inequality and Eq.~\eqref{eq:expapprox}
imply
\begin{align}
\| e^{-\beta W/2} (\cU \otimes \one) \ket{\Psi_0}\| - 
\| X(\cU \otimes \one) \ket{\Psi_0} \|& \le
\| (e^{-\beta W/2}-X) (\cU \otimes \one) \ket{\Psi_0}\| \\
& \le \frac \epsilon 2 \| e^{-\beta W/2} (\cU \otimes \one) \ket{\Psi_0}\| \\
&\le \frac 1 2 \| e^{-\beta W/2} (\cU \otimes \one) \ket{\Psi_0}\|
\;.
\end{align}
Equivalently,
\begin{align}
   \| X(\cU \otimes \one) \ket{\Psi_0} \| &\ge  \frac{1}{2} \| e^{-\beta W/2} (\cU \otimes \one) \ket{\Psi_0}\| \\
   & =  \frac{1}{2}  e^{-\beta \, \Delta \! A/2} \;,
\end{align}
where the last equality follows from Eq.~\eqref{eq:norm-PFratio}.
Then, Eq.~\eqref{eq:aarounds} is
\begin{align}
    Q &= \cO \left( e^{\sqrt{\ln(1/\epsilon)}} e^{\beta \, (\Delta \! A- w_l)/2}  \right) \;, \\
\end{align}
which proves the statement in Thm.~\ref{thm:main}.

Each amplitude amplification round applies  the reflections $R$ and $R'$ once, requiring one use
of $U_0$, the unitary that prepares $\ket{\Psi_0}=U_0 \ket0$, and one use of $U_0^\dagger$, $\cU$, and $\cU^\dagger$. 
It also makes $\cO(J)$ uses of the controlled-$U$ and controlled-$U^\dagger$, where
\begin{align}
    J = \cO(    (\ln(1/\epsilon) + \beta(w_{\max}-w_l))^{3/2}) 
\end{align}
according to Lemma~\ref{lem:lemmafourier}.
If $X$ is implemented using the LCU approach of Sec.~\ref{sec:LCU},
then each amplitude amplification round uses $B$, $B^\dagger$, $B^{\rm T}$, and $B^*$ once. The number of two-qubit gates for each of these is $\cO(J)$.
If $X$ is implemented using the QSP approach of Sec.~\ref{sec:LCU},
then each amplitude amplification round also requires $\cO(J)$ two-qubit gates to implement the corresponding sequences of single-qubit rotations required by QSP.
The  additional  gate complexity per round using the LCU or QSP approaches is then 
$\cO((\ln(1/\epsilon) + \beta(w_{\max}-w_l))^{3/2})$.
 These results prove the remaining complexity statements in Thm.~\ref{thm:main}. To prepare $\ket{\Psi_1'}$, the quantum algorithm uses $\cU^*$ once.

Last,   Lemma~\ref{lem:lemmafourier} implies
\begin{align}
\delta &\le \frac {2\pi} {32 + \beta(w_{\max}-w_l)} \\
& \le \frac \pi {16} \;.
\end{align}
\qed

The quantum algorithm requires obtaining a work cutoff $w_l$ in Step 1
that satisfies Eq.~\eqref{eq:Wlconditionmain} or, equivalently, Eq.~\eqref{eq:Wlcondition}. 
This incurs in an additional complexity
that was disregarded in our analysis. 
Next we show how to efficiently determine possible values of $w_l$ for several classes of Hamiltonians and, if these values are used, the dominant complexity of the quantum algorithm is expected to be that of Step 2. 
The values of $w_l$ discussed below might be improved. However, finding the largest $w_l$
that satisfies Eq.~\eqref{eq:Wlconditionmain}
is expected to be computationally intensive in general.

%%%%%%%%%%%%%%%%%%%%%%%%%%%%%%%%%%%%%%%%%%%%%
%%%%%%%%%%%%%%%%%%%%%%%%%%%%%%%%%%%%%%%%%%%%%

\section{Applications: $w_l$ for various Hamiltonians}
\label{sec:applications}

In this section we provide suitable choices
for the work cutoff $w_l$ in various examples, including 
general Hamiltonians, commuting Hamiltonians, and local spin Hamiltonians.

%%%%%%%%%%%%%%%%%%%%%%%%%%%%%%%%%%%%%%%%%%%%%%%%%%%%%%%%%%%
\subsection{General Hamiltonians: Proof of Thm.~\ref{thm:w_lgeneral}}
\label{sec:generalHamiltonians}

The proof follows from Crooks equality~\cite{Crooks1999Entropy}, which we provide below, and we consider first the case where the non-equilibrium unitary is simply $\cU=\one$.
Let $\rho_0$ and $\rho_1$ be the thermal states at inverse temperature $\beta \ge 0$ associated with $H_0$ and $H_1$, respectively, and $P(w)$ the work distribution that follows from the two-time or two-copy measurement schemes of Sec.~\ref{sec:fluctuationtheorems}. If $w_l \le 0$, we obtain
\begin{align}
    \sum_{w < w_l} P(w) e^{-\beta w}
    & \le \sum_{w < w_l}  \frac{w^2}{(w_l)^2} P(w) e^{-\beta w} \\
    \label{eq:MarkovIneq}
    & \le \frac{1}{(w_l)^2} \sum_w P(w) w^2 e^{-\beta w} \;.
\end{align}

Crooks equality relates $P(w)$ to a reverse work distribution $P^{\rm rev}(w)$ that follows from running the two-time measurement scheme in Sec.~\ref{sec:two-time}
in reverse. That is, starting from $\rho_1$, we first perform a measurement of $H_1$ and then a measurement of $H_0$. If the eigenvalues are distinct, the first measurement outputs $\varepsilon_{1,n}$ with (Boltzmann) probability $P_1(\varepsilon_{1,n})=e^{-\beta \varepsilon_{1,n}}/\cZ_1$. The second measurement outputs $\varepsilon_{0,m}$ with probability $P^{\rm rev}(\varepsilon_{0,m}|\varepsilon_{1,n})=|\! \bra{\phi_{0,m}}\phi_{1,n}\rangle|^2$, if the outcome of the first measurement is $\varepsilon_{1,n}$. Note that 
$P^{\rm rev}(\varepsilon_{0,m}|\varepsilon_{1,n})=P(\varepsilon_{1,n}|\varepsilon_{0,m})$, where the latter is given in Sec.~\ref{sec:two-time}. The amount of work is $w=\varepsilon_{0,m}-\varepsilon_{1,n}$ and is sampled according to 
\begin{align}
    P^{\rm rev}(w) &= \sum_{m,n:\varepsilon_{0,m}-\varepsilon_{1,n}=w} P^{\rm rev}(\varepsilon_{0,m}|\varepsilon_{1,n}) P_1(\varepsilon_{1,n}) \;.
\end{align}
For any function $f(w)$, a slight generalization of
Crooks equality states~\cite{Crooks1999Entropy,Pohorille2010}
\begin{align}
    \frac{\sum_w P(w)e^{-\beta w/2} f(w)}{\sum_w P^{\rm rev}(w) e^{-\beta w/2} f(-w)} &=  e^{-\beta \, \Delta \! A}  \;,
\end{align}
and, in particular, choosing $f(w) = e^{-\beta w/2} w^2$,
\begin{align}
\label{eq:Crooksw^2}
    \sum_w P(w) w^2 e^{-\beta w} = e^{-\beta \, \Delta \! A} \sum_w P^{\rm rev}(w) w^2 \;.
\end{align}
It is possible to bound the right hand side as follows.
The term $\sum_w P^{\rm rev}(w) w^2$ is the expected value of $w^2$ in the reverse process. This term can also be written as
\begin{align}
    \sum_{m,n} P^{\rm rev}(\varepsilon_{0,m}|\varepsilon_{1,n}) P_1(\varepsilon_{1,n}) (\varepsilon_{0,m}-\varepsilon_{1,n})^2 &=
   \sum_{m,n} P_1(\varepsilon_{1,n}) |\!\bra{\phi_{0,m}}\phi_{1,n}\rangle|^2 (\varepsilon_{0,m}-\varepsilon_{1,n})^2 \\
    & =   \sum_{m,n} P_1(\varepsilon_{1,n}) |\!\bra{\phi_{0,m}}(H_0-H_1)\ket{\phi_{1,n}}|^2 \\
    & = \sum_n P_1(\varepsilon_{1,n}) \bra{\phi_{1,n}} V^2 \ket{\phi_{1,n}} \;,
\end{align}
which is the expected value of $V^2$ in $\rho_1$, that is, $\tr(V^2 \rho_1)$.
Then,
\begin{align}
\label{eq:reversebound}
   \sum_w P^{\rm rev}(w) w^2\le \|V\|^2 \;.   
\end{align}

Equations~\eqref{eq:MarkovIneq},~\eqref{eq:Crooksw^2}, and~\eqref{eq:reversebound} imply
\begin{align}
\label{eq:WlV}
    \sum_{w<w_l} P(w) e^{-\beta w} \le  \left(\frac{\|V\|}{w_l} \right)^2 e^{-\beta \, \Delta \! A}
\end{align}
and, for all $w_l \le -6 \|V\|/\epsilon$, we obtain
Eq.~\eqref{eq:Wlconditionmain}; that is
\begin{align}
     \sum_{w<w_l} P(w) e^{-\beta w} \le  \left(\frac{\epsilon}{6} \right)^2 e^{-\beta \, \Delta \! A} \;.
\end{align}

This result is easily generalized to the case $\cU \ne \one$, which is equivalent to considering a two-time measurement scheme where $H_0$, and $\rho_0$ are replaced by $H_0':=\cU H_0 \cU^\dagger$ and $\rho_0':=\cU \rho_0 \cU^\dagger$, respectively. That is, 
the two-time measurement scheme with $\cU \ne \one$ is equivalent to that where a projective measurement of $H_0'$ is performed on $\rho_0'$ first, and 
a projective measurement of $H_1$ is done next.
Then, the perturbation
$V$ in Eq.~\eqref{eq:WlV} must be replaced by $V_\cU:=H_1-H_0'=H_1 - \cU H_0 \cU^\dagger$.
This proves Thm.~\ref{thm:w_lgeneral}.
\qed

\vspace{0.1cm}
Theorem~\ref{thm:w_lgeneral} is useful as it allows us to choose the work cutoff depending on $\|V_\cU \|$ and $\epsilon$ only. In the case where $\cU = \one$ it might be possible to exactly compute or tightly bound $\|V_\one \| \equiv \| V \|$ and thus efficiently compute a work cutoff. However, in the more general case where $\cU \neq \one$, it may not be straightforward to tightly bound $\|V_\cU \|$ such that $w_l \leq -6 \|V_\cU\|/\epsilon$ provides a useful result. Thus the complexity of computing a useful work cutoff in this more general setting remains open.
We further note, for $\epsilon \ll 1$, the result in Thm.~\ref{thm:w_lgeneral} can be impractical
as the complexity of our quantum algorithm is exponential in $-\beta w_l$ and $1/\epsilon$. In this case a better choice for the cutoff might be
$w_l=w_{\min}\ge -\|H_0\|-\|H_1\|$, so that $P(w)=0$ for all $w<w_l$.

%%%%%%%%%%%%%%%%%%%%%%%%%%%%%%%%%%%%%%%%%%%%%%%%%%%%%%%%%%%%%%
\subsection{Commuting Hamiltonians: Proof of Thm.~\ref{thm:w_lcommuting}}
\label{sec:commutingHamiltonians}

When Hamiltonians commute, i.e., $[H_1,H_0]=[V,H_0]=0$,
they can be diagonalized in the same basis where $\ket{\phi_{0,m}}$ are also eigenstates of $H_1$:
\begin{align}
    H_1 \ket{\phi_{0,m}} &= (H_0+V) \ket{\phi_{0,m}} \\
    & = (\varepsilon_{0,m} + \nu_m)\ket{\phi_{0,m}} \\
    & = \varepsilon_{1,m} \ket{\phi_{0,m}} \;,
\end{align}
where $\nu_m$ is the corresponding eigenvalue of $V$. 
Then,
\begin{align}
    w_{\min}&= \min_m (\varepsilon_{1,m}-\varepsilon_{0,m})\\
    & = \min_m \nu_m \\
    & \ge - \|V\| \;.
\end{align}
If $w < -\|V\| \le w_{\min}$, we obtain $P(w)=0$. Then,
Eq.~\eqref{eq:Wlconditioncommuting} is satisfied for all $w_l < -\|V\|$.
\qed

%%%%%%%%%%%%%%%%%%%%%%%%%%%%%%%%%%%%%%%%%%%%%%%%%%%%%%%%%%%%%%%
\subsection{Local Hamiltonians}
\label{sec:localHamiltonians}

We first define local Hamiltonians precisely:
\begin{definition}[Local Hamiltonians]\label{def:LocalHamiltonian}
For a given lattice $\Lambda$,
$H:= \sum_{X \in \Lambda} h_{X}$ is a $k$-local Hamiltonian of degree $g$ if each local term $h_X$ acts on at most $k$ qudits (spins) of dimension $d$ and if each qudit is involved in, at most, $g$ of the subsets $X$. The number of qudits is $\mathfrak {n}$. The local terms satisfy $\|h_X\|\leq h$, where $h$ is a strength parameter.
\end{definition}

The structure of local Hamiltonians allow us
to obtain complexity bounds for our quantum algorithm that are a significant improvement of the general case; see below.

%%%%%%%%%%%%%%%%%%%%%%%%%%%%%%%%%%%%%%%%%%%%%%
\subsubsection{The case $\cU = \one$: Proof of Theorem~\ref{thm:w_llocal}}
\label{sec:localHamiltoniansU=1}

Theorem~\ref{thm:w_llocal} is mainly a consequence of the following lemma that analyzes the support of distributions of eigenstates of $H_0$
on the eigenspaces of $H_1$. This result might be of independent interest and is proven in Appendix~\ref{app:EigenspaceOverlap}.

\begin{lemma}[Eigenspace modification under perturbations]
\label{lem:EigenspaceOverlap}
Let $H_1=H_0+V$, where $H_1=\sum_{X \in \Lambda} h_{1,X}$, $H_0=\sum_{X \in \Lambda} h_{0,X}$, and $V=\sum_{X \in \Lambda} v_X$ are $k$-local Hamiltonians as in Def.~\ref{def:LocalHamiltonian}. The number of local terms in $V$ is, at most, $M$ and $\|h_{0,X}\| \le h$ and $\|v_X\| \le v$ for all subsets $X \in \Lambda$. Then,
\begin{align}
\|{\Pi^1}_{\le \varepsilon_1} \Pi^0_{> \varepsilon_0}\| \leq e^{-\frac{\varepsilon_0 -\varepsilon_1- 2Mv}{2hgk}}\;,
\end{align}
where ${\Pi^1}_{\le \varepsilon_1}:=\sum_{n: \varepsilon_{1,n}\le \varepsilon_1} \ketbra{\phi_{1,n}}$ and $\Pi^0_{> \varepsilon_0}:=\sum_{m: \varepsilon_{0,m}> \varepsilon_0} \ketbra{\phi_{0,m}}$ are the orthogonal projectors into the subspaces of $H_1$ and $H_0$ of eigenvalues at most $\varepsilon_1$ and larger than $\varepsilon_0$, respectively, and we assume $\varepsilon_1 \le \varepsilon_0$.
\end{lemma}

We now use Lemma~\ref{lem:EigenspaceOverlap} to prove Thm.~\ref{thm:w_llocal}. For $\cU=\one$ and $w_l<0$, we obtain
\begin{align}
     \sum_{w<w_l} P(w) e^{-\beta w}
     & =  \sum_{m,n\colon \varepsilon_{1,n}-\varepsilon_{0,m}<w_l} P(\varepsilon_{1,n}|\varepsilon_{0,m})P_0(\varepsilon_{0,m}) e^{-\beta (\varepsilon_{1,n}-\varepsilon_{0,m})}\\
     &=\frac 1 {\cZ_0}
     \sum_{m,n\colon \varepsilon_{1,n}-\varepsilon_{0,m}<w_l} |\!\bra{\phi_{1,n}}\phi_{0,m}\rangle \! |^2  e^{-\beta \varepsilon_{1,n} }  \\
     & = \frac 1 {\cZ_0} \sum_{n} \bra{\phi_{1,n}} \Pi^0_{>\varepsilon_{1,n}-w_l}\ket{\phi_{1,n}}  {e^{-\beta \varepsilon_{1,n}}} \\
     & \le \frac 1 {\cZ_0} \sum_{n}  \| \Pi^1_{\le \varepsilon_{1,n}}\Pi^0_{>\varepsilon_{1,n}-w_l}\|^2 {e^{-\beta \varepsilon_{1,n}}}
     \\
     & \le \frac 1 {\cZ_0} e^{-\frac{-w_l - 2Mv}{hgk}}\ \sum_{n}  {e^{-\beta \varepsilon_{1,n}}} 
     \\ 
     & = e^{\frac{w_l + 2Mv}{hgk}} \frac{\cZ_1} {\cZ_0} \\
     & = e^{\frac{w_l + 2Mv}{hgk}} e^{-\beta \, \Delta \! A}\;.
\end{align}
The result follows by demanding $e^{\frac{w_l + 2Mv}{hgk}} \leq (\epsilon/6)^2$ or, equivalently,
\begin{align}
\label{eq:mainlocalbound}
    w_l \le -2Mv -2hgk \ln(6/\epsilon) \;.
\end{align}
\qed

%%%%%%%%%%%%%%%%%%%%%%%%%%%%%%%%%%%%%%%%%%%%%%
\subsubsection{The case $\cU \ne \one$: Transverse-field Ising model}
\label{sec:localHamiltoniansUne1}

As explained in Sec.~\ref{sec:generalHamiltonians}, the case $\cU \ne \one$ is equivalent to that where we replace $H_0 \rightarrow H_0' =\cU H_0 \cU^\dagger$ and $V \rightarrow V_\cU= H_1 - H_0'$, and we start from the thermal state $\rho'_0$ of $H_0'$. This allows for a direct generalization of Thm.~\ref{thm:w_llocal} to the case $\cU \ne \one$, which is useful as long as $H_0'$ and $V_\cU$ remain
as $k$-local Hamiltonians of degree $d$.
In fact, a generalization of this result
applies to the more general case  $\cU=\cU_1^\dagger \cU_0$, where $\cU_0$ and $\cU_1$ are two unitaries that take $H_0$ and $H_1$ into $k$-local Hamiltonians of degree $d$; that is $H_0'=\cU_0 H_0 \cU^\dagger_0$ and 
$H_1'=\cU_1 H_1 \cU^\dagger_1$ are local Hamiltonians.
This result also follows from an equivalent interpretation of the two-time measurement scheme of Sec.~\ref{sec:two-time}, where the projective measurement of $H_0$ on $\rho_0$
followed by a unitary $\cU_0$ is equivalent to a projective measurement of $H_0'$ on $\rho_0'$. If another unitary $\cU_1^\dagger$ is applied and then $H_1$ is measured, this step is equivalent to the projective measurement of $H_1'$ followed by the action of $\cU_1^\dagger$.
The projective measurements  of $H_0'$ and $H_1'$ produce a work value that is sampled according to the same distribution $P(w)$ of the original process. Hence, Thm.~\ref{thm:w_lcommuting} applies
to the case $\cU =\cU_1^\dagger \cU_0$
if we replace $V \rightarrow V_{\cU}:= H_1'-H_0'$.
 Often, a unitary $\cU$ will break the locality of the Hamiltonians and the bound in Eq.~\eqref{eq:mainlocalbound} does not apply to that case.

One example of a local spin system is the well-known transverse-field Ising model in one dimension, which has the Hamiltonian
\begin{align}
    \label{eq:TFIM}
    H=\sum_{j=1}^{\mathfrak{n}} {\rm Z}_j -v \sum_{j=1}^{\mathfrak{n}-1} \;
    {\rm X}_j {\rm X}_{j+1} \;,
\end{align}
where ${\rm X}_j$ and ${\rm Z}_j$ are the corresponding Pauli operators of the $j$-th qubit. This model can be solved exactly using the Jordan-Wigner mapping~\cite{LSM61,Pfeu70}. In particular, this mapping
provides a unitary that diagonalizes $H$
and takes it to a linear combination of the ${\rm Z}_j$'s 
only.
As a consequence, the thermal state can be prepared efficiently using other methods than the one presented in here. Since our goal in this section is to illustrate the benefits of using a nontrivial unitary $\cU$ in our quantum algorithm in general, we do not take advantage of the integrability or exact-solvability of the model in our calculations.

Let $\beta=1$ and consider $H_0$ given by $v=0$ and $H_1=H_0+V$ given by $v=1/2$ in Eq.~\eqref{eq:TFIM}, i.e.,  
\begin{align}
    H_0&=\sum_{j=1}^{\mathfrak{n}} {\rm Z}_j \;, \\
    V &= - \frac{1}{2} \sum_{j=1}^{\mathfrak{n}-1}
    {\rm X}_j{\rm X}_{j+1} \; .
\end{align}
Our quantum algorithm requires a unitary $U_0$ that prepares $\ket{\Psi_0}$, the purification of the thermal state of $H_0$ given in Eq.~\eqref{eq:Psi0}. A quantum circuit for $U_0$ can be simply constructed as follows. Let $R_y(\theta)=e^{-i\theta Y}$ be a single qubit rotation with $\theta=\arccos(e^{-\beta/2}/\sqrt{2 \cosh(\beta)})$. Acting on two $\mathfrak{n}$-qubit registers, we first apply $R_y(\theta)^{\otimes {\mathfrak n}}$ to the first register and next apply a sequence of $\mathfrak n$ controlled-NOT gates between the corresponding qubits of each register, controlled by the qubits in the first register. The prepared state is
\begin{align}
    U_0 \ket 0 \rightarrow \ket{\Psi_0}= \frac 1 {\sqrt{\cZ_0}}\sum_{b_1,\ldots,b_{\mathfrak n}} e^{-\frac \beta 2 ({\mathfrak n}-2b_1 - \ldots -2 b_{\mathfrak n})}\ket{b_1,\ldots,b_{\mathfrak n}} \ket{b_1,\ldots,b_{\mathfrak n}} \;,
\end{align}
where $b_j \in \{0,1\}$ and $\ket{b_1,\ldots,b_{\mathfrak n}}$ are states in the computational basis.
The gate complexity of $U_0$ is $\cO(\mathfrak n)$ in this case.

To illustrate the advantages of considering certain efficiently-implementable $\cU$'s in our quantum algorithm, we numerically analyze the work distribution arising in this model for various unitaries. 
In particular, we consider the unitaries $\cU_T= \mathcal{T} e^{-i \int_0^T dt H(t)}$, which are time-evolution with Hamiltonians $H(t) = H_0 + (t/T)V$ from time 0 to $T$, with the convention $\cU_0=\one$. 
In the limit $T\rightarrow \infty$ this corresponds to adiabatic evolution and, in the absence of level crossings (which is the case for this model) we have $\lim_{T\rightarrow \infty}\cU_T= \cU_{\rm opt}(\epsilon=0) \equiv \cU_{\rm opt}$ (formally defined in Cor.~\ref{cor:Uopt_eps0} in Appendix~\ref{app:optimal_unitary}).~\footnote{We do not consider the true optimal unitary $\cU_{\rm opt}(\epsilon\neq 0)$ as this can be very hard to implement in practice, even approximately.} 
For each $\cU_T$, we numerically compute the largest value of the cutoff $w_l$ that satisfies Eq~\eqref{eq:Wlconditionmain} of Thm.~\ref{thm:main}. To this end, we compute all possible $2^{\mathfrak n} \times 2^{\mathfrak n}$ transition probabilities $P(\varepsilon_{1,n}|\varepsilon_{0,m})$ and determine $P(w)$ as in Eq.~\eqref{eq:TTdist}. 
While it might be possible to do this more efficiently in this model, obtaining the largest value of $w_l$ that satisfies Eq~\eqref{eq:Wlconditionmain} for general models might require simulating the entire scheme.

In the left panel of Fig.~\ref{fig:EpsilonAndTmax} we plot this value for $\mathfrak{n}=6$ against $\epsilon$ and for different values of $T$, and in the right panel against $T$ and for different values of $\epsilon$. It can be observed that whereas reducing error tolerance decreases the (largest) work cutoff for all interpolation times, the cutoff for the optimal unitary does not have a strong dependence on $\epsilon$. We also observe that increasing the interpolation time in general increases the work cutoff.  

\begin{figure}[htb]
    \centering
    \includegraphics[scale=0.6]{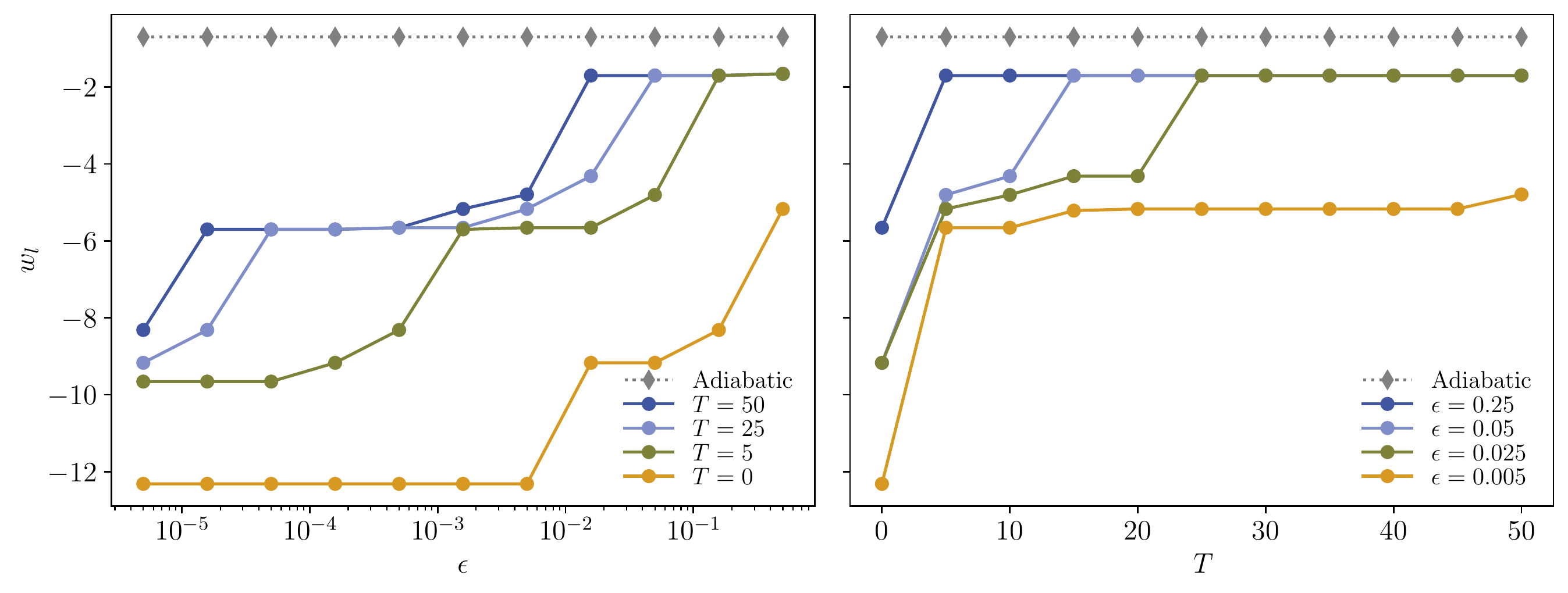}
    \caption{Largest work cutoff $w_{l}$ that satisfies Eq.~\eqref{eq:Wlconditionmain} for each $\cU$, given an initial Hamiltonian $H_0 = \sum_{j=1}^\mathfrak{n} {\rm Z}_j$ and the perturbation $V = -(1/2) \sum_{j=1}^\mathfrak{n} {\rm X}_j {\rm X}_{j+1}$, for an $\mathfrak{n}=6$ qubit system at  $\beta =1$. On the left, we plot $w_l$  as a function of the error $\epsilon$, for when $\cU$ is i) adiabatic evolution, %unitary $\cU_{\rm opt}$ defined in Cor.~\ref{cor:Uopt_eps0},
    and ii) the evolution under the time dependent Hamiltonian $H(t) =  H_0 + (t/T) V $ for different choices in $T$. On the right, we consider the same scenario but plot $w_l$ as a function of $T$ for different choices in $\epsilon$. 
    }
    \label{fig:EpsilonAndTmax}
\end{figure}

In Fig.~\ref{fig:WorkDistributions} we plot work distributions for different choices of the unitaries along with the associated work cutoffs. We observe that the optimal unitary results in a most peaked distribution with largest mean, and distributions tend to get narrower for longer time evolutions and their means increase. These observations are consistent with those in Fig.~\ref{fig:EpsilonAndTmax}. Note that the work cutoffs are at the far tails of the distributions, which is easier to appreciate in the inset. This is due to the importance of rare events~\cite{Jarzynski2006}.

Finally, in Fig.~\ref{fig:WorkCutScaling} we plot the work cutoff against system size for $\epsilon=0.005$ and for  interpolation times $T=0$, $T=2$, and interpolation times linear in system size, $T=\mathfrak{n}$ and $T=5\mathfrak{n}$. We observe that even a time evolution of modest duration can significantly increase the work cutoff.

\begin{figure}[htb]
    \centering
    \includegraphics[scale=0.6]{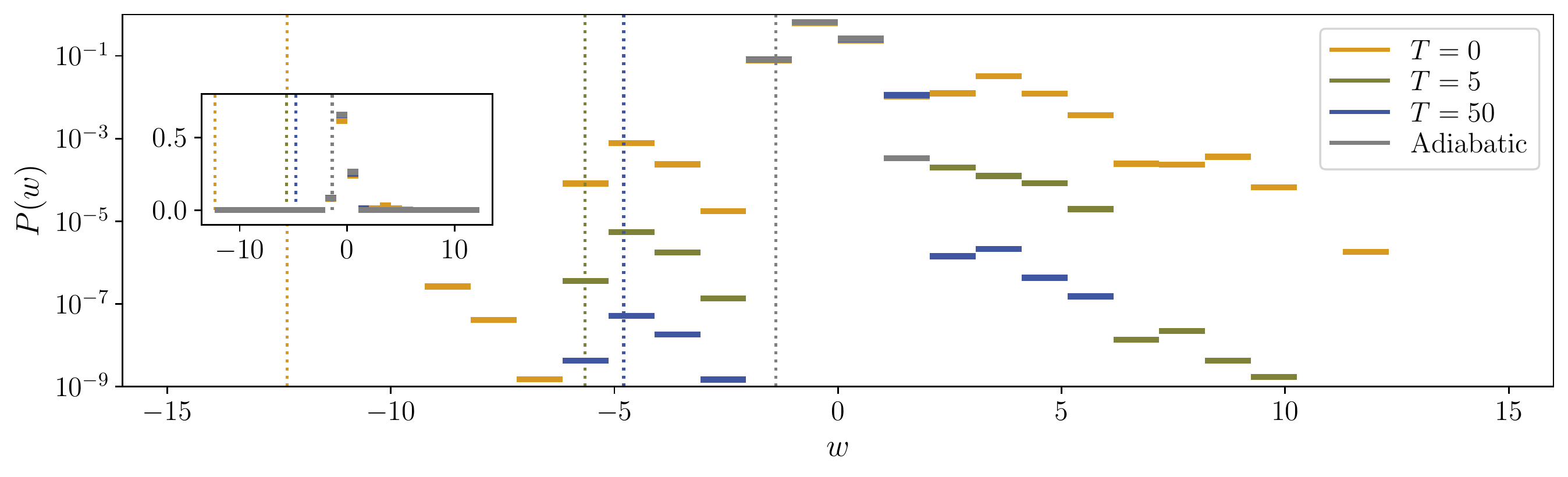}
    \caption{The binned work distributions for the scenario described in Fig.~\ref{fig:EpsilonAndTmax} (i.e., $H_0 = \sum_{j=1}^\mathfrak{n} {\rm Z}_j$ and $V = -(1/2) \sum_{j=1}^\mathfrak{n} {\rm X}_j {\rm X}_{j+1}$ for an $\mathfrak{n}=6$ qubit system at $\beta =1$ and $\epsilon = 0.005$) for the adiabatic evolution and dynamical evolution for different choices in $T$.
    The main plot and the inset show the same data using a logarithmic and a linear scale for the vertical axis, respectively. The vertical dotted lines indicate the corresponding $w_l$ value for each case.
    }
    \label{fig:WorkDistributions}
\end{figure}

\begin{figure}[htb]
    \centering
    \includegraphics[scale=0.35]{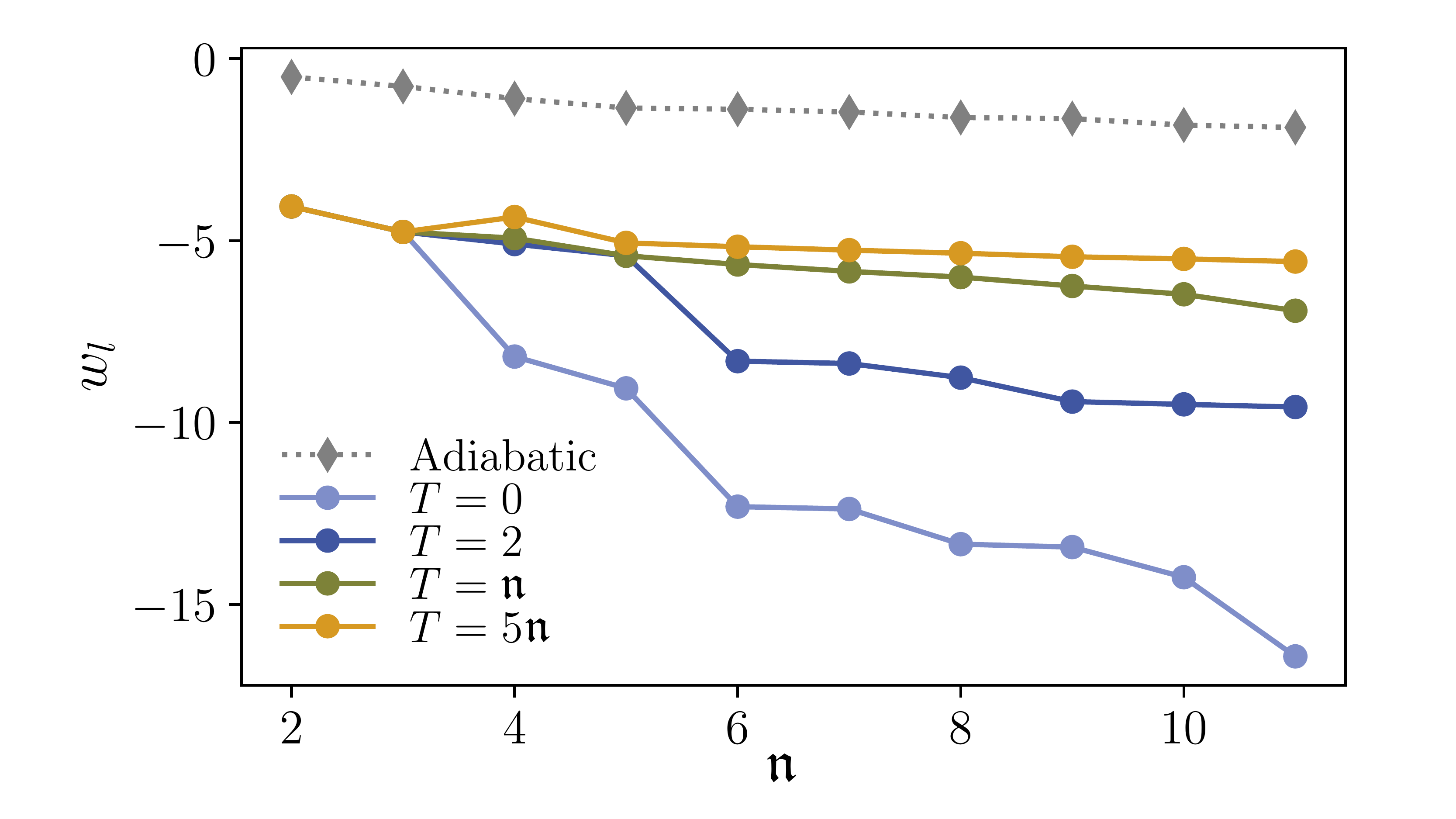}
    \caption{Largest work cutoff $w_l$ for each $\cU$ for the scenario described in Fig.~\ref{fig:EpsilonAndTmax} (i.e. $H_0 = \sum_{j=1}^\mathfrak{n} {\rm Z}_j$ and $V = -(1/2) \sum_{j=1}^\mathfrak{n} {\rm X}_j {\rm X}_{j+1}$ at $\beta =1$ and $\epsilon = 0.005$) but with varying system size $\mathfrak{n}$. We compare the cases where the total evolution time $T$ is fixed, i.e. to $T = 0$ (light blue) or $T = 2$ (dark blue), and when it is allowed to scale linearly with $\mathfrak{n}$, i.e. $T = \mathfrak{n}$ (green) or $T = 5 \mathfrak{n}$ (yellow).}
    \label{fig:WorkCutScaling}
\end{figure}

%%%%%%%%%%%%%%%%%%%%%%%%%%%%%%%%%%%%%%%%%%%%%

%%%%%%%%%%%%%%%%%%%%%%%%%%%%%%%%%%%%%%%%%%%%%%
%%%%%%%%%%%%%%%%%%%%%%%%%%%%%%%%%%%%%%%%%%%%%%
\section{Discussions and open problems}
\label{sec:conclusions}

Using fluctuation theorems, we provided a quantum algorithm to prepare a purification of the thermal state of a quantum system $H_1$ at inverse temperature $\beta \ge 0$, starting from a purification of the thermal state of a quantum system $H_0$. The dominant factor in  the complexity of our quantum algorithm is $e^{\frac \beta 2 (\Delta\! A- w_l)}$, where $\Delta \! A$ is the free-energy difference and $w_l$ is a work cutoff. 
This result is a significant improvement of prior quantum algorithms~\cite{PW09,CS16}. The work cutoff also depends on a non-equilibrium unitary $\cU$ that can be arbitrary. 
We showed that certain $\cU$'s allow for an improvement in $w_l$ and the runtime of our algorithm.
We also obtained suitable choices of $w_l$ for large classes of Hamiltonians, including general Hamiltonians, commuting Hamiltonians, and the local spin Hamiltonians. In all these cases we found a lower bound of $w_l$  that, for $\cU=\one$, depends on the strength of the perturbation $V$.
These results are especially useful in the regime where $\beta \|V\|  \ll 1$.

In the construction of our algorithm we produced 
an approximation of the operator $e^{-\beta W/2}$
as a linear combination of unitaries, each corresponding to an evolution with $W$ for some time [Eq.~\eqref{eq:WorkOp}]. This approximation
results from a Fourier analysis, which has been proven useful in  related contexts (cf.~\cite{CKS17}).
Our approximation allowed us to take advantage of the cutoff $w_l$, which is important to reduce complexity. 

Other useful approximations of $e^{-\beta W/2}$ for quantum algorithms might be obtained using, for example, Chebyshev polynomials~\cite{CKS17,CSS18,LC19}. In contrast with approximations based on time evolutions, these polynomials can be constructed and implemented exactly using quantum walks~\cite{CKS17,LC19}. However, these Chebyshev approximations 
might not help when we seek the approximation to be accurate on a subspace only. For example, 
following the method in Ref.~\cite{CKS17} to approximate the action of $e^{-\beta W/2}$ produces a Chebyshev approximation $\sum_j \alpha_j \cT_j(W/|W|)$, where $\alpha_j \in \mathbb C$,
$\cT_j$ is the $j$-th Chebyshev polynomial of the first kind, and $|W|$ is an upper bound of $\|W\|$ that is obtained
from the presentation of $W$. 
The $L_1$ norm of this approximation, given by $\sum_j |\alpha_j|$, is exponential in $\beta \|W\|=\beta(\|H_1\|+\|H_0\|)$ in the worst case, 
even when we require the approximation to be accurate in the subspace where $w \ge w_l$ only. 
This is in sharp contrast with Eq.~\eqref{eq:L1fourier}, which is exponential in $-\beta w_l$.  Finding a suitable Chebyshev approximation that is accurate in the subspace 
and whose $L_1$ norm exponential in $-\beta w_l$ remains as an open problem. (A related problem was encountered in Ref.~\cite{LowEnergySim} for simulating quantum dynamics on a subspace.)

Last, we emphasize that our quantum algorithm considers the thermal-state preparation problem only. Preparing thermal states suffices to compute thermal properties accurately. However, it is possible to build other quantum algorithms that compute thermal properties directly without preparing the thermal state. Such techniques can also be constructed from fluctuation
theorems and are left for future work.

%%%%%%%%%%%%%%%%%%%%%%%%%%%%%%%%%%%%%%%%%%%%%%
%%%%%%%%%%%%%%%%%%%%%%%%%%%%%%%%%%%%%%%%%%%%%%

\section{Acknowledgements}
RS thanks Augusto Roncaglia and YS thanks Carleton Coffrin for  discussions. 
This material was supported by the U.S. Department of Energy, Office of Science, High-Energy Physics and
Office of Advanced Scientific Computing Research, under the Accelerated Research in Quantum Computing
(ARQC) program, and by the U.S. Department of Energy, Office of Science, National Quantum Information Science Research Centers, Quantum Science Center. This material was also supported by the LDRD program at Los Alamos National Laboratory. Los Alamos National Laboratory is managed by Triad National Security, LLC, for the National Nuclear Security Administration of the U.S. Department of Energy under Contract No. 89233218CNA000001

\newpage
%%%%%%%%%%%%%%%%%%%%%%%%%%%%%%%%%%%%%%%%%%%%%%%%%%%%%%%%%%

%%%%%%%%%%%%%%%%%%%%%%%%%%%%%%%%%%%%%%%%%%%%%%%%%%%%%%%%%%
%%%%%%%%%%%%%%%%%%%%%%%%%%%%%%%%%%%%%%%%%%%%%%%%%%%%%%%%%%
\onecolumn\newpage
\appendix

%%%%%%%%%%%%%%%%%%%%%%%%%%%%%%%%%%%%%%%%%%%%%%%%%%%%%%%%%%%%%%%%%%%%%%%%%
\section{Gate complexity of $U$}
\label{app:Ucomplexity}

Our quantum algorithm uses the unitary $U$, which corresponds to the time-evolution operator for time $-\delta \beta/2$ of the Hamiltonian 
$H_1 \otimes \one - \one \otimes H_0^*$, or $V \otimes \one=(H_1-H_0)\otimes \one$ if the Hamiltonians commute.
The parameter $\delta$ is discussed in Sec.~\ref{sec:Fourier}.
The gate complexity of $U$ can be obtained from known results in Hamiltonian simulation and depends on the presentation of the Hamiltonians and the desired approximation error~\cite{Llo96,SOGKL02,BAC07,WBH+10,CW12,BCC+14,BCC+15,LC17,LC19}.  A common setting is one where the Hamiltonians 
are presented as linear combination of unitaries.
The results in quantum signal processing and qubitization in Refs.~\cite{LC17,LC19} apply to this setting and
imply:

\begin{theorem}[Gate complexity of $U$  using quantum signal processing and qubitization]
\label{thm:Ucomplexity}
Let $H_0=\sum_{l=1}^L \alpha_{0,l} U_{0,l}$ and $H_1=\sum_{l=1}^L \alpha_{1,l} U_{1,l}$ be two Hamiltonians acting on $\cH_{\sy}$, where the $U_{0,l}$'s and $U_{1,l}$'s are $m$-qubit unitaries, and $\alpha_{0,l} >0$, $\alpha_{1,l}>0$ for all $1 \le l \le L$. Then, if the quantum algorithm makes $Q'$ uses the time-evolution unitary $U:=e^{i\delta \beta W/2}$, where $W:=H_1 \otimes \one - \one \otimes H_0^*$, $\beta \ge 0$,
and $\delta \le \pi/16$,
we can simulate $U$ with approximation error $\epsilon/Q'$ using $\tilde \cO(L2^m((\alpha_0+\alpha_1)\delta \beta + \log(Q'/\epsilon)))$ two-qubit gates, where $\alpha_0=\sum_l \alpha_{0,l}$ and $\alpha_1=\sum_l \alpha_{1,l}$. 
In addition, if $V=H_1-H_0=\sum_l \alpha_{V,l} V_l$, where the $V_l$ are also $m-$qubit unitaries, $\alpha_{V,l}>0$, and $[H_0,H_1]=0$, we can simulate $U$ with approximation error $\epsilon/Q'$ using
 $\tilde \cO(L2^m(\alpha_V \delta \beta + \log(Q/\epsilon)))$ two-qubit gates, where $\alpha_V=\sum_l \alpha_{V,l}$. The $\tilde \cO$ notation hides a factor that is logarithmic in $L$. The method implements a unitary
 that acts on the space $\cH_{\rm anc} \otimes \cH_{\sy} \otimes \cH_{\sy}$, where $\cH_{\rm anc}$ is an ancillary space of $\mathfrak m'=\cO(\log(L))$ qubits.
\end{theorem}

\begin{proof}
Consider a block-encoding for a Hamiltonian $H'$ acting on $\cH_{\sy}\otimes\cH_{\sy} =\mathbb C^{2N}$, $N=2^{\mathfrak{n}}$, where $H'=\bra G U_{H'} \ket G$. The operation $U_{H'}$ is unitary and acts on $2\mathfrak{n}+\tilde {\mathfrak{m}}$ qubits,  $\ket G =U_G \ket 0$ is some fiducial state of the $\tilde {\mathfrak{m}}$ ancilla qubits, and $U_G$ is a unitary acting on the ancilla space $\mathbb C^{2^{\tilde{\mathfrak{m}}}}$. Theorem 1 of Ref.~\cite{LC19} states that $H'$ can be simulated for time $t'$ and error $\epsilon'$
with $\cO(t' + \log(1/\epsilon'))$ uses of the controlled-$U_{H'}$ and controlled-$U_G$, and their inverses, and additional two-qubit gates. 
The dominant gate complexity is the one coming from implementing 
the relevant unitaries. The total number of ancilla qubits is $\mathfrak m'=2+\tilde{\mathfrak{m}}$ where, beyond the  $\tilde{\mathfrak{m}}$ ancilla qubits needed for $\ket G$, one additional ancilla qubit is needed for the controlled operations in quantum signal processing and another ancilla qubit is needed to construct a so-called ``unitary iterate'' from $U_{H'}$. This unitary iterate is basically
a controlled operation that implements $U_{H'}$ or $U_{H'}^\dagger$ depending on the state of an ancilla qubit, followed by a sequence of other unitaries that do not use $U_l$; see  Ref.~\cite{LC19} for details. Repeated ($k$) implementations of the unitary iterate provide a block-encoding of the $k$-th Chebyshev polynomial of $H'$, which is useful to accurately approximate $e^{-it'H'}$ as a linear combination of these polynomials.

Let $H = \sum_{l=1}^{2L} \alpha_l U_l$ be a Hamiltonian presented as a linear combination of $m$-qubit unitaries $U_l$, where $\alpha_l >0$. Our goal is to simulate $H$ for time $t$.
Then, we can rescale the Hamiltonian
to simulate $H'=H/\alpha$ for time $t' = \alpha t$, where $\alpha = \sum_{l=1}^{2L} \alpha_l$, i.e., $e^{-itH}=e^{-it'H'}$. In this setting, the unitary that encodes the Hamiltonian $H'$ is
\begin{align}
\label{eq:UH'}
   U_{H'}= \sum_{l=1}^{2L} \ketbra l \otimes U_l \;,
\end{align}
and the unitary that prepares $\ket G$ performs the map
\begin{align}
  \ket 0 \rightarrow U_G \ket 0= \sum_{l=1}^{2L} \sqrt{\alpha_l/\alpha} \ket l \;.
\end{align}
If we use a binary encoding, the number of ancilla qubits is $\tilde {\mathfrak{m}}=\lceil \log_2(2L) \rceil$, hence  ${\mathfrak{m}'}=\cO(\log(L))$, and the gate complexity of $U_G$ is $\cO(L)$.

Implementing $U_{H'}$ in Eq.~\eqref{eq:UH'} requires implementing each $U_l$ controlled on the state $\ket l$ of the $\tilde{\mathfrak{m}}$ ancilla qubits. If each $U_l$ acts on, at most, $m$ qubits,
then $U_{H'}$ requires $\cO(L2^m \log (L))$ two-qubit gates. The results in Ref.~\cite{LC19} imply that the overall number of two-qubit gates
to simulate $H$ for time $t$ and approximation error $\epsilon'$ is
\begin{align}
\label{eq:Ucomplexity}
{\tilde \cO}(L 2^m (\alpha t + \log(1/\epsilon'))  ) \;,
\end{align}
where we dropped the logarithmic factor in $L$.

Theorem~\ref{thm:Ucomplexity} follows directly from Eq.~\eqref{eq:Ucomplexity} if we replace $H$, $\alpha$, $t$, and $\epsilon'$ with $H_1 \otimes \one - \one \otimes H_0^*$ (or $H_1-H_0$), $\alpha_0+\alpha_1$, $\delta \beta/2$, and $\epsilon/Q'$.
\end{proof}

In general, we do not have an upper bound of $\alpha_0$
or $\alpha_1$ in terms of the spectral norms of $H_0$
or $H_1$, but in many cases it is simple to show
$\alpha_0 \le c\|H_0\|$ and $\alpha_1 \le c\|H_1\|$,
where $c>0$ is a constant. 
This implies $(\alpha_0+\alpha_1)\delta \beta =\cO((\|H_0\|+\|H_1\|)\beta)$ and the gate complexity of $U$ from Thm.~\ref{thm:Ucomplexity} is
\begin{align}
\label{eq:UcomplexityB}
{\tilde \cO}(L 2^m ((\|H_0\|+\|H_1\|) \delta \beta + \log(Q'/\epsilon))  ) \;.
\end{align}
Moreover, in Thm.~\ref{thm:main}, $Q'$ can be determined 
from the number of amplitude amplification rounds
times the number of uses of $U$ per round. 
Then,  the dominant factor in Eqs.~\eqref{eq:Ucomplexity} and \eqref{eq:UcomplexityB} is expected to be
$\ln(Q'/\epsilon)$, since $Q' >Q$
is exponential in $\beta$. 

Evolving with $H_1 \otimes \one - \one \otimes H_0^*$
can be done in parallel, acting with $H_1$ on one system and 
with $-H_0^*$ on the other. Also, by setting the approximation error to $\epsilon/Q'$, we guarantee that the additional error coming from the simulation of $U$ is, at most, $\epsilon$. In Thm.~\ref{thm:main}, this would imply that the state prepared by our quantum algorithm satisfies $\frac {1}{2}\|\tau_1 -\rho_1 \|_1 \le  \epsilon$ due to the approximation of $U$. 

\vspace{-.3cm}

%%%%%%%%%%%%%%%%%%%%%%%%%%%%%%%%%%%%%%%%%%%%%%%%%%%%%%%%%%%%%%%%%%%%%%%%%%%%%%%%%%%%%%%%%%%
\section{Optimal non-equilibrium unitaries}
\label{app:optimal_unitary}

Our quantum algorithm uses a non-equilibrium unitary $\cU$ that determines
the work distribution $P(w)$ and the work cutoff $w_l$. 
To reduce the query complexity, a good choice for $\cU$ is 
one that, for a given $\epsilon \ge 0$, allows us to maximize $w_l$ in Eq.~\eqref{eq:Wlconditionmain}.
Similarly, we can treat the left hand side of Eq.~\eqref{eq:Wlconditionmain} as a cost function $C(\cU) \ge 0$ to minimize.
Due to the positivity of the terms in the sum, finding the unitary $\cU$ that gives the largest work cutoff $w_l^*$, i.e. the largest $w_l$ for which Eq.~\eqref{eq:Wlconditionmain} is satisfied for a given $\epsilon \ge 0$, is equivalent to finding the unitary $\cU$ that minimizes $C(\cU)$ for a given work cutoff $w_l^*$. We focus on the latter problem.

Instead of Eq.~\eqref{eq:Wlconditionmain}, it will prove more convenient to work with the equivalent condition Eq.~\eqref{eq:WlconditionPrev}. Our goal is then to find the non-equilibrium unitary that minimizes the following cost function 
\begin{align}
C(\cU) &=\sum_{w>-w_l^*} P^\text{rev}(w) \\
\label{eq:costdef}
&=\sum_{m,n:\varepsilon_{0,m}-\varepsilon_{1,n}>-w_l^*} P^{\rm rev}(\varepsilon_{0,m}|\varepsilon_{1,n}) P_1(\varepsilon_{1,n}) 
\end{align}
for the largest work cutoff $w_l^*$. Here,  $P^\text{rev}(w)$ is the work distribution of the reverse two-time measurement scheme discussed in Sec.~\ref{sec:generalHamiltonians},
$P^{\rm rev}(\varepsilon_{0,m}|\varepsilon_{1,m})=|\!\bra{\phi_{0,m}} \cU^\dagger \ket{\phi_{1,n}}\!|^2$ is the transition probability that depends on $\cU$,
and $P_1(\varepsilon_{1,n}) = e^{-\beta \varepsilon_{1,n}}/\cZ_1$ is the Boltzmann weight, which
is independent of $\cU$.
Intuitively, the optimal unitary, which we call $\cU_\text{opt}^\epsilon$, forbids transitions from $\ket{\phi_{0,m}}$ to $\ket{\phi_{1,n}}$ corresponding to $\varepsilon_{0,m}-\varepsilon_{1,n} > -w_l^*$. If this can be achieved for all transitions, i.e. $P^{\rm rev}(\varepsilon_{0,m}|\varepsilon_{1,n})=0$ if $\varepsilon_{0,m}-\varepsilon_{1,n} > -w_l^*$, then $C(\cU_\text{opt}^\epsilon)=0$. However, this is not always possible. The unitary $\cU_\text{opt}^\epsilon$ described below is such that it minimizes the contribution to the cost function due to these transitions.

It is easiest to state the form of $\cU_\text{opt}^\epsilon$ in terms of the action of its inverse $\cU_\text{opt}^{\epsilon \; \dagger}$ on the eigenstates of $H_1$, so we adopt this approach. Recall that our convention is such that $\varepsilon_{0,0}\le \varepsilon_{0,1} \le \ldots \le \varepsilon_{0,N-1}$ and $\varepsilon_{1,0}\le \varepsilon_{1,1} \le \ldots \le \varepsilon_{1,N-1}$. The rough intuition underlying $\cU_\text{opt}^{\epsilon \; \dagger}$ is to pursue a greedy strategy to minimize $C(\mathcal{U})$. Namely, starting with the lowest eigenvalue eigenstate of $H_1$, $\ket{\phi_{1,0}}$, this is mapped to the lowest eigenvalue eigenstate of $H_0$, $\ket{\phi_{0,0}}$, if $\varepsilon_{0,0}-\varepsilon_{1,0} \leq -w_l^*$, without contributing to the cost function. If $\varepsilon_{0,0}-\varepsilon_{1,0}> -w_l^*$, 
the $\ket{\phi_{1,0}}$ state will contribute $P_1(\varepsilon_{1,0})$ to the cost function
``no matter what", and so $\varepsilon_{1,0}$ is mapped to the largest eigenvalue eigenstate of $H_0$. 
This is to leave more ``room'' at the bottom of the spectrum of $H_0$, ensuring that other  eigenstates of $H_1$ with $n \ge 1$ have eigenstates of $H_0$ of lower eigenvalue (shifted by $w_l^*$) available to them, and so are less likely to contribute to the cost function.
This general greedy strategy is then applied iteratively to each of the eigenstates of $H_1$ in increasing order $n=0,1,\ldots$ to obtain $\cU_\text{opt}^{\epsilon \; \dagger}$. More concretely, $\cU_\text{opt}^{\epsilon \; \dagger}$ can be described as follows.

\begin{theorem}
[Optimal non-equilibrium unitary]
\label{thm:Uopt_finite_eps}
Let $H_0$ and $H_1$ be two Hamiltonians acting on $\cH_{\sy}$, $\epsilon\ge 0$, and $w_l^*$ be the largest work cutoff for which Eq.~\eqref{eq:Wlconditionmain} holds upon optimizing over all unitaries $\cU$. Then, an optimal unitary $\cU_\text{opt}^\epsilon$, which maps eigenstates of $H_0$ to eigenstates of $H_1$, can be constructed as follows: \\ 
\vspace{-3mm}
\\
Let $S_0:           =\{0,1,\dots,N-1\}$. \\
For $n=0,1,\dots,N-1$: \\
\phantom{----} if $\varepsilon_{1,n} \ge \varepsilon_{0,\min(S_n)} + w_l^*$:\\
\phantom{----------} set $\cU_\text{opt}^{\epsilon \; \dagger}\ket{\phi_{1,n}}=\ket{\phi_{0,\min(S_n)}}$, \\ 
\phantom{----------} set $S_{n+1} = S_{n}\setminus \min(S_n)$; \\
\phantom{----} else: \\
\phantom{----------} set $\cU_\text{opt}^{\epsilon \; \dagger}\ket{\phi_{1,n}}=\ket{\phi_{0,\max(S_n)}}$,\\
\phantom{----------} set $S_{n+1} = S_{n}\setminus \max(S_n)$;
\end{theorem}

\vspace{0.1cm}
We provide the proof of this theorem below. One immediate implication is:

\begin{corollary}[Optimal non-equilibrium unitary for $\epsilon=0$]
\label{cor:Uopt_eps0}
Let $H_0$ and $H_1$ be two Hamiltonians acting on $\cH_{\sy}$.  Let $\cU_\text{opt}$ be the unitary that maximizes the largest possible work cutoff $w_l$ that satisfies Eq.~\eqref{eq:Wlconditionmain} with $\epsilon=0$, i.e. $\cU_\text{opt} := \cU_\text{opt}^0$.
Equivalently, $\cU_\text{opt}$ is the unitary that maximizes the minimum value of work $w_{\min}$ that can occur with nonzero probability.   It follows from Thm.~\ref{thm:Uopt_finite_eps} that
$\cU_\text{opt}$ maps eigenstates of $H_0$ to eigenstates of $H_1$, while preserving the ascending order of the eigenvalues, i.e. $\cU_\text{opt} \ket{\phi_{0,n}}=\ket{\phi_{1,n}}$ for all $0 \le n \le N-1$.
\end{corollary}

Optimal unitaries for $\epsilon >0$ and $\epsilon=0$ are demonstrated for a system with $N=6$ in Fig.~\ref{fig:optimal_unitaries}(a) and (b), respectively.
These constructions follow Thm.~\ref{thm:Uopt_finite_eps} and Cor.~\ref{cor:Uopt_eps0}. Note that, 
for $\cU_\text{opt}^\epsilon$, the transition probabilities satisfy $P^{\rm rev}(\varepsilon_{0,m}|\varepsilon_{1,n})=1$ if 
$\ket{\phi_{0,m}}=\cU_\text{opt}^{\epsilon \; \dagger} \ket{\phi_{1,n}}$ and $P^{\rm rev}(\varepsilon_{0,m}|\varepsilon_{1,n})=0$ otherwise.
Then, $C(\cU_{\rm opt})$ is simply a sum of Boltzmann weights $P_1(\varepsilon_{1,n})$ for those $n$ 
where $P^{\rm rev}(\varepsilon_{0,m}|\varepsilon_{1,n})=1$ and
$\varepsilon_{0,m}-\varepsilon_{1,n}>-w_l^*$.
Analyzing the constructive proof below we will see that optimal unitaries are in general not unique. The one presented here is chosen for convenience of presentation.

\begin{figure}\centering
\includegraphics[width=\textwidth]{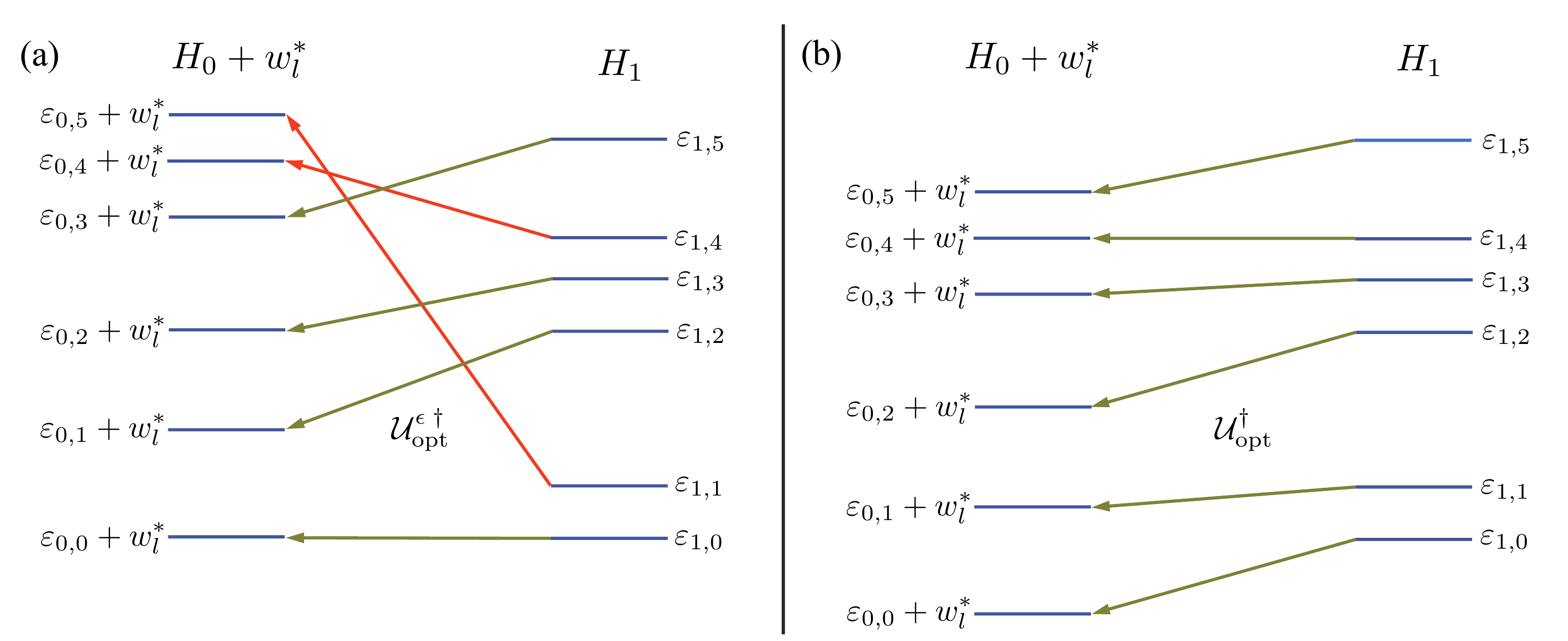}
\caption{Description of optimal unitaries for an $N=6$ dimensional system, where arrows denote the action of $\cU_{\rm opt}^{\epsilon \; \dagger}$ on the corresponding eigenstates of $H_1$ (the spectrum is represented by blue lines). (a) The case of $\epsilon>0$. Starting from the lowest eigenvalue eigenstate of $H_1$, $\ket{\phi_{1,0}}$ is mapped to $\ket{\phi_{0,0}}$, as $\varepsilon_{1,0} \ge \varepsilon_{0,0} + w_l^*$. Then, $\ket{\phi_{1,1}}$ is mapped to the largest eigenvalue eigenstate of $H_0$, $\ket{\phi_{0,N-1}}$, as $\varepsilon_{1,1} < \varepsilon_{0,1} + w_l^*$. This process is iterated, resulting in the transformation shown. The optimal unitary shown here achieves the minimal cost function value $C(\cU_\text{opt}^\epsilon) = P_1(\varepsilon_{1,1})+P_1(\varepsilon_{1,4})>0$. (b) The case of $\epsilon=0$. The cost function in this case is $C(\cU_{\rm opt})=0$. As expected, the largest work cutoff $w_l^*$ in case (b) is smaller than that of case (a).}
\label{fig:optimal_unitaries}
\end{figure}

We now prove that the unitary $\cU_{\rm opt}^\epsilon$ described in Thm.~\ref{thm:Uopt_finite_eps} minimizes the value of Eq.~\eqref{eq:costdef} for a given $w_l^*$.
The proof follows two steps: first, we argue that the optimal unitary is a permutation that maps eigenstates of $H_0$ to eigenstates of $H_1$, and then we show that the unitary of Thm.~\ref{thm:Uopt_finite_eps} is optimal and minimizes Eq.~\eqref{eq:costdef} among all such permutations.

For the first step,
unitarity implies that, for all $\cU$, 
\begin{align}
\sum_m P^{\rm rev}(\varepsilon_{0,m}|\varepsilon_{1,n}) &=\sum_n P^{\rm rev}(\varepsilon_{0,m}|\varepsilon_{1,n}) \\
&=1 \;.
\end{align}
Thus $P^{\rm rev}(\varepsilon_{0,m}|\varepsilon_{1,n})$ gives rise to an $N \times N$ doubly stochastic matrix. Birkhoff-von Neumann theorem states that the class of $N\times N$ doubly stochastic matrices is a convex polytope $\mathcal{B}_{N}$, also known as the Birkhoff polytope, and is the convex hull of the set of $N\times N$ permutation matrices. Furthermore the vertices $\mathcal{B}_{N}$ are precisely the permutation matrices. The cost function $C(\cU)$ is linear in the entries of this matrix and the coefficients are positive. By the fundamental theorem of linear programming, the minimum of $C(\cU)$ is attained at the vertices of the Birkhoff polytope. Thus, we only need to find the permutation $\pi_\text{opt} \in \cS_N$ in the symmetric group that determines $\cU^{\epsilon}_{\rm opt}$ and minimizes $C(\cU)$~\footnote{Not every doubly stochastic matrix can be obtained from an underlying unitary process. 
However, since the minimum of $C(\cU)$ is achieved by a basis transformation between $H_0$ and $H_1$, which is unitary, we do not need to worry about this fact.}. 
For simplicity, we use a convention where $\pi_\text{opt}$ corresponds to $\cU_\text{opt}^{\epsilon \; \dagger}$, i.e., $\cU_\text{opt}^{\epsilon \, \dagger}$ maps the $m$-th eigenstate of $H_1$ to the $\pi_\text{opt}(m)$-th eigenstate of $H_0$ or, equivalently, $\cU_\text{opt}^\epsilon$ maps the $m$-th eigenstate of $H_0$ to the $\pi^{-1}_\text{opt}(m)$-th eigenstate of $H_1$.

\medskip

For the second step, we describe a procedure that iteratively transforms any permutation $\pi \in \cS_N$ to $\pi_\text{opt}$, and show that the cost function does not increase in the procedure. (With some abuse of notation, we write $C(\pi^{-1})$ for the cost function in Eq.~\eqref{eq:costdef} when $\cU$ corresponds to a permutation between eigenstates given by some $\pi^{-1}$, that is,  $\cU \ket{\phi_{0,m}}=\ket{\phi_{1,\pi^{-1}(m)}}$ for all $m$, or $\cU \ket{\phi_{0,\pi(n)}}=\ket{\phi_{1,n}}$ for all $n$.)
As this procedure works for \textit{any} permutation $\pi$, it follows that $C(\pi^{-1}_\text{opt})\le C(\pi^{-1})$ for all $\pi \in \cS_N$, and therefore 
the unitary $\cU_{\rm opt}^\epsilon$ determined by $\pi^{-1}_\text{opt}$ is optimal.

Let then $\pi_\text{opt} \in \cS_N$ be the permutation that corresponds to the unitary $\cU_{\rm opt}^{\epsilon \; \dagger}$ given by Thm.~\ref{thm:Uopt_finite_eps}.
The core idea of the procedure is to transform $\pi$ into $\pi_\text{opt}$ one transition at a time by performing a ``swap''.
The procedure is as follows:

\medskip
\noindent
% Let $\tilde \pi=\pi$. \\
% For $n=0,1,\dots,N-1$: \\
% \phantom{----} if $\tilde \pi(n) \ne \pi_{\rm opt}(n)$ \\
% \phantom{----------} set $m= \pi_{\rm opt}(n)\;,$ \\
% \phantom{----------} set $m'=\tilde \pi(n)\;,$ \\
% \phantom{----------} set $n'=\tilde \pi^{-1}(m)\;,$ \\
% \phantom{----------} set $\tilde \pi(n) = m\;,$ \\
% \phantom{----------} set $\tilde \pi(n') = m'\; ;$
For $n=0,1,\dots,N-2$: \\
\phantom{----} if $ \pi(n) \ne \pi_{\rm opt}(n)$ \\
\phantom{----------} set $m= \pi_{\rm opt}(n)\;,$ \\
\phantom{----------} set $m'= \pi(n)\;,$ \\
\phantom{----------} set $n'= \pi^{-1}(m)\;,$ \\
\phantom{----------} set $ \pi(n) = m\;,$ \\
\phantom{----------} set $ \pi(n') = m'\; ;$
\medskip

Note that the swap at the $n$-th step
does not change $\pi(0)=\pi_{\rm opt}(0),\pi(1)=\pi_{\rm opt}(1),\ldots,\pi(n-1)=\pi_{\rm opt}(n-1)$.
Then, at the end, the updated $\pi$  coincides with $\pi_{\rm opt}$. 
Below, we will argue that the cost function
does not increase after each step.

Let $\tilde \varepsilon_{0,m} = \varepsilon_{0,m}+w_l^*$ for all $m$. 
Prior to each swap in the procedure, the contributions to $C( \pi^{-1})$ come only from terms where ($m'=\pi(n)$)
\begin{align}
\varepsilon_{1,n} < \tilde \varepsilon_{0,m'} \, ,
\end{align}
with the magnitude of each such contribution given by $P_1(\varepsilon_{1,n})$.  
The change in the cost function
after each swap in the procedure only depends on the eigenvalues $\varepsilon_{0,m}$ and  $\varepsilon_{0,m'}$ of $H_0$ and 
$\varepsilon_{1,n}$ and  $\varepsilon_{1,n'}$ of $H_1$.
By construction, $n'>n$ and then $\varepsilon_{1,n'}\ge \varepsilon_{1,n}$.
With regard to $\tilde\varepsilon_{0,m}$,
there are two cases to consider, as shown in Fig.~\ref{fig:cases}.
In case 1, $\tilde \varepsilon_{0,m} \le \varepsilon_{1,n}$ and therefore $\pi_\text{opt}(n)=\min(S_n)$ according to Thm.~\ref{thm:Uopt_finite_eps}. 
As a result, $m'>m$, $ \varepsilon_{0,m'}\ge  \varepsilon_{0,m}$ and $\tilde \varepsilon_{0,m'}$ can take values in three intervals, as shown in the left panel of Fig.~\ref{fig:cases}. 
In case 2, $\tilde \varepsilon_{0,m} > \varepsilon_{1,n}$ and therefore $\pi_\text{opt}(n)=\max(S_n)$ according to Thm.~\ref{thm:Uopt_finite_eps}. 
As a result, $m'<m$, $\varepsilon_{0,m'} \le \varepsilon_{0,m}$, and $\varepsilon_{1,n'}$ can take values in three intervals, as shown in the right panel of Fig.~\ref{fig:cases}. For these total of six cases, 
the contributions to the cost function before and after each update are shown in the figure. 
The change in the cost function for each case is non-positive.

\newpage

\begin{figure}\centering
\includegraphics[width=\textwidth]{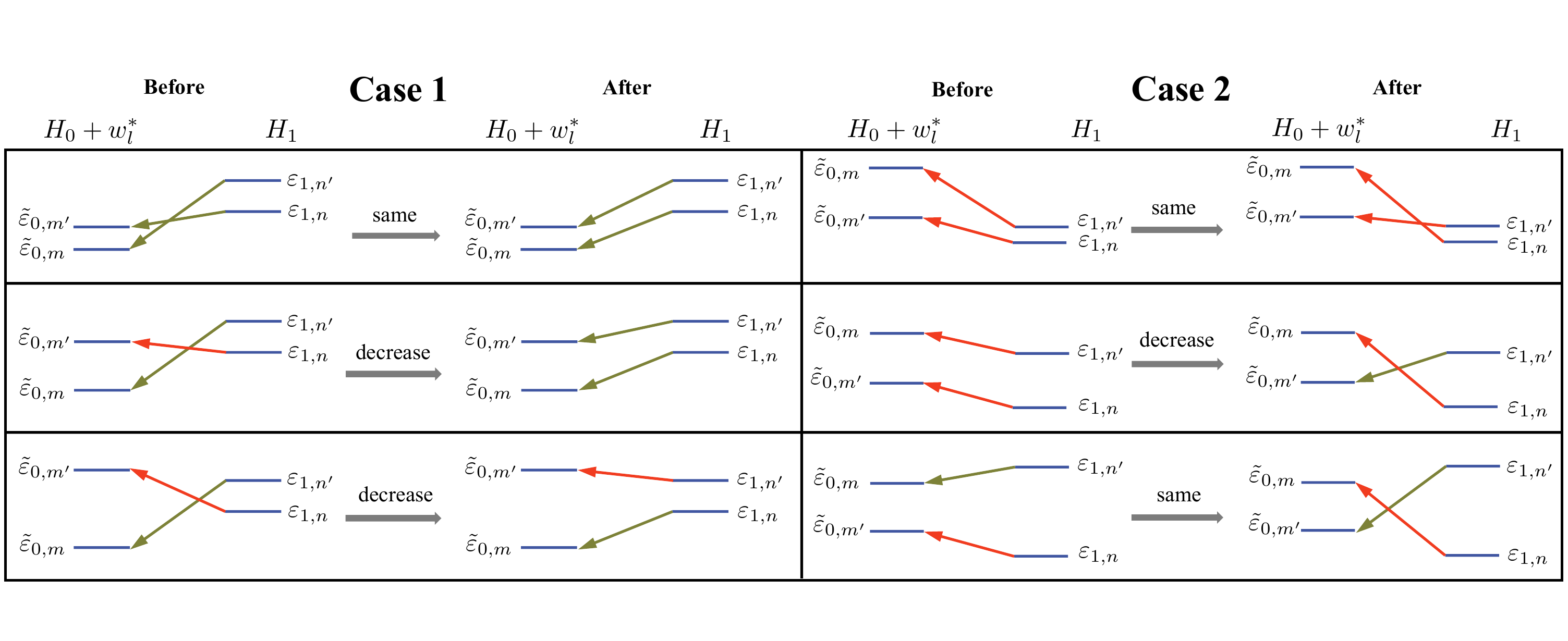}
\caption{Change in cost function $C(\pi^{-1})$ after each swap of the procedure. At each step, there are six possible cases that need to be considered. Case 1 (left) corresponds to $\tilde \varepsilon_{0,m} \le \varepsilon_{1,n}$  and case 2 (right) corresponds to $\tilde \varepsilon_{0,m} > \varepsilon_{1,n}$. The subcases correspond to the possible orderings of the other energy levels. Red arrows indicate the transitions that contribute to the cost. For example, a red arrow originating at a level labeled $n$ contributes $P_1(\varepsilon_{1,n})$ to the cost function. For the subcase in the lowest panel under Case 1, we used the fact that $P_1(\varepsilon_{1,n})> P_1(\varepsilon_{1,n'})$ for $\varepsilon_{1,n}<\varepsilon_{1,n'}$ to establish that the cost function decreases. In all other panels, it is straightforward to see that the cost is either decreasing or stays the same.}
\label{fig:cases}
\end{figure}

%%%%%%%%%%%%%%%%%%%%%%%%%%%%%%%%%%%%%%%%%%%%%%%%%%%%%%%%%%
\section{Proof of Thm.~\ref{thm:ratioPF}}
\label{app:ratioPF}

For $0\le s\le 1$ we let $H_s:=H_0+sV$,  
$\ket{\psi_{s,p}}$ be the eigenstates of eigenvalue $\varepsilon_{s,p}$, $\cZ_s:=\tr(e^{-\beta H_s})=\sum_p e^{-\beta \varepsilon_{s,p}}$ be the partition function, and
$\rho_s:=e^{-\beta H_s}/\cZ_s$
be the thermal state   of $H_s$.
Then,
\begin{align}
    \partial _s \cZ_s &=-\beta \sum_p (\partial_s \varepsilon_{s,p}) e^{-\beta \varepsilon_{s,p}} \\
    & = -\beta \sum_p (\partial_s \bra{\psi_{s,p}} H_s \ket{\psi_{s,p}}) e^{-\beta \varepsilon_{s,p}} \\
    & = -\beta \sum_p \bra{\psi_{s,p}} \partial_s H_s \ket{\psi_{s,p}} e^{-\beta \varepsilon_{s,p}} \\
    & = -\beta \sum_p \bra{\psi_{s,p}} V \ket{\psi_{s,p}} e^{-\beta \varepsilon_{s,p}} \\
    & = -\beta \cZ_s \tr(V\rho_s) \;.
\end{align}
This implies
\begin{align}
    \ln \left(\frac{\cZ_1} {\cZ_0} \right)& =  \int_0^1 ds \; \partial_s \ln \cZ_s \\
    & =   \int_0^1 ds \; \frac{\partial_s \cZ_s }{\cZ_s} \\
    & =- \beta  \int_0^1 ds \; \tr(V\rho_s)
\end{align}
and then
\begin{align}
|\Delta \! A|&=
    \left| \frac 1 \beta \ln \left(\frac{\cZ_1} {\cZ_0} \right)\right| \\
    & \le   \int_0^1 ds \; |\tr(V\rho_s)| \\
    & \le \|V\| \;.
\end{align}
This implies Eq.~\eqref{eq:ratioPF}.
\qed

%%%%%%%%%%%%%%%%%%%%%%%%%%%%%%%%%%%%%%%%%%%%%
%%%%%%%%%%%%%%%%%%%%%%%%%%%%%%%%%%%%%%%%%%%%%

\section{Proof of Lemma~\ref{lem:lemmafourier}}
\label{app:lemmafourier}

The proof follows three approximation steps. First, we approximate the exponential operator by a convolution using $h(x)=(f \star g)(x)$ and the cutoff $w_l$.  Second, we approximate the convolution by an infinite sum obtaining a Fourier series. Third, we approximate the infinite sum by a finite one. These steps will set $\Delta$ as a function of $\epsilon$, and also $\delta$, $J+1$, and $\omega_{J+1}$ as a function of other problem parameters. Once these are determined, it is simple to prove Eq.~\eqref{eq:L1fourier}. 
Each step is realized by operators $X_1$, $X_2$, and $X$, respectively,
allowing us to use triangle inequality:
\begin{align}
\nonumber
    \|(e^{-\beta W/2}- X\ket {\Psi_0} \| \le & \|(e^{-\beta W/2}-X_1)\ket {\Psi_0}\| +\|(X_1 -X_2) \ket {\Psi_0} \|+ \\
%    \label{eq:trianglefourier}
    &+\|( X_2 - X) \ket {\Psi_0} \|  \;.
\end{align}
 The operators $X_1$, $X_2$, and $X$ commute with $W$
 and have the same eigenstates $\ket{\psi_{m,n}}$. They also 
 depend on $\beta$ and $\epsilon$, but we do not make these dependencies
 explicit for simplicity.

\vspace{0.2cm}

\subsection{Step I: The approximation $X_1$}

We start with the approximation $e^{-x} \approx h(x)=(f \star g)(x)$, where $f(x)$ and $g(x)$ are given in Eqs.~\eqref{eq:fFourier} and~\eqref{eq:gFourier}, respectively, and replace $x \rightarrow \hat x:=\beta (W -w_l)/2$. (The eigenvalues of $\hat x$ are non-negative when acting on eigenstates $\ket{\psi_{m,n}}$ of eigenvalue $w_{m,n} \ge w_l$.)
Then, the operator in the first approximation step is
[see Eq.~\eqref{eq:gdef}]
\begin{align}
\label{eq:firstapproxFourier}
 X_1:= e^{-\beta w_l/2}   \frac{e^{-1/4} }{\sqrt \pi}  \int_{-\Delta-1/2}^{\infty} dy \; & e^{-(\hat x-y)^2} e^{-y}  \;,
\end{align}
where $\Delta \geq 0$.
Following Eq.~\eqref{eq:herf},
the eigenvalues of $X_1$ are
\begin{align}
\label{eq:firstapproxFouriereigenv}
e^{-\beta w_l/2} h(x_{m,n})& = e^{-\beta w_l/2}   \frac{e^{-1/4} }{\sqrt \pi}  \int_{-\Delta-1/2}^\infty dy \;  e^{-(x_{m,n}-y)^2} e^{-y}
  \\ 
    & =e^{-\beta w_{m,n}/2} \frac{1+ {\rm Erf}(\Delta+x_{m,n})}2 \;,
\end{align}
where $x_{m,n}:=\beta(w_{m,n}-w_l)/2$.
This implies
\begin{align}
(e^{-\beta W} - X_1)\ket{\psi_{m,n}}&= e^{-\beta w_{m,n}/2} \left(1 - \frac{1 + {\rm Erf}(\Delta + x_{m,n})}{2} \right)\ket{\psi_{m,n}}\\
&= e^{-\beta w_{m,n}/2} \frac{{\rm Erfc} (\Delta + x_{m,n})}{2} \ket{\psi_{m,n}}
\end{align}
where ${\rm Erfc} (z):= \frac{2}{\sqrt{\pi}} \int^{\infty}_z dy e^{-y^2}= 1 - {\rm Erf(z)}$.
Note that ${\rm Erfc}(z) \leq e^{-z^2}$ for $z\geq 0$ and hence
\begin{equation}\label{eq:Bound1}
\|(e^{-\beta W}-X_1)\ket {\psi_{m,n}}\| \leq  {\begin{cases*}
    e^{- \beta w_{m,n}/2} e^{-\Delta^2}/2 & if $w_{m,n} \geq w_l$, \\
    e^{- \beta w_{m,n}/2} & if $w_{m,n} < w_l$.
    \end{cases*}}
\end{equation}

%%%%%%%%%%%%%%%%%%%%%%%%%%%%%%%%%%%%%%%%%%%%%%%
\subsection{Step II: The approximation $X_2$}

For the second approximation step
we rewrite  $X_1$ using the Fourier transform as
\begin{align}
 X_1=   e^{-\beta w_l/2} \frac 1 {\sqrt{2\pi}} \int_{-\infty}^\infty d \omega \; H(\omega) e^{i \omega \hat x}\;,
\end{align}
where $H(\omega)=\sqrt{2 \pi}F(\omega)G(\omega)$ is the Fourier transform of $h(x)$, and $F(\omega)$ and $G(\omega)$ are given in Eqs.~\eqref{eq:FFouriertransform} and~\eqref{eq:GFouriertransform}, respectively.
In the second approximation, we replace the integral by the sum, i.e., 
\begin{align}
\label{eq:Fouriersecondapprox}
X_2:=  e^{-\beta w_l/2} \frac{\delta}{\sqrt{2\pi}} \sum_{j=-\infty}^\infty   H(\omega_j) e^{i\omega_j \hat x} \;,
\end{align}
where $\omega_j = j \delta$, and $\delta >0$.
The eigenvalues of $X_2$ are 
\begin{align}
\label{eq:FouriersecondapproxW}
  e^{-\beta w_l/2} \frac{\delta}{\sqrt{2\pi}} \sum_{j=-\infty}^\infty   H(\omega_j) e^{i \omega_j x_{m,n} } \;.
\end{align}
This Fourier series is a periodic function in $x_{m,n}$, where the period is $2\pi/\delta$.
Hence, it is possible that the error when approximating the exponential operator is undesirably large for some large values of $x_{m,n}$.
However, we are interested in a useful approximation to the exponential operator {\em only} in the domain where the eigenvalues of $W$ are $w_{m,n} \in [w_{\min} , w_{\max}]$, which can be obtained by making the period $2\pi/\delta$ sufficiently large, as we discuss below.

For an $L^1$ function $r(x)$ with Fourier transform $R(\omega)$, the Poisson summation formula implies 
\begin{align}
 \frac{\delta} {\sqrt{2\pi}}   \sum_{j=-\infty}^\infty R(\omega_j) e^{i a \omega_j}= \sum_{k=-\infty}^\infty r(a+ k 2 \pi/\delta) \;.
\end{align}
Then, using Eq.~\eqref{eq:herf} and the Poisson summation formula above, we rewrite Eq.~\eqref{eq:FouriersecondapproxW} as
\begin{align}
\label{eq:secondapproxFourier}
 e^{-\beta w_l/2}  &\sum_{k=-\infty}^\infty  h(x_{m,n}+k2 \pi /\delta) = e^{-\beta w_{m,n}/2}\sum_{k=-\infty}^\infty e^{-k2 \pi /\delta} \frac {1+{\rm Erf}(\Delta + x_{m,n} +k2 \pi /\delta)} 2 \;.
\end{align}
In the left hand side of Eq.~\eqref{eq:secondapproxFourier}, the term where $k=0$ is $e^{-\beta w_l/2} h(x_{m,n})$ and coincides with Eq.~\eqref{eq:firstapproxFouriereigenv}.
This implies 
\begin{align}
(X_1 - X_2)\ket{\psi_{m,n}}=  e^{-\beta w_{m,n}/2}\sum_{k \ne 0}  e^{-k2 \pi /\delta} \frac {1+{\rm Erf}(\Delta + x_{m,n} +k2 \pi /\delta)} 2 \ket{\psi_{m,n}} \;.
\end{align}
If $2 \pi/\delta \ge 1$ and for $k \ge 1$, we obtain, in general, 
\begin{align}
    \sum_{k \ge 1} e^{-k2 \pi /\delta} \frac {1+{\rm Erf}(\Delta + x_{m,n} +k2 \pi /\delta)} 2 &\le \sum_{k \ge 1} e^{-k2 \pi /\delta} \\
    & = \frac{e^{-2\pi/\delta}}{1-e^{-2\pi/\delta}} \\
    \label{eq:SecondFourierlh}
     & \le 2{e^{-2\pi/\delta}}\;.
%    & \le 2{e^{-2\Delta}}\;, \\
   % & \le 2\;,
\end{align}
For $k \le -1$ and $w_{m,n} \ge w_l$ ($x_{m,n}\ge 0$) we could obtain a large relative error -- as compared to that in  Eq.~\eqref{eq:hconstraint}-- if there exists a negative $k$
such that $\Delta + x_{m,n} +k2 \pi /\delta \ge 0$ (or
${\rm Erf} \approx 1$) and $e^{-k 2\pi/\delta}$ is large.
This is consistent with the fact that our approximation is a periodic function of $x_{m,n}$.
To fix this we demand that $2 \pi/\delta \ge 2(\Delta + x_{m,n})=2\Delta + \beta(w_{m,n}-w_l)$
for all $m,n$, and thus suffices if $2 \pi/\delta \ge 2\Delta + \beta(w_{\max}-w_l)$.
Similarly, for $k \le -1$ and $w_{m,n} \ge w_l$ ($x_{m,n} \ge 0$),
\begin{align}
\Delta + x_{m,n} +k2 \pi /\delta & = \Delta + x_{m,n} +k \pi /\delta
   +k \pi /\delta \\
   & \le \Delta + x_{m,n} + k (\Delta + x_{m,n}) +k \pi /\delta \\
   & = (1+k) (\Delta + x_{m,n}) +k \pi /\delta \\
   & \le k \pi /\delta \;.
\end{align}
Then, for all $w_{m,n} \ge w_l$,
\begin{align}
   \sum_{k \le -1} e^{-k2 \pi /\delta} \frac {1+{\rm Erf}(\Delta + x_{m,n} +k2 \pi /\delta)} 2 & \le \sum_{k \le -1} e^{-k2 \pi /\delta} \frac {1+{\rm Erf}(k \pi /\delta)} 2 \\
   & \le  \sum_{k \ge 1} e^{k2 \pi /\delta} \frac{e^{-(k \pi/\delta)^2} } 2\\
   & =\frac  e 2\sum_{k \ge 1} e^{-(k \pi/\delta -1)^2} 
   \\
   & \le \frac e 2 \sum_{k \ge 1} e^{-(k \pi/(2\delta) )^2} \\
   & \le \frac e 2\sum_{k \ge 1} e^{-k( \pi/(2\delta) )^2}\\
   & = \frac e 2 \frac{e^{-(\pi/(2\delta))^2}} {1-e^{-(\pi/(2\delta))^2}} \\
   \label{eq:SecondFourierh}
   & \le e {e^{-(\pi/(2\delta))^2}} \;,
\end{align}
where we also assumed $ \pi/(2\delta) \ge 1$ and used $1+{\rm Erf}(z)= {\rm Erfc}(|z|) \le e^{-z^2}$ for $z \le 0$.

For $w_{m,n} < w_l$ ($x_{m,n}<0$), we have
\begin{align}
     \sum_{k \le -1} e^{-k2 \pi /\delta} \frac {1+{\rm Erf}(\Delta + x_{m,n} +k2 \pi /\delta)} 2 & \le 
    \sum_{k \ge 1} e^{k2 \pi /\delta} \frac {1+{\rm Erf}(\Delta - k2 \pi /\delta)} 2 \\
  & \le  \sum_{k \ge 1} e^{k2 \pi /\delta} \frac{e^{-(k 2 \pi/\delta-\Delta)^2}} 2 \\
  & = \frac{e^{1/4+\Delta} } 2\sum_{k \ge 1}  e^{-(k 2 \pi/\delta-\Delta-1/2)^2}\\
  & \le  \frac{e^{1/4+\Delta}}2   \sum_{k \ge 1}  e^{-(k 2 \Delta -\Delta-1/2)^2} \\
& =  \frac{e^{1/4+\Delta}}2  \sum_{k \ge 1}  e^{-(k  \Delta -1/2)^2} \\
    & =  \frac 1 2\sum_{k \ge 1}  e^{-k^2 \Delta^2 +(k+1)  \Delta}\\
    & \le \frac 1 2 \sum_{k \ge 1}  e^{-k^2 \Delta^2 /2} \\
    & \le \frac 1 2\sum_{k \ge 1}  e^{-k \Delta^2 /2} \\
    & =\frac 1 2 \frac{e^{-\Delta^2/2}} {1- e^{-\Delta^2/2}} \\
    \label{eq:SecondFourierl}
    & \le  e^{-\Delta^2/2} \;,
\end{align}
where we used $2 \pi/\delta \ge \Delta$ and $\Delta \ge 4$, so that $k^2\Delta^2 -(k+1) \Delta \ge k^2 \Delta^2/2$ for $k \ge 1$, and $1-e^{-\Delta^2/2} \ge 1/2$.
Hence
\begin{equation}\label{eq:Bound2}
\|(X_1- X_2) \ket {\psi_{m,n}}\|\leq {\begin{cases*}
    e^{-\beta w_{m,n}/2}(2e^{-2\pi/\delta} + ee^{-(\pi/(2\delta))^2}) & if $w_{m,n} \geq w_l$ ,\\
    e^{-\beta w_{m,n}/2}(2e^{-2\pi/\delta} + e^{-\Delta^2/2}) & if $w_{m,n} < w_l$.
    \end{cases*}}
\end{equation}

\subsection{Step III: The approximation $X$}

Last, we place the corresponding upper bound in $J$,
and the resulting approximation is $X$ in Eq.~\eqref{eq:Fourierapproximationoperator}.
From Eq.~\eqref{eq:FouriersecondapproxW}, we have
\begin{align}
(X_2 - X)\ket{\psi_{m,n}}=  e^{-\beta w_l/2} \frac{\delta}{\sqrt{2\pi}}   \sum_{|j| >J} H(\omega_j) e^{i \omega_j x_{m,n}}  \ket{\psi_{m,n}}. 
\end{align}
Then,
\begin{align}
 e^{-\beta w_l/2} \frac{\delta}{\sqrt{2\pi}} |   \sum_{|j| >J} H(\omega_j) e^{i \omega_j x_{m,n}}   |
  & \le   e^{-\beta w_l/2} \frac{\delta}{\sqrt{2\pi}}\sum_{|j| >J} |H(\omega_j) | \\
  & =   e^{-\beta w_l/2} \delta  \sum_{|j| >J} |F(\omega_j)  G(\omega_j)| \\
  & \le    e^{-\beta w_l/2} \frac {\delta} {2\pi} e^{1/4+ \Delta} \sum_{|j| >J} e^{-\omega_j^2/4} \\
  & \le 2  e^{-\beta w_l/2} \frac {\delta} {2\pi}  e^{1/4+ \Delta} \sum_{j >J} e^{-j(J+1) \delta^2/4} \\
  & =  e^{-\beta w_l/2} \frac {\delta} {\pi}   e^{1/4+ \Delta} \frac{e^{-\omega_{J+1}^2/4}} {1-e^{-\omega_{J+1}\delta/4}}\\
   & =  e^{-\beta w_l/2}   \frac{4(\omega_{J+1}\delta /4)} {\pi\omega_{J+1}} e^{1/4+ \Delta} \frac{e^{-\omega_{J+1}^2/4}} {1-e^{-\omega_{J+1}\delta/4}}\\
  & \le  2 e^{-\beta w_l/2} \frac{e^\Delta}{\omega_{J+1}}  \left( \frac{\omega_{J+1}\delta}4 +1\right) {e^{-\omega_{J+1}^2/4}}
  \\
  & =  2 e^{-\beta w_l/2} e^\Delta   \left( \frac{\delta}4 +\frac 1{\omega_{J+1}}\right) {e^{-\omega_{J+1}^2/4}}\\
  & \le \frac 5 2 e^{-\beta w_l/2}  {e^{\Delta-\omega_{J+1}^2/4}} \\
  \label{eq:ThirdFourierlh}
   & = \frac{5}{2}e^{-\beta w_{m,n}/2} e^{\Delta+x_{m,n}}  {e^{-\omega_{J+1}^2/4}}  \;,
\end{align}
where we used  Eqs.~\eqref{eq:FFouriertransform} and~\eqref{eq:GFouriertransform} for bounding $|H(\omega_j)|$, $4 e^{1/4}/\pi \leq 2$, $y/(1-e^{-y}) \le y+1$ for $y \ge 0$ ($y=\omega_{J+1}\delta/4$),
and assumed $\delta \le 1$ and $\omega_{J+1}\ge 1$. 
Hence, we obtain
\begin{align}\label{eq:Bound3}
\|(X_2- X)\ket{\psi_{m,n}}\| \leq \frac{5}{2} e^{-\beta w_{m,n}/2} e^{\Delta+x_{m,n}}  {e^{-\omega_{J+1}^2/4}}. 
\end{align}

%%%%%%%%%%%%%%%%%%%%%%%%%%%%%%%%%%%%%%%%%%%%%
\subsection{Choice of parameters and final bounds}

We start from the triangle inequality using the three approximations:
\begin{align}
\nonumber
    \|(e^{-\beta W/2}- X) \ket{\psi_{m,n}} \| \le & \|(e^{-\beta W/2}-X_1)\ket {\psi_{m,n}}\| +\|(X_1 - X_2) \ket {\psi_{m,n}} \| \\
   \label{eq:trianglefourier}
   &+\|(X_2 - X) \ket {\psi_{m,n}} \|  \;.
\end{align}
Let
\begin{align}
    X_{m,n}:= e^{-\beta w_l/2} \frac{\delta}{\sqrt{2\pi}} \sum_{j=-J}^J H(\omega_j) e^{i \omega_j x_{m,n}}
\end{align}
be the eigenvalue of $X$ when acting on $\ket {\psi_{m,n}}$. In addition, let $\Delta$, $\delta$, and $J$ be such that
\begin{align}
    \Delta & \ge 4 \;, \\
    2 \pi/\delta & \ge 2 \Delta + \beta (w_{\max} -w_l) \;, \\
    \omega_{J+1} &\ge 1 \;,
\end{align}
which are consistent with all  previous assumptions.
Then, according to  Eq.~\eqref{eq:Bound1}, Eq.~\eqref{eq:Bound2}, Eq.~\eqref{eq:Bound3}, and Eq.~\eqref{eq:trianglefourier}, we obtain
\begin{align}
\nonumber
    |e^{-\beta w_{m,n}/2} &- X_{m,n}| \\
    \label{eq:fourierrelativeh}
    &\le e^{-\beta w_{m,n}/2} \left( \frac {e^{-\Delta^2}} 2  + 2 e^{-2 \pi/\delta} + e e^{-(\pi/(2\delta))^2} + \frac 5 2  e^{\Delta+x_{m,n}}  e^{-\omega_{J+1}^2/4} \right) \;,
\end{align}
for all $w_{m,n}  \ge w_l$.
According to  Eq.~\eqref{eq:Bound1}, Eq.~\eqref{eq:Bound2}, Eq.~\eqref{eq:Bound3}, and Eq.~\eqref{eq:trianglefourier}, we obtain
\begin{align}
\label{eq:fourierrelativel}
    |e^{-\beta w_{m,n}/2} - X_{m,n}| \le e^{-\beta w_{m,n}/2} \left(1+ 2 e^{-2 \pi/\delta}+ e^{-\Delta^2/2}+ \frac 5 2 e^{\Delta+x_{m,n}}  \; {e^{-\omega_{J+1}^2/4}} \right)\;,
\end{align}
for all $w_{m,n}  < w_l$.

We will choose our parameters so that, when $w_{m,n} \ge w_l$, each term in Eq.~\eqref{eq:fourierrelativeh} provides an $(\epsilon/12)$-relative error and the sum is an overall $(\epsilon/3)$-relative error, as required by Eq.~\eqref{eq:hconstraint}
in
Lemma~\ref{lem:cutoffs}.
To this end, let
\begin{align}
\label{eq:Deltachoice1}
    \Delta & = \max \{4,\sqrt{\ln(6/\epsilon)}\}\;, \\
\label{eq:Deltachoice2}
    2 \pi/\delta & \ge\max\{2 \Delta + \beta (w_{\max} -w_l),2\Delta^2 \} \;, \\
\label{eq:Deltachoice3}
    \omega^2_{J+1}/4 &\ge \Delta + \beta(w_{\max}-w_l)/2 + {\ln(30/\epsilon)} \; ;
\end{align}
these are also consistent with all the previous assumptions.
Then
\begin{align}
e e^{-(\pi/(2\delta))^2} <  2 e^{-2 \pi/\delta} <  \frac  {e^{-\Delta^2}} 2 \le \epsilon/12 \;,
\end{align}
and using $w_{m,n} \le w_{\max}$ (i.e., $x_{m,n}\le \beta(w_{\max}-w_l)/2$),
\begin{align}
 \frac 5 2  e^{\Delta+x_{m,n}}  \; {e^{-\omega_{J+1}^2/4}} \le \frac 5 2 e^{-\ln(30/\epsilon)}= \epsilon/12 \;.
\end{align}
Hence, the above conditions imply, for $w_{m,n} \ge w_l$,
\begin{align}
\|(e^{-\beta W/2} - X) \ket{\psi_{m,n}}\| \le \frac \epsilon 3 e^{-\beta w_{m,n}/2}  \;,
\end{align}
which is Eq.~\eqref{eq:hconstraint}.
For $w_{m,n}<w_l$, we follow Eq.~\eqref{eq:fourierrelativel} and obtain
\begin{align}
|e^{-\beta w_{m,n}/2} - X_{m,n}| &\le e^{-\beta w_{m,n}/2}(1+\epsilon/12 + e^{-8}+ \epsilon/12) \\
& \le 2 e^{-\beta w_{m,n}/2}\;.
\end{align}
Equivalently, for $w_{m,n}<w_l$,
\begin{align}
\|(e^{-\beta W/2} - X) \ket{\psi_{m,n}}\| \le 2 e^{-\beta w_{m,n}/2}  \;,
\end{align}
which is Eq.~\eqref{eq:lconstraint}.
Therefore, any $\Delta$, $\delta$, and $\omega_{J+1}$ satisfying the conditions given in Eqs.~\eqref{eq:Deltachoice1}, ~\eqref{eq:Deltachoice2}, and ~\eqref{eq:Deltachoice3}, suffice to achieve the desired overall error.

Hence, we can choose ($\Delta^2 \ge 4 \Delta$) 
\begin{align}
  \delta &= \frac {2\pi}{   \beta(w_{\max}-w_l)+ 2  \Delta^2} 
\end{align}
to satisfy Eq.~\eqref{eq:Deltachoice2}
and any
\begin{align}
\label{eq:omegachoice}
\omega_{J+1} &\ge \frac 5 3  \sqrt{ \beta(w_{\max}-w_l)+ 2 \Delta^2} \\
& \ge \sqrt{ 2 \beta(w_{\max}-w_l)+ (11/2) \Delta^2} \\
&\ge \sqrt{ 2 \beta(w_{\max}-w_l)+\Delta^2+ 4 \Delta^2 +8} \\
    &\ge \sqrt{ 2 \beta(w_{\max}-w_l)+ \Delta^2 + 4 \ln(6/\epsilon) + 4 \ln(5)} \\
    & \ge \sqrt{ 2 \beta(w_{\max}-w_l)+ 4 \Delta + 4 \ln(30/\epsilon)}\;  
\end{align}
to satisfy Eq.~\eqref{eq:Deltachoice3}.
It follows that for any $J$ satisfying 
\begin{align}
\label{eq:J+1choice}
    J+1& \ge \left\lceil \frac 1 3 ( \beta(w_{\max}-w_l) +2 \Delta^2)^{3/2} \right \rceil \\
    &\ge \left\lceil \frac 5 3 \frac 1 {2\pi} ( \beta(w_{\max}-w_l) +2 \Delta^2)^{3/2} \right \rceil\;,
\end{align}
Eqs.~\eqref{eq:omegachoice} and~\eqref{eq:Deltachoice3} are satisfied.
Since $\Delta$ is sublogarithmic in $1/\epsilon$,
we would expect the term $\beta (w_{\max}-w_l)$ to be the dominant term in many interesting instances.
The results of the lemma follow from the definition $z:= \beta (w_{\max}-w_l) + 2\Delta^2$, so that
$\delta=2\pi/z$, and choosing
$J+1 \ge \lceil\frac 1 3 z^{3/2} \rceil$ suffices
to satisfy Eq.~\eqref{eq:Deltachoice3}.
In particular, there exists a value of $J$ (i.e., given by the right hand side of Eq.~\eqref{eq:J+1choice}) that implies $\omega_{J+1} \approx (2 \pi/3) z^{1/2}$. 
The above are possible choices for the parameters, but other choices that respect the bounds in Eqs.~\eqref{eq:Deltachoice1},~\eqref{eq:Deltachoice2}, and Eq.~\eqref{eq:Deltachoice3}, will work as well.
This proves the main claim of Lemma~\ref{lem:lemmafourier}.

Last, using again the Poisson summation formula, we obtain the following bound 
\begin{align}
\frac \delta {\sqrt{2\pi}}   \sum_{j=-J}^J |H(\omega_j)| & \le
 \frac \delta {\sqrt{2\pi}}   \sum_{j} |H(\omega_j)| \\ & \le
 \frac{\delta}{2\pi }e^{1/4+\Delta} \sum_j e^{-\omega_j^2/4} \\
 & = \frac{1}{\sqrt  \pi} e^{1/4+\Delta} \sum_k e^{-(k2 \pi/\delta)^2} \\
 & \le e^{\Delta} (1 + 2 \sum_{k=1}^\infty e^{-(k2 \pi/\delta)^2}) \\
 & \le e^{\Delta} (1 + 2 \sum_{k=1}^\infty e^{-(k2 \Delta)^2})
 \\
 & \le  e^{\Delta} (1 + 2 \sum_{k=1}^\infty e^{-k 64}) \\
 & \le 2 e^{\Delta} \\
 & \le 2 e^4 e^{ \sqrt{\ln(6/\epsilon)}}\;,
 \end{align}
which proves Eq.~\eqref{eq:L1fourier} in Lemma~\ref{lem:lemmafourier}.
This term is $1/\epsilon^{o(1)}$ and thus subpolynomial in $1/\epsilon$.
This implies that the coefficients appearing in the operator $X$ satisfy $\alpha=\sum_{j=-J}^J |\alpha_j| \le e^{-\beta w_l/2}2 e^{\Delta}$.
\qed

%%%%%%%%%%%%%%%%%%%%%%%%%%%%%%%%%%%%%%%%%%

\section{Improved approximations for $W \ge 0$}
\label{app:HSapproximation}

It is possible that for some instances the work operator in Eq.~\eqref{eq:WorkOp} is $W \ge 0$.
This occurs, for example, when $H_1 \ge 0$ and $H_0 \le 0$~\cite{CS16,CSS18}. The following result is a direct consequence of Lemma 3 in Ref.~\cite{CSS21}:
%%%%%%%%%%%%%%%%%%%%%%%%%%%%%%%%%%%%
\begin{lemma}[Hubbard-Stratonovich approximation]
\label{lem:HSapproximation}
Let $\epsilon > 0$ and assume $w_{\min} \ge 0$.
Then, for
\begin{align}
\delta & = \frac 1 {2\pi}\left(  \sqrt{\beta w_{\max}} + \sqrt{6\ln(2/\epsilon)}\right)^{-1} \;,
\end{align}
and all
\begin{align}
    J &\ge \lceil 2  \pi \left(\sqrt {\beta w_{\max}} +\sqrt{6\ln(2/\epsilon)}\right)\sqrt{6\ln(2/\epsilon)} \rceil\;,
\end{align}
we obtain 
\begin{align}
\label{eq:psdapprox}
    \| e^{-\beta W/2} - X \| \le \epsilon \;,
\end{align}
where
\begin{align}
    X:= \frac{\delta}{\sqrt{2 \pi}} \sum_{j=-J}^J e^{-\omega_j^2/2} e^{i \omega_j \sqrt{\beta W}}
\end{align}
and $\omega_j = j \delta$. In addition, the coefficients in the expansion of $X$ satisfy
\begin{align}
\label{eq:psdL1}
    \frac{\delta}{\sqrt{2 \pi}} \sum_{j=-J}^J e^{-\omega_j^2/2} \le 2 \;.
\end{align}
\end{lemma}

\begin{proof}
The result follows from the Hubbard-Stratonovich transformation~\cite{Hub59}, which itself follows from the identity ($x \ge 0$)
\begin{align}
    e^{-\beta x/2}=\frac 1 {\sqrt{2\pi}} \int_{-\infty}^\infty dy \; e^{-y^2/2} e^{iy \sqrt{\beta x}} \;.
\end{align}
We use some results of Appendix C of Ref.~\cite{CSS18}. In that work it is first shown that if ($\lambda \ge 0$)
\begin{align}
    \frac {2\pi}\delta \ge \sqrt{\beta \lambda} +\sqrt{2\ln(5/\epsilon)} \; ,
\end{align}
then
\begin{align}
   \left |e^{-\beta \lambda /2} -\frac \delta {\sqrt{2\pi}}  \sum_{j=-\infty}^\infty e^{-\omega_j^2/2}e^{i \omega_j 
   \sqrt{\beta \lambda}}\right| \le \epsilon/2 \;.
\end{align}
Next, it is shown that if
\begin{align}
    \omega_{J} \ge \sqrt{6 \ln(2/\epsilon)} \;,
\end{align}
then
\begin{align}
    \frac \delta {\sqrt{2\pi}}  \sum_{|j|>J} e^{-\omega_j^2/2} \le \epsilon/2 \;.
\end{align}
The triangle inequality gives
 \begin{align}
   \left |e^{-\beta \lambda/2 } -\frac \delta {\sqrt{2\pi}}  \sum_{j=-J}^J e^{-\omega_j^2/2}e^{i \omega_j 
   \sqrt{\beta \lambda}}\right| \le \epsilon\;.
\end{align}

We replace $\lambda$ by an eigenvalue of $W$, $w_{m,n}$. In particular, we are in the case where $w_{\min}\ge 0$ and then the largest eigenvalue is $w_{\max}$. We can then choose
\begin{align}
    \delta &= \frac 1 {2\pi} \left(\sqrt{\beta w_{\max}}+ \sqrt{6 \ln(2/\epsilon)}\right)^{-1} \\
    & \le \frac 1 {2\pi} \left(\sqrt{\beta w_{\max}}+ \sqrt{2 \ln(5/\epsilon)}\right)^{-1}
\end{align}
and
\begin{align}
    J &\ge \lceil\frac 1 \delta \sqrt{6 \ln(2/\epsilon)} \rceil\\
    \label{eq:JHSapproximation}
    & = \lceil 2 \pi \sqrt{6 \ln(2/\epsilon)}\left(\sqrt{\beta w_{\max}}+ \sqrt{6 \ln(2/\epsilon)}\right) \rceil
\end{align}
to satisfy, for all eigenvalues $w_{m,n}$ of $W$,
\begin{align}
   \left |e^{-\beta w_{m,n}/2 } -\frac \delta {\sqrt{2\pi}}  \sum_{j=-J}^J e^{-\omega_j^2/2}e^{i \omega_j 
   \sqrt{\beta w_{m,n}}}\right| \le \epsilon\;.
\end{align}
Equivalently,
\begin{align}
   \left \|e^{-\beta W/2 } -\frac \delta {\sqrt{2\pi}}  \sum_{j=-J}^J e^{-\omega_j^2/2}e^{i \omega_j 
   \sqrt{\beta W}}\right\| \le \epsilon\;,
\end{align}
which is Eq.~\eqref{eq:psdapprox}.

Last, considering the case $\lambda =0$, we already showed
 \begin{align}
   \left |1 -\frac \delta {\sqrt{2\pi}}  \sum_{j=-J}^J e^{-\omega_j^2/2} \right| \le \epsilon\;.
\end{align}
This implies
\begin{align}
  \sum_{j=-J}^J e^{-\omega_j^2/2}&  \le 1 + \epsilon \\
  & \le 2 \;,
\end{align}
which is Eq.~\eqref{eq:psdL1}.
\end{proof}

Lemma~\ref{lem:HSapproximation} is useful as long as we can implement the unitaries $e^{-i \omega_j \sqrt{\beta W}}$,
which correspond to evolving with $\sqrt{W}$
for time $\omega_j \sqrt \beta$~\cite{SqrtW}. Under this assumption, it is possible to approximate the exponential operator using a number of unitaries ($2J+1$) that is an improvement with respect to the general case given in Lemma~\ref{lem:lemmafourier}. That is, the relevant quantities depend now on $\sqrt{\beta w_{\max}}$
rather than $\beta (w_{\max}-w_l)$. In general, this is a mild improvement of the complexity,  which is still expected to  
be   dominated by the factor $e^{\beta \, \Delta \! A /2}$ in many cases, even when $W \ge 0$. 

%%%%%%%%%%%%%%%%%%%%%%%%%%%%%%%%%%%%%%%%%%

%%%%%%%%%%%%%%%%%%%%%%%%%%%%%%%%%%%%%%%%%%

\section{Proof of Lemma~\ref{lem:EigenspaceOverlap}}
\label{app:EigenspaceOverlap}

For any $a \ge 0$, we obtain
\begin{align}
\|{\Pi}^1_{\le \varepsilon_1}  \Pi^0_{> \varepsilon_0}\| & =
\| {\Pi}^1_{\le \varepsilon_1} e^{a {H_1}} e^{-a {H_1}} e^{ a H_0} e^{-a H_0} \Pi^0_{> \varepsilon_0}\| \\
& \leq \| {\Pi}^1_{\le \varepsilon_1} e^{a {H_1}} \| \ \| e^{-a {H_1}} e^{a H_0} \| \ \| e^{-a H_0} \Pi^0_{> \varepsilon_0} \| \\
& \le e^{-a(\varepsilon_0 - \varepsilon_1)} \ \| e^{-a {H_1}} e^{ a H_0} \|.
\end{align}
Let $F(a):= e^{-a{H_1}} e^{a H_0}$. This operator satisfies
\begin{align}
\partial_a F(a)&= - F(a) e^{-aH_1} V e^{aH_0} \\
& =-F(a) V(a) \;,
\end{align}
where $V(a):=e^{-aH_0} V e^{aH_0}$.
Hence
\begin{align}
F(a)= \mathcal{T} e^{-\int^{a}_0 da' \; V(a')}\;,
\end{align}
where $\mathcal{T}$ is the time-ordering operator.
Then,
\begin{align}
\|F(a)\|&=
\|\mathcal{T} e^{\int^{a}_0 da' \; V(a')}\| \\
& \leq 1+ \int_0^a da' \|V(a')\| + \int_0^a da'\int_0^{a'} da'' \|V(a')\| \|V(a'')\| + \ldots
 \\
& \leq e^{a \max_{a' \in [0, a]} \| V(a') \|} \;.
\end{align}

% \begin{align}
% \|F(a)\|&=
% \|\mathcal{T} e^{\int^{a}_0 da' \; V(a')}\| \\
% & \leq e^{\|\int^{a}_0 da' \; V(a')\|} \\
% & \leq e^{a \max_{a' \in [0, a]} \| V(a') \|}.
% \end{align}

Lemma~3.1  in Ref.~\cite{arad2016connecting} implies $\|V(a')\| \leq \frac{Mv}{(1- a'hgk)^r}$ if $0 \leq 1-a'hgk < 1$, where $r:= R/(hgk)$ and $R:= \max_{Y \in \Lambda} \sum_{X \in \Lambda
: [h_{0,X}, v_Y]\neq 0} \|h_{0,X}\|$.
(In their notation, $s \rightarrow a'$ and $g \rightarrow g h$.)
Then,
\begin{align}
    \max_{a' \in [0, a]} \| V(a') \| \le \frac{Mv}{(1- a hgk)^r}
\end{align}
under the assumption $0 \leq ahgk < 1$. 
Note that $R \le hgk$, since each term $v_X$ in $V$
involves at most $k$ spins, and each spin interacts with at most $g$ other spins. Then, $r \le 1$.
As a result, we obtain
\begin{align}
\|\Pi^1_{\leq \varepsilon_1} \Pi^0_{> \varepsilon_0}\| \leq e^{-a(\varepsilon_0 - \varepsilon_1)} e^{a \frac{M v}{(1-ahgk)}} \;.
\end{align}
Choosing $a= 1/(2hgk)$, we obtain the desired result:
\begin{align}
\|\Pi^1_{\leq \varepsilon_1} \Pi^0_{> \varepsilon_0}\| \leq e^{-\frac{\varepsilon_0 - \varepsilon_1-2Mv}{2hgk}} \;.
\end{align}
\qed

Note that we can obtain a better bound by minimization over all possible values of $a \in \mathbb{R}^+$ subject to the assumption. Nevertheless, the previous result
suffices for Thm.~\ref{thm:w_llocal}.

%%%%%%%%%%%%%%%%%%%%%%%%%%%%%%%%%%%%%%%%%%%%%%%%%%%%%%%%%%%%%%%%%

\section{Phases for quantum signal processing}
\label{app:QSPImplementation}

We seek to implement the operator $X$ given in Eq.~\eqref{eq:Fourierapproximationoperator}  using QSP.
This requires finding the phases associated with the ancilla qubit rotations, as discussed in Sec.~\ref{sec:QSP}.
Here we use the MATLAB package QSPPACK (https://github.com/qsppack/QSPPACK) to explicitly calculate these phases using optimization methods. 
The MATLAB package is based on the recent work by Yulong Dong et al. \cite{dong2021}. It requires the input polynomials 
to be real and given in a Chebyshev basis, as a linear combination of Chebyshev polynomials $\cT_j(y)$ of the first kind. 
The input polynomials are also required to have
definite parity (even or odd) in $y$. With a little extra work, we note that the package can be modified to provide the phases when the input function is given as a linear combination of $\sqrt{1-y^2} \cR_j(y)$, where the $\cR_j$'s are the Chebyshev polynomials of the second kind; more details follow.

The general quantum circuit $V_\Phi$ used for QSP implementation in the package is in Fig.~\ref{fig:QSPPACKcircuit}.
The ancilla-qubit gates determined by $\Phi=\{\phi_0, \phi_1,\cdots,\phi_d\}$ are rotations around the Z axis by corresponding angles.
H is the Hadamard gate and X is the Pauli gate 
that performs the map ${\rm X}\ket 0 \rightarrow \ket 1$
and ${\rm X}\ket 1 \rightarrow \ket 0$. The operations $\cV$ and $\cV^\dagger$ are unitary, and these are controlled on the state $\ket 1$ of the ancilla. 
Let $\theta/2$ be an eigenphase of $\cV$, $\ket{\psi_{\theta/2}}$ the corresponding eigenstate, and define $y=\cos(\theta/2)$. 
Then, in the two-dimensional subspace spanned by $\{\ket 0 \ket{\psi_{\theta/2}} ,\ket 1 \ket{\psi_{\theta/2}}\}$, the quantum circuit in Fig.~\ref{fig:QSPPACKcircuit} implements the $SU(2)$
operation
\begin{align}
    V_\Phi (y)&=   e^{i \phi_0 {\rm Z}}\cW (y) e^{i \phi_1 {\rm Z}} \cdots  \cW(y) e^{i \phi_d {\rm Z}} \;,
    \end{align}
where
\begin{align}
    \cW(y)&:= e^{\pm i \arccos(y) {\rm X}}\\
    &= \begin{pmatrix}
    y &  \pm i \sqrt{1-y^2} \\
     \pm i \sqrt{1-y^2} & y 
    \end{pmatrix} \; 
\end{align}
is in $SU(2)$ and, for integer $k$,
\begin{align}
   \cW^k(y) &= e^{\pm i k \arccos(y) {\rm X}}\\
    & = \begin{pmatrix}
   \cos(k \theta/2) & i\sin(k \theta/2) \\
    i\sin(k \theta/2) &  \cos(k \theta/2)
    \end{pmatrix} \\
    \label{eq:Wk(y)}
   & = \begin{pmatrix}
    \cT_k(y) & \pm i\sqrt{1-y^2} \cR_{k-1}(y) \\
    \pm i\sqrt{1-y^2} \cR_{k-1}(y) & \cT_k(y) 
    \end{pmatrix}  \;.
\end{align}
(We adopt the convention
$\cT_{-k}(y)=\cT_k(y)$ and $\cR_{-k-1}(y)=-\cR_{k-1}(y)$
if $k \ge 1$.) The sign in the off-diagonal entry depends 
on the value of $\theta$; it is $+1$ if $\theta/2 \mod 2 \pi \in [0,\pi)$
and $-1$ if $\theta/2 \mod 2 \pi \in [\pi,2\pi)$.
Using Eq.~\eqref{eq:Wk(y)} and the fact that ${\rm Z} \cW^k(y) {\rm Z}=\cW^{-k}(y)$, ${\rm Z}^2=\one$, we can prove
\begin{align}
    V_\Phi(y)&=(\cos \phi_0 \one + i \sin \phi_0 {\rm Z}) \cW(y) (\cos \phi_1 \one + i \sin \phi_1 {\rm Z})\cW(y) \ldots \\
    & =\sum_{k=-d,-d+2,\ldots,d} D_k \cW^k(y) \;,
\end{align}
where $D_k$ are two-dimensional diagonal matrices that can be written as $D_k=\Re(\beta_k) \one + i \Im(\beta_k) {\rm Z}$, for some coefficients $\beta _k \in \mathbb C$. Then,
\begin{align}
\label{eq:VPhi_app}
  V_\Phi (y) =    \begin{pmatrix}
    P(y) & \pm i\sqrt{1-y^2} Q(y) \\
  \pm  i\sqrt{1-y^2} Q^*(y) & P^*(y) 
    \end{pmatrix} \;,
\end{align}
where $P(y) = \sum_{k=-d,-d+2,\ldots} \beta_k \cT_k(y)$ and $Q(y)=\sum_{k=-d,-d+2,\ldots} \beta_k \cR_{k-1}(y)$.
Setting (even) $d=2J$, these are
\begin{align}
\label{eq:PQSP}
    P(y) &= \beta_0 +\sum_{j=1}^J (\beta_{2j}+\beta_{-2j}) \cT_{2j}(y) \;,\\
\label{eq:QQSP}
    Q(y) & =\sum_{j=1}^J (\beta_{2j}-\beta_{-2j})\cR_{2j-1}(y) \;.
\end{align}
Note that $P(y)=P(-y)$ and $Q(y)=-Q(-y)$ have definite parity (even and odd degrees, respectively).

The input to QSPPACK is a polynomial $p_1(y)\in \mathbb R[y]$  of definite parity, such that $|p_1(y)|\leq 1 \ \ \forall y \in [-1,1]$, given as a linear combination of Chebyshev polynomials of the first kind, e.g., a linear combination of $\cT_{2j}(y)$ if it is an even function. The package finds a set  $\Phi_1$ of $2J+1$ phases that produces an $SU(2)$ unitary $V_{\Phi_1}(y)$, where the first entry
satisfies $\Re({\bra{0} V_{\Phi_1} (y) \ket{0}}) = p_1(y)$. Finding $\Phi_1$ is a straightforward application of the function QSP\_solver in the package, which uses the limited-memory BFGS optimization algorithm. A modification of the solver allows us to address a polynomial $q_2(y)\in \mathbb R[y]$ of definite parity, such that $|q_2(y)|\leq 1 \ \ \forall y \in [-1,1]$,  given as a linear combination of Chebyshev polynomials of the second kind, e.g., a linear combination of $\cR_{2j-1}(y)$ if it is an odd function. In this case, the package finds a set $\Phi_2$ of $2J+1$ phases  that produces an $SU(2)$ unitary $V_{\Phi_2}(y)$, where the second entry
satisfies  $\Im({\bra{0} V_{\Phi_2} (y) \ket{1}}) = \pm \sqrt{1-y^2}q_2(y)$. To that end,  we modify the solver such that objective function to be minimized is  $\frac{1}{2}\left(\Im({\bra{0} V_{\Phi} (y) \ket{1}}) - (\pm \sqrt{1-y^2})Q(y)\right)^2$. 
We can also combine the unitaries such that $\Re({\bra{0} V_{\Phi_1}(y) + V_{\Phi_2}(y) e^{-i\frac{\pi}{2} {\rm X}} \ket{0}}) = p_1(y) + (\pm \sqrt{1-y^2}) q_2(y)$, allowing us to 
address more general functions using a simple LCU.

To implement the operator $X$ using QSPPACK, we need to find a presentation that is compatible
with Eqs.~\eqref{eq:PQSP} and~\eqref{eq:QQSP}.
For $\theta \in \mathbb R$, let $f(e^{i\theta}):=\frac 1 \alpha \sum_{j=-J}^J \alpha_j e^{ij\theta}$, where $\alpha=\sum_{j=-J}^J |\alpha_j|$,
\begin{align}
    \alpha_j& = e^{-\beta w_l/2} \frac{\delta} {\sqrt{2\pi}}H(\omega_j)e^{-i \omega_j \beta w_l/2} \; , \\
H(\omega) &= \frac{e^{\Delta + 1/2 - (\omega^2 +1)/4}}{\sqrt{2 \pi} (1+\omega^2)} (1-i\omega) e^{i\omega (\Delta + 1/2)},
\end{align}
and $\omega = j \delta$; see Eqs.~\eqref{eq:Fouriercoefficients} and~\eqref{eq:L1fourier}. In particular, if $\theta$ is an eigenphase of $U$, where $U$ is the unitary in Eq.~\eqref{eq:Uaction}, and if we replace $e^{i \theta} \rightarrow U$, we obtain $f(U)=X/\alpha$. The property $\alpha_j=\alpha_{-j}^*$ implies $f(e^{i\theta}) \in \mathbb R$. Then, $f(e^{i\theta})=\frac 1 \alpha \sum_{j=-J}^J \Re(\alpha_j)\cos(j\theta) -\Im(\alpha_j)\sin(j\theta)$ and, if $y=\cos(\theta/2)$, we obtain $f(e^{i \theta})=p_1(y)+(\pm \sqrt{1-y^2}) q_2(y)$, where
\begin{align}
\label{eq:pQSP}
    p_1(y)&:=\frac 1 \alpha \left(\alpha_0 +2 \sum_{j=1}^J \Re(\alpha_j) \cT_{2j}(y) \right)\;, \\
\label{eq:qQSP}
    q_2(y)&:=-\frac 2 \alpha \sum_{j=1}^J \Im(\alpha_j)  \cR_{2j-1}(y)\;.
\end{align}
These polynomials satisfy the properties required by the (modified) package. That is, for $y \in [-1,1]$, $p_1(y) \in \mathbb R$ and $q_2(y) \in \mathbb R$ are presented as linear combinations of Chebyshev polynomials of the first and second kinds, respectively. Also, $p_1(y)=p_1(-y)$ is an even function, $q_2(y)=-q_2(-y)$ is an odd function, and $|p_1(y)|\le 1$, $|q_2(y)|\le 1$. 
We can then produce each of these polynomials in the entries of $SU(2)$ matrices using QSPPACK,
which outputs the sets $\Phi_1$ and $\Phi_2$ discussed earlier, of $2J+1$ phases each. In particular, in the two-dimensional subspace,
\begin{align}
\label{eq:QSPLCU}
 \frac 1 2 \left(   V_{\Phi_1}(y) + V^\dagger_{\Phi_1}(y)+ V_{\Phi_2}(y) e^{-i \frac \pi 2 {\rm X}} + e^{i \frac \pi 2 {\rm X}} V^\dagger_{\Phi_2}(y)  \right) =
 \begin{pmatrix} p_1(y) +(\pm \sqrt{1-y^2}) q_2(y) & . \cr . & .\end{pmatrix} \;,
\end{align}
and the first entry is $f(e^{i\theta)}$.
The $SU(2)$ unitaries $V_{\Phi_1}(y)$ and $V_{\Phi_2}(y)$ correspond
to unitaries $V_{\Phi_1}$ and $V_{\Phi_2}$ in the larger Hilbert space, respectively, 
that can be implemented using a circuit like Fig.~\ref{fig:QSPPACKcircuit}. The unitaries
$V^\dagger_{\Phi_1}$ and $V^\dagger_{\Phi_2}$
can also be implemented using a similar circuit noting
that, in general, $V^\dagger_\Phi=e^{-i \phi_d ({\rm Z}\otimes \one)}\cW\cdots  e^{-i \phi_1 ({\rm Z}\otimes \one)}\cW e^{-i \phi_0 ({\rm Z}\otimes \one)}$.

The circuit in Fig.~\ref{fig:QSPPACKcircuit} uses controlled-$\cV$
and controlled-$\cV^\dagger$ operations. In the analyzed two-dimensional subspaces, $\theta/2$ is the eigenphase of $\cV$ and $\theta$ is the eigenphase of $U$, implying $\cV=U^{\frac 1 2}$. Our quantum algorithm assumes access to $U$ but not necessarily to $U^{\frac 1 2}$.  Nevertheless, a simple compilation of the circuit in Fig.~\ref{fig:QSPPACKcircuit}
allows it to be implemented using controlled-$U$ and controlled-$U^{\dagger}$ operations instead:

\begin{align} \label{eq:qsphalf}
 \Qcircuit @C=1em @R=0.7em {
& \gate{\phi_d} & \gate{{\rm H}}  & \ctrlo{1}                      & \ctrl{1}                         & \gate{{\rm H}} &  \gate{\phi_{d-1}} & \gate{{\rm H}}  & \ctrlo{1}                      & \ctrl{1}                         & \gate{{\rm H}} & \qw  & \gate{\phi_{d-2}} & \qw   & \cdots  \\
& \qw           & \qw       & \multigate{2}{U^{\frac{1}{2}}} & \multigate{2}{U^{-\frac{1}{2}}}  & \qw      & \qw            & \qw       & \multigate{2}{U^{\frac{1}{2}}} & \multigate{2}{U^{-\frac{1}{2}}}  & \qw      & \qw  &\qw            & \qw   & \cdots  \\
& \qw           & \qw       & \ghost{U^{\frac{1}{2}}}       & \ghost{U^{-\frac{1}{2}}}          &  \qw     & \qw            & \qw       & \ghost{U^{\frac{1}{2}}}       & \ghost{U^{-\frac{1}{2}}}          &  \qw     & \qw  &\qw            & \qw   & \cdots  \\
& \qw           & \qw       & \ghost{U^{\frac{1}{2}}}       & \ghost{U^{-\frac{1}{2}}}          & \qw      & \qw             & \qw       & \ghost{U^{\frac{1}{2}}}       & \ghost{U^{-\frac{1}{2}}}          & \qw     & \qw  &\qw            & \qw   & \cdots \\
}
\end{align}

\begin{align}
 \Qcircuit @C=1em @R=0.7em {
&& \gate{\phi_d} & \gate{{\rm H}}  & \ctrlo{1}                      & \ctrl{1}                         & \gate{{\rm H}} &  \gate{\phi_{d-1}} & \gate{{\rm H}}  & \ctrlo{1}                      & \ctrl{1}                         & \gate{{\rm H}} & \qw  & \gate{\phi_{d-2}} & \qw   & \cdots  \\
\mbox{=}&& \qw           & \qw       & \multigate{2}{U^{\frac{1}{2}}} & \multigate{2}{U^{-\frac{1}{2}}}  & \multigate{2}{U^{\frac{1}{2}}} & \qw & \multigate{2}{U^{-\frac{1}{2}}}                  & \multigate{2}{U^{\frac{1}{2}}} & \multigate{2}{U^{-\frac{1}{2}}}  & \qw      & \qw  &\qw            & \qw   & \cdots \\
&& \qw           & \qw       & \ghost{U^{\frac{1}{2}}}       & \ghost{U^{-\frac{1}{2}}}          & \ghost{U^{\frac{1}{2}}}       & \qw  & \ghost{U^{-\frac{1}{2}}}                  & \ghost{U^{\frac{1}{2}}}       & \ghost{U^{-\frac{1}{2}}}          &  \qw     & \qw  &\qw            & \qw   & \cdots  \\
&& \qw           & \qw       & \ghost{U^{\frac{1}{2}}}       & \ghost{U^{-\frac{1}{2}}}          & \ghost{U^{\frac{1}{2}}}      & \qw & \ghost{U^{-\frac{1}{2}}}                     & \ghost{U^{\frac{1}{2}}}       & \ghost{U^{-\frac{1}{2}}}          & \qw     & \qw  &\qw            & \qw   & \cdots \\
}
\end{align}

\begin{align}\label{eq:FinalQSPcircuit}
 \Qcircuit @C=1em @R=0.7em {
&&\qw & \multigate{3}{V_{\Phi}}&\qw& \qw &&& \gate{\phi_d} & \gate{{\rm H}}  & \ctrlo{1}                                             & \gate{{\rm H}} &  \gate{\phi_{d-1}} & \gate{{\rm H}}                      & \ctrl{1}                         & \gate{{\rm H}} & \qw   &\gate{\phi_{d-2}} & \qw   & \cdots  \\
\mbox{=}&&\qw & \ghost{V_{\Phi}}&\qw& \qw &\mbox{=} && \qw           & \qw       & \multigate{2}{U}   & \qw & \qw            & \qw       & \multigate{2}{U^{\dagger}}  & \qw      & \qw  &\qw            & \qw   & \cdots  \\
&&\qw & \ghost{V_{\Phi}}&\qw& \qw &  && \qw           & \qw       & \ghost{U}                 & \qw       & \qw           & \qw            & \ghost{U^{\dagger}}          &  \qw     & \qw  &\qw            & \qw   & \cdots  \\
&&\qw & \ghost{V_{\Phi}}&\qw& \qw &&& \qw           & \qw       & \ghost{U}                & \qw       &\qw            & \qw            & \ghost{U^{\dagger}}          & \qw     & \qw  &\qw            & \qw   & \cdots 
}
\end{align}

Equation~\eqref{eq:QSPLCU} implies
\begin{align}
\label{eq:QSPLCU2}
 \frac 1 4 \left(   V_{\Phi_1} + V^\dagger_{\Phi_1} + V_{\Phi_2} e^{-i \frac \pi 2 {\rm X}} + e^{i \frac \pi 2 {\rm X}} V^\dagger_{\Phi_2}  \right) =
 \begin{pmatrix} \frac X{2\alpha} & . \cr . & .\end{pmatrix} \;,
\end{align}
and the first entry contains the desired operator $X$.
This
is a linear combination of four unitaries, where the coefficients add up to 1, and can be directly implemented using the LCU framework described in Sec.~\ref{sec:LCU}. It requires two additional ancilla qubits, bringing the total ancilla count to three in this case.
The number of controlled-$U$ and controlled-$U^{\dagger}$
is bounded by $4 \times (2J)=\cO(J)$, and the additional
number of two-qubit gates is also $\cO(J)$, as stated in Thm.~\ref{thm:main}. Since using this method we implement $X/(2\alpha)$ rather than $X/\alpha$ (see Eq.~\eqref{eq:S'action} and Sec.~\ref{sec:algorithmcomplexity}), the factor 1/2 will bring an additional but constant overhead, where the average number of amplitude amplification rounds is doubled.

In Fig.~\ref{fig:QSPPACK_results}  
we show the sets of phases $\Phi_1=\{\phi_{1,0},\phi_{1,1},\ldots,\phi_{1,2J}\}$ and $\Phi_2=\{\phi_{2,0},\phi_{2,1},\ldots,\phi_{2,2J}\}$ obtained by running the QSPPACK MATLAB package, such that the unitaries $V_{\Phi_1}$ and $V_{\Phi_2}$
satisfy Eq.~\eqref{eq:QSPLCU2}. Note that, to better display the results,  we have not plotted the first and last phases of these sets, which are given in the caption.
The package was modified to have an extra input called ``kind", which has the value $1$ for functions given as linear combinations of $\cT_j(y)$ and $2$ for functions given as linear combinations of $\cR_j(y)$. The objective function and the gradient function is defined in the package depending on the value of ``kind". The relevant parameters chosen in this simulation are $\beta w_{\max} = 50$, $\beta w_l = -1$, $\Delta = 4$. We then use Lemma~\ref{lem:lemmafourier} to compute the rest of the parameters. We obtain $J = 252$ and $\delta = 0.0757$ for Eqs.~\eqref{eq:pQSP} and~\eqref{eq:qQSP}, and these also
determine $X$. An initial guess for the phases is needed by QSPPACK, and we set $\Phi^{\rm ini}_1 = \{\frac{\pi}{4}, 0, 0, \cdots, \frac{\pi}{4} \}$ and $\Phi^{\rm ini}_2 = \{0, 0, 0, \cdots, -\frac{\pi}{2}\}$. The package uses limited memory BFGS algorithm (a quasi Newton method) to find the phases. The resulting error in generating the function $p_1(y)+(\pm \sqrt{1-y^2})q_2(y)$ in the entry of the $SU(2)$ matrix in Eq.~\eqref{eq:QSPLCU} is less than machine precision for all $y\in[-1,1]$.

\begin{figure}[htb]
%\centering
\begin{center}

     \includegraphics[width = 16cm]{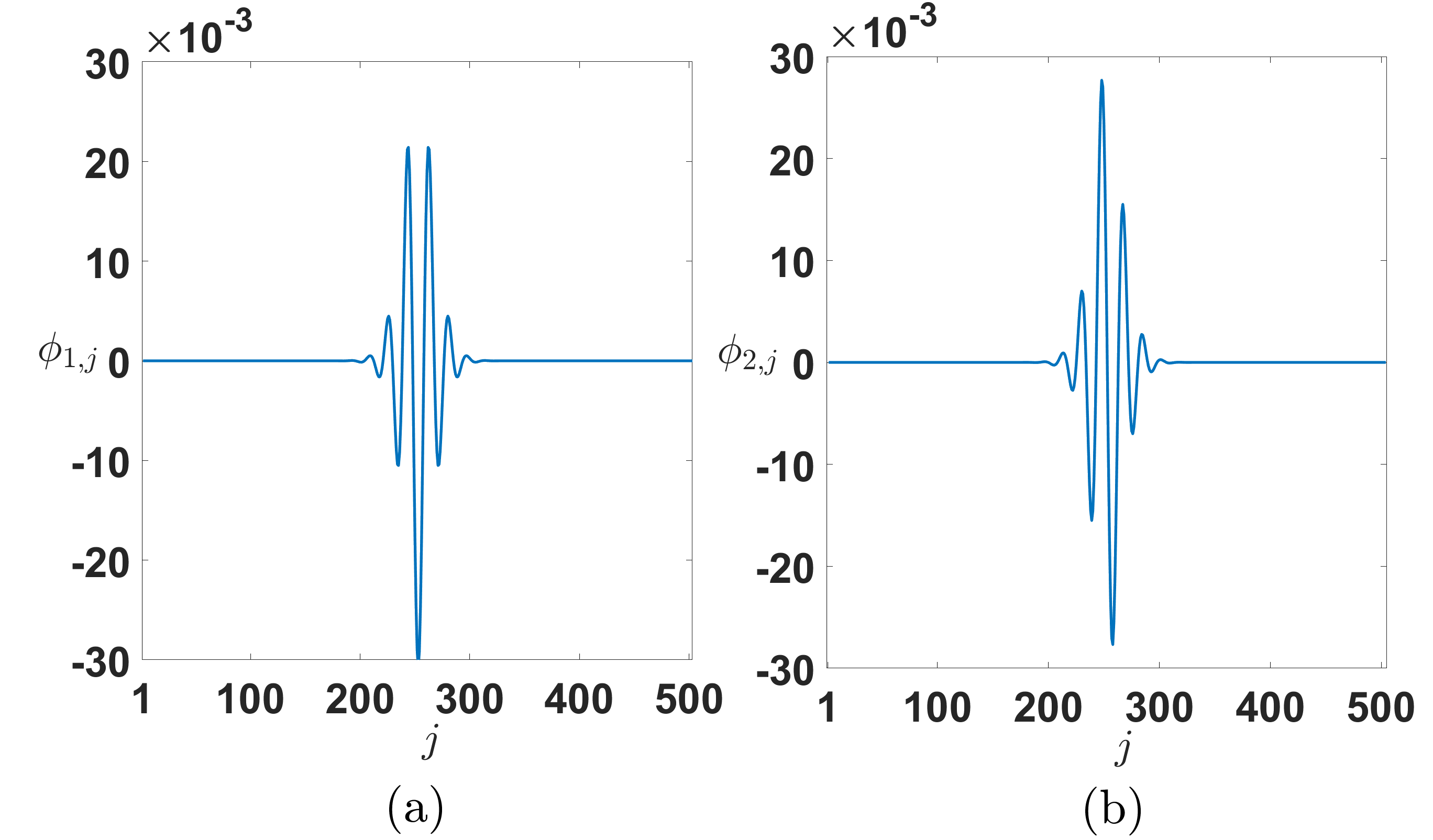}
    \caption{Phases to implement $X$ using QSP. The simulation parameters are $\Delta =4$, $\beta w_{\max} = 50$, $\beta w_l = -1$, $J=252$, and $\delta=0.0757$. (a) Phases $\Phi_1=\{\phi_{1,0},\phi_{1,1},\ldots,\phi_{1,504}\}$ obtained using QSPPACK for generating $p_1(y)$, where $\phi_{1,0} = \phi_{1,504} = \frac{\pi}{4}$ are not plotted. (b) Phases $\Phi_2=\{\phi_{2,0},\phi_{2,1},\ldots,\phi_{2,504}\}$ obtained using QSPPACK for generating $q_2(y)$, where $\phi_{2,0}= 0$ and $\phi_{2,504} = -\frac{\pi}{2}$ are not plotted.
    }
    \label{fig:QSPPACK_results}
     \end{center}
\end{figure}

\end{document}